\documentclass[12pt]{article}

\usepackage{url,mathrsfs }
\usepackage{tikz}
\usepackage{ifthen,amssymb}
\usepackage{cite}
\usepackage[cmex10]{amsmath} 
\usepackage{amsfonts,mathtools}
\usepackage{enumerate}
\usepackage{amsthm}

\newcommand{\U}{\mathcal{U}}

\def\textiid{i.i.d.\@\xspace}
\newcommand\iid{\ifmmode\text{ i.i.d. } \else \textiid \fi}

\newcommand{\beqs}{\begin{equation*}}
\newcommand{\eeqs}{\end{equation*}}
\newcommand{\beq}{\begin{equation}}
\newcommand{\eeq}{\end{equation}}
\usepackage{xcolor}
\usepackage{appendix}
\usepackage[a4paper, total={6in, 8in},margin=2cm]{geometry}
\usepackage{hyperref}
\hypersetup{
    colorlinks=true,
    linkcolor=blue,
    filecolor=magenta,      
    urlcolor=cyan,
}

\theoremstyle{plain}
\newtheorem{theorem}{Theorem}
\newtheorem{claim}{Claim}
\newtheorem{lemma}{Lemma}
\newtheorem{proposition}{Proposition}
\newtheorem{corollary}{Corollary}

\newtheorem{definition}{Definition}
\newtheorem{remark}{Remark}
\newtheorem{example}{Example}
\newtheorem{note}{Note}

\newcommand{\A}{{\mathsf{A}}}
\newcommand{\B}{{\mathsf{B}}}
\newcommand{\mf}{{\mathsf{f}}}

\newcommand{\D}{{\mathscr{D}}}
\newcommand{\RR}{{\mathscr{R}}}

\newcommand{\tw}{{\tilde{\textbf{w}}}}

\newcommand{\tW}{{\tilde{\textbf{W}}}}

\begin{document}

\title{ 
$\mf$-divergences and their applications in lossy compression and bounding generalization error  
}

\author{Saeed Masiha\qquad Amin Gohari \qquad Mohammad Hossein Yassaee}

\allowdisplaybreaks

\maketitle
\begin{abstract}

In this paper, we provide three applications for $\mf$-divergences: (i) we introduce Sanov's upper bound on the tail probability of the sum of independent random variables based on super-modular $\mf$-divergence and show that our generalized Sanov's bound strictly improves over ordinary one, (ii) we consider the lossy compression problem which studies  the set of achievable rates for a given distortion and code length. We extend the rate-distortion function using mutual $\mf$-information and provide new and strictly better bounds on achievable rates in the finite blocklength regime using super-modular $\mf$-divergences, and (iii) we provide a connection between the generalization error of algorithms with bounded input/output mutual $\mf$-information and a generalized rate-distortion problem. This connection allows us to bound the generalization error of learning algorithms using lower bounds on the $\mf$-rate-distortion function. Our bound is based on a new lower bound on the rate-distortion function that (for some examples) strictly improves over previously best-known bounds.\footnote{Some preliminary ideas at the early stages of this work were presented in \cite{masiha2021learning}.}



\end{abstract}
\section{Introduction}

The generalized \emph{relative entropy} $D_\mf(\cdot\|\cdot)$, also known as the $\mf$-divergence,  was introduced by Ali-Silvey \cite{ali1966general} and Csiszar \cite{csiszar2011information,csiszar1967information} as a measure of dissimilarity between the two distributions defined on the same sample space. $\mf$-divergences have found various applications in information theory, statistics and machine learning among other fields. Some applications and properties of $\mathsf f$-divergence are given in \cite{sason2016f,ziv1973functionals,tridenski2015ziv,raginsky2016strong,makur2020comparison}.

In this paper, we give the applications of $\mf$-divergences in Sanov's theorem \cite{csiszar2006simple} on the tail probability of the sum of random variables, lossy source coding, and upper bounding generalization error. 

Before describing our contributions, we introduce a class of $\mf$-divergences that satisfies a supermodularity property. More specifically, given an arbitrary distribution $p_{X_1,\cdots, X_n}$ and an arbitrary product distribution $q_{X_1,\cdots, X_n}=\prod_{i=1}^n q_{X_i}$, we consider the   function $g:2^{[n]}\mapsto \mathbb{R}$ defined as 
\[g(A)=D_\mf(p_{X_{A}}\|q_{X_A}), \qquad \forall A\subset [n],\]
where $X_A=(X_i: i\in A)$. We say that $\mf$-divergences satisfy a supermodularity property if $g(A)$ is a super-modular set function. We provide an explicit class of convex functions $\mf$ that lead to supermodular $\mf$-divergences.

Next, we list our contributions as follows:

\begin{itemize}
    \item  We generalize the upper bound part of Sanov's theorem \cite{csiszar2006simple} on the tail probability of the sum of random variables using supermodular $\mf$-divergences. We show that our extension of Sanov's bound based on some choice of supermodular $\mf$-divergence strictly improves over the ordinary Sanov's upper bound in the non-asymptotic regime.
    \item $\mf$-divergences can be used to define a mutual $\mf$-information $I_\mf(X;Y)$ between two random variables $X$ and $Y$. Mutual $\mf$-information is a measure of dependence between two random variables and generalizes Shannon's mutual information. We show that supermodular $\mf$-divergences imply that the resulting mutual $\mf$-information satisfies the following property: for any random variables $A,B,C$ we have
    \begin{align}I_\mf(A,B;C)\geq I_\mf(A;C)+I_\mf(B;C)\label{ABproperty}
    \end{align}
     as long as $A$ and $B$ are independent. We call the above property the \textit{$AB$-property}. The $AB$-property is stronger than the data processing inequality $I_\mf(A,B;C)\geq I_\mf(A;C)$. We continue by giving an explicit application of the $AB$-property. We note that in the converse proof of lossy source coding problem (also known as the rate-distortion problem), one can rewrite the proof in such a way that instead of using the chain-rule property of mutual information, we use the $AB$-property to complete the proof. This insight shows that we can mimic the standard proof and replace Shannon's mutual information with $\mf$-information in all places in the proof. We define a notion of $\mf$-rate-distortion function and show that it can provide new and strictly better bounds on the achievable rates for a given distortion in the finite blocklength regime.\footnote{Note that in the asymptotic regime, as the blocklength tends to infinity, a full characterization of the achievable rates is known and no improvement is possible in that case.} We also extend a previous work by Ziv and Zakai in \cite{ziv1973functionals}.
     \item 
     It is known that under certain assumptions, the generalization error of a learning algorithm can be bounded from above in terms of the mutual information between the input and output of the algorithm \cite{russo2019much,xu2017information}. Similar bounds on generalization error are obtained in \cite{bu2020tightening,lopez2018generalization,wang2019information,hellstrom2020generalization,aminian2020jensen, aminian2022tighter,aminian2021characterizing,aminian2021information,esposito2020robust,esposito2019generalization,jiao2017dependence,asadi2018chaining} for various learning algorithms and using other measures of dependence. In this paper, we are interested in the generalization error for the class of algorithms with bounded
     input/output mutual $\mf$-information. For this class of algorithms, we give a novel connection between the $\mf$-rate-distortion function and the generalization error. Moreover, this leads to a new upper bound on the generalization error using the $\mf$-rate-distortion function that strictly improves over the previous bounds in  \cite{xu2017information,bu2020tightening}. Finally, $\mf$-rate-distortion function defined using super-modular $\mf$-divergences enjoys certain properties that facilitate its evaluation when the number of data samples is large.
\end{itemize}

 As stated above, we provide a novel connection between generalization error and the rate-distortion theory.  This connection allows us to strictly improve over the bound in
\cite{xu2017information}. In order to show that our rate-distortion bound strictly improves over the bound by Xu and Raginsky \cite{xu2017information}, we provide a new lower bound on the rate-distortion function (and on $\mf$-rate distortion in general). Not only this bound allows us to relate our bound to the bound by Xu and Raginsky, but it also strictly improves over the previously known bounds on the rate-distortion function in some cases \cite{xu2017information,bu2020tightening}.

We also give some variants of the rate-distortion bound in Appendices \ref{rate_of_consistency} and \ref{further-ideas} that can tighten the ordinary rate-distortion bounds on the generalization error. 

This paper is organized as follows: in Section \ref{f-divergence-sec}, we define supermodular $\mf$-divergence and discuss its application in Sanov's theorem. In Section \ref{sec: mutual-f-information}, we explore various definitions of mutual $\mf$-information and $\mf$-entropy and their properties. Lossy compression with mutual $\mf$-information is discussed in Section \ref{sec: lossy compression}. Lower bounds on ($\mf$-)rate-distortion function and its connection with generalization error of learning with bounded input/output mutual $\mf$-information are discussed in Sections \ref{sec:lower bound on rate-dist} and \ref{sec: gen} respectively.

\subsection{Notation and preliminaries}
 
 Random variables are shown in capital letters, whereas their realizations are shown in lowercase letters. We show sets with calligraphic font. For a random variable $X$ generated from a distribution $\mu$, we use $\mathbb{E}_{X}$ or $\mathbb{E}_{\mu}$ to denote the expectation taken over $X$ with distribution $\mu$. $P_{Z}$ means the distribution over $Z$. We use $\log_{2}(x)$ and $\ln(x)$ to denote the logarithm in base two and in base $e$ respectively. We use the notation $\mathcal{O}$ to hide constants.

\section{$\mf$-divergence}
\label{f-divergence-sec}

The generalized \emph{relative entropy} of Ali-Silvey \cite{ali1966general} and Csiszar \cite{csiszar2011information}\cite{csiszar1967information} (also called the ``$\mf$-divergence") is defined as follows:
\begin{definition}
 Let $\mathsf f:\mathbb{R}_{+}\to \mathbb{R}$ be a convex function with $\mathsf f(1)=0$. Let $P$ and
$Q$ be two probability distributions on a measurable space $(\mathcal{X},\sigma(\mathcal{X}))$. If $P\ll Q$ \footnote{We say $\mu\ll\nu$, i.e., $\mu$ is absolutely continuous with respect to $\nu$ if $\nu(A)=0$ for some $A\in\mathcal{X}$, then $\mu(A)=0$.} then the $\mathsf f$-divergence is defined as
\[
D_{\mathsf f}(P\|Q)=\mathbb{E}_{Q}\mathsf f\left(\frac{dP}{dQ}(X)\right)
\]
where $\frac{dP}{dQ}$ is a Radon-Nikodym derivative and $\mf(0)\triangleq \mf(0^{+})$.

Define the conditional $\mathsf f$-divergence as follows: \[
D_{\mathsf f}(P_{Y|X}\|Q_{Y|X}|P_{X})=\mathbb{E}_{X\sim P_{X}}[D_{\mathsf f}(P_{Y|X}\|Q_{Y|X})],
\]

\end{definition}

\begin{theorem}(Properties of $\mf$-divergences)\cite[Chap. 6]{polyanskiy2014lecture}
\begin{itemize}
    \item Non-negativity: $D_{\mathsf f}(P\|Q)\ge 0$ and equality holds if and only if $P=Q$.

\item Joint convexity: $(P,Q)\to D_{\mathsf f}(P\|Q)$ is a jointly convex function. In particular, this property implies that conditioning does not decrease $\mathsf f$-divergence: Let $P_{X}\xrightarrow[]{P_{Y|X}} P_{Y}$ and $P_{X}\xrightarrow[]{Q_{Y|X}} Q_{Y}$. Then
\[
D_{\mathsf f}(P_{Y}\|Q_{Y})\le D_{\mathsf f}(P_{Y|X}\|Q_{Y|X}|P_{X}).
\]

\item Data processing inequality: Let $P_{X}\xrightarrow[]{P_{Y|X}} P_{Y}$ and $Q_{X}\xrightarrow[]{P_{Y|X}} Q_{Y}$. Then
\[
D_{\mathsf f}(P_{Y}\|Q_{Y})\le D_{\mathsf f}(P_{X}\|Q_{X}).
\]
\end{itemize}
\label{theorem_properties_D_f}
\end{theorem}

For the special case of $\mathsf{f}(t)=t\ln(t)$,  $D_{\mathsf f}(\cdot\|\cdot)$ reduces to the KL divergence. For the special case of $\mathsf f_{\alpha}(t)=t^{\alpha}-1$ for $\alpha\geq 1$, the
$\mathsf f_{\alpha}$-divergence can be written as 
\[
D_{\mf_{\alpha}}(\mu\|\nu):=\int \left(\frac{d\mu}{d\nu}\right)^{\alpha}d\nu(x)-1
\]
Renyi's divergence of order $\alpha$ can be derived by $D_\alpha(\mu\|\nu):=\frac{1}{\alpha-1}\ln(1+D_{\mf_{\alpha}}(\mu\|\nu))$. In particular, we have
$D_2(\mu\|\nu)=\ln\left(1+\chi^{2}(\mu\|\nu)\right)$ where  $\chi^{2}$-divergence is defined as \begin{align}\label{def:chi-squared}\chi^{2}(\mu\|\nu)=\mathbb{E}_{\nu}\left(\frac{d\mu}{d\nu}-1\right)^{2}.
\end{align}

\subsection{Supermodular $\mf$-divergences}\label{class_conv_func_F}
\begin{definition}
Given a convex function $\mf$ with $\mf(1)=0$, we say that the $D_\mf$ is super-modular if for any joint distribution $p_{X_1,X_2,X_3}$ and any product distribution $q_{X_1}q_{X_2}q_{X_3}$ on arbitrary alphabets $\mathcal{X}_1, \mathcal{X}_2, \mathcal{X}_3$ we have
\begin{align}D_\mf(p_{X_1X_2X_3}\|q_{X_1}q_{X_2}q_{X_3})+D_\mf(p_{X_3}\|q_{X_3})\geq D_\mf(p_{X_1, X_3}\|q_{X_1}q_{X_3})+D_\mf(p_{X_2, X_3}\|q_{X_2}q_{X_3}).\label{def-sub-mod}
\end{align}
\end{definition}

\begin{remark}
The reason for using the term \emph{super-modularity} is that \eqref{def-sub-mod} implies that the   function $g:2^{[n]}\mapsto \mathbb{R}$ defined as 
\[g(A)=D_\mf(p_{X_{A}}\|q_{X_A}), \qquad \forall A\subset [n]\]
is a super-modular set function for any product distribution $q_{X^n}=\prod_{i=1}^n q_{X_i}$ and any arbitrary joint distribution $p_{X^n}$. Here, we have $X_A=(X_i: i\in A)$.
\end{remark}
\begin{corollary}
Let $P_{X_1, \cdots, X_n, W}$ and $Q_{X_1, \cdots, X_n, W}$ be two distributions on  $n+1$ random variables. Assume that $Q_{X_1, \cdots, X_n, W}=Q_{X_1}Q_{X_2}\cdots Q_{X_n}Q_W$. Then, by induction on $n$, \eqref{def-sub-mod} implies that
\begin{align}D_\mf(P_{X^n,W}\|Q_{X^n,W})-D_\mf(P_{W}\|Q_{W})\geq \sum_{i=1}^n\left[D_\mf(P_{X_i,W}\|Q_{X_i,W})-D_\mf(P_{W}\|Q_{W})\right].
\end{align}
\label{corr:def-sub-mod22}
\end{corollary}

To proceed, we need the following definition: 

\begin{definition}\label{def:class-C}
Define $\mathcal F$ be the class of convex functions $\mathsf f(t)$ on $[0,\infty]$ that 
$\mathsf f(1)=0$, $\mf^{''}$ is strictly positive, and $1/\mathsf f''$ is concave. 
\end{definition}
The above class of convex functions is important because it makes $\Phi$-entropy  subadditive \cite[Theorem 14.1]{boucheron2013concentration}. Alternative equivalent definitions for the class $\mathcal F$ are given in \cite{boucheron2013concentration}.

\begin{remark}
This class of convex functions appears in other contexts \cite{beigi2018phi}\cite[Appendix B]{mojahedian2019correlation}. For instance, it makes the definition of $\Phi$-entropic measures of correlation possible \cite{beigi2018phi}.
\end{remark}

\begin{proposition}\label{prop_super-modular_D_f}
$D_\mf$ is super-modular for any function $\mf\in\mathcal{F}$. 
\end{proposition}
The proof is given in Appendix \ref{proof_prop_super_modular_D_f}. The above proposition provides a sufficient condition for super-modularity. We do not know if $\mf\in\mathcal{F}$ is also a necessary condition for $D_\mf$ being super-modular.
 
The functions
$\mf(t)=t\ln t$ and $\mf(t)=\frac{1}{\alpha-1}(t^{\alpha}-1)$ for $\alpha\in(1,2]$ are two examples in class $\mathcal{F}$. The following lemma (proven in Appendix \ref{section_class_F_proofs}) shows that $t\ln t$ and $t^2$ are ``extreme" members of  $\mathcal{F}$ in the sense that the growth rate of  any function $\mf(t)\in\mathcal{F}$ (as $t$ converges to infinity) is no smaller than $t\ln(t)$ and no larger than $t^2$.

\begin{lemma}\label{newlmma3}
Take an arbitrary  $\mf(t)\in\mathcal{F}$. Then, $t\mapsto t\mf''(t)$ is a non-decreasing function and
\begin{align}
\lim_{t\to \infty}\frac{\mathsf f(t)}{t\ln t}&=\lim_{t\to \infty}t\mf''(t)>0.\label{const1tlogt}\end{align}
Next, $t\mapsto \mf''(t)$ is a non-negative and non-increasing function and
\begin{align}
\lim_{t\to \infty}\frac{\frac{1}{2}t^2}{\mathsf f(t)}&=\lim_{t\to \infty}\frac{1}{\mf''(t)}>0.\label{const2t2}
\end{align}
\end{lemma}

\subsubsection{Application}
As an application for supermodular $\mf$-divergence, we extend Sanov's upper bound \cite{csiszar2006simple}. Upper bound part of Sanov's theorem states that for $n$ i.i.d. random variables distributed according to $P_{X}$, we have the following upper bound:
\[
\mathbb{P}\left[\frac{1}{n}\sum_{i=1}^{n}X_{i}\ge \delta\right]\le \exp\left(-n\inf_{Q_{X}:\,\mathbb{E}_{Q}[X]\ge \delta}D(Q_{X}\|P_{X})\right).
\]
Our extension of Sanov's upper bound is two-fold: firstly, we generalize it to the sum of dependent random variables and secondly we replace the KL-divergence with the $\mf$-divergence.

\begin{theorem}\label{thm2nsa}
Take two arbitrary sets $\mathcal{X}$ and $\mathcal{W}$. Let $S=(X_1, \cdots, X_n)$ be a sequence of $n$ i.i.d.\ random variables according to $P_X$ and $W\sim P_{W}$ be independent of $S$. Let $\ell(\cdot,\cdot):\mathcal{X}\times \mathcal{W}\to \mathbb{R}$ be a measurable function. Let 
\begin{align}
\gamma\triangleq\mathbb{P}\left[\frac{1}{n}\sum_{i=1}^{n}\ell(X_{i},W)\ge \delta\right].\label{eqnGammaDef}
\end{align}
Then  for any $\mf\in\mathcal{F}$,
\begin{align}\label{gen_sanov_f_div}
    \gamma\mf\left(\frac{1}{\gamma}\right)+(1-\gamma)\mf(0)\ge n\cdot\inf_{\substack{Q:\\\mathbb{E}_{Q_{X,W}}[\ell(X,W)]\ge \delta}}\left\{D_{\mf}(Q_{X,W}\|P_{X}Q_{W})-\frac{n-1}{n}D_{\mf}(Q_{W}\|P_{W})\right\}.
\end{align}
\end{theorem}
The proof is given in Appendix \ref{proof_of_thm2nsa}. 

\begin{remark} Note that random variables $\ell(X_i,W), i=1,2,\cdots, n$ are dependent in general, since they all depend on $W$. If we set $\mf(x)=x\ln(x)$ and $W$ to be constant,  the above theorem reduces to Sanov's theorem for the sum of independent random variables. If $W$ is not a constant and $\mf(x)=x\ln(x)$, the above bound reduces to 
\begin{align}
    \ln\left(\frac{1}{\gamma}\right)\ge n\cdot\inf_{\substack{Q:\\\mathbb{E}_{Q_{X,W}}[\ell(X,W)]\ge \delta}}\left\{D(Q_{X,W}\|P_{X}Q_{W})-\frac{n-1}{n}D(Q_{W}\|P_{W})\right\}.\label{eqnfSv23}
\end{align}
On the other hand, Chernoff's bound is
\begin{align}
    \mathbb{P}\left[\frac{1}{n}\sum_{i=1}^{n}\ell(X_{i},W)\ge \delta\right]&=\mathbb{E}_{W}\left[\mathbb{P}\left[\frac{1}{n}\sum_{i=1}^{n}\ell(X_{i},W)\ge \delta\mid W\right]\right]\nonumber\\
    &\le e^{-n\alpha\delta}\mathbb{E}_{W}\left[\left(\mathbb{E}[e^{\alpha \ell(X,W)}|W]\right)^{n}\right].
\end{align}
Hence, Chernoff's bound yields that $\mathbb{P}[\frac{1}{n}\sum_{i=1}^{n}\ell(X_{i},W)\ge \delta]\le e^{-nE_{CH}}$ where
\begin{align}
E_{CH}=\sup_{\alpha>0}\left\{\alpha\delta-\frac{1}{n}\ln\mathbb{E}_{W}\left[\left(\mathbb{E}[e^{\alpha \ell(X,W)}|W]\right)^{n}\right]\right\}.\label{E_CH}
\end{align}
We show in Appendix \ref{proof_remark 004} that the bound in \eqref{eqnfSv23} matches Chernoff's bound when $X$ and $W$ are discrete random variables on the finite alphabet sets.
\label{remark 004}
\end{remark}

\begin{example} We next show the benefit of using a function $\mf(x)\neq x\ln(x)$ in a hypothesis testing problem. Consider the following hypothesis testing problem for some $\eta>0$: 
\[
\begin{cases}
  X\sim Bern(1/2+\eta)  & \emph{under hypothesis } H_0, \\
  X\sim Bern(1/2-\eta) &\emph{under hypothesis }  H_1.
\end{cases}
\]
Note that under $H_0$, we have
\begin{align*}
\mathbb{E}_{H_0}\left[\sum_i X_i\right]&=\frac{n}{2}+n\eta, \qquad \mathsf{Var}_{H_0}\left[\sum_i X_i\right]=n\left(\frac14-\eta^2\right).
\end{align*}
Under $H_0$ we expect $\sum_i X_i$ to be around its mean plus a constant times its standard deviation with high probability. On the other hand, 
\begin{align*}
\mathbb{E}_{H_1}\left[\sum_i X_i\right]&=\frac{n}{2}-n\eta.
\end{align*}
Our goal is to choose $\eta$ such that the probability of error does not vanish as $n$ tends to infinity. If we choose $\eta=\dfrac{c}{2\sqrt{n}}$ for some constant $c$, then under both $H_0$ and $H_1$, $\sum_i X_i$ will be $n/2\pm\mathcal O(\sqrt{n})$ and we will have overlap between the ranges of $\sum_i X_i$  under the two hypotheses. Hence we expect that the error probability does not vanish even if $n$ tends to infinity. Note that for $\eta=\mathcal{O}(\frac{1}{\sqrt{n}})$, the distribution on $X$ under $H_0$ and $H_1$ are very close to each other. From a practical perspective, this hypothesis testing problem is relevant when we want to detect a very weak background signal (as in the low probability of detection communication, satellite communication, or underwater communication).

To study our hypothesis testing problem, we consider uniform prior on $H_0$ and $H_1$. Therefore, the optimal decision rule is to compare $\frac{1}{n}\sum_{i}X_i$ with $1/2$ and the error probability can be written as
\[\gamma=\mathbb{P}_{H_1}\left[\frac{1}{n}\sum_{i=1}^{n}X_{i}\ge \frac{1}{2}\right].\]
Next, we show that Sanov's bound is insufficient in this scenario.

Let
$\mf(x)=x^{2}-x$. Then,
\[D_\mf(Q\|P)=\chi^2(Q\|P).
\]
Let $W$ be a constant random variable and $\ell(x,w)=x$. Then, from \eqref{gen_sanov_f_div} we get

\begin{align}
    \frac{1}{\gamma}-1&\ge n\cdot\inf_{\substack{Q:\\\mathbb{E}_{Q}[X]\ge 1/2}}\chi^2(Q_{X}\|P_{X})=n\left(\frac{1}{4(1/2-\eta)}+\frac{1}{4(1/2+\eta)}-1\right)
    \\&=\frac{4n\eta^2}{1-4\eta^2}
\end{align}
or
\[\gamma\leq U_{x^{2}-x}:=\frac{1}{1+\frac{4n\eta^2}{1-4\eta^2}}.\]
From Sanov's theorem we get
\[\gamma\leq U_{x\ln x}:=\exp\left(-n\left(\frac12\ln(\frac{1}{1-4\eta^2})\right)\right)=(1-4\eta^2)^{n/2}.\]
We will show that for a finite number of samples $n\ge1$ and $\eta\in(0,\frac{1}{2\sqrt{n}})$, we have $U_{x^{2}-x}< U_{x\ln x}$. Assume that $\eta=\frac{c}{2\sqrt{n}}$ for some $c<1$. Then, we get $ U_{x\ln x}=(1-\frac{c^{2}}{n})^{n/2}\ge 1-\frac{c^{2}}{2}$ for $n\ge 2$ since $(1-x)^n\geq 1-nx$ for $x\in[0,1]$ and $n\geq 1$. On the other hand, we have
\[
U_{x^{2}-x}=\frac{1}{1+\frac{c^{2}}{1-\frac{c^{2}}{n}}}=\frac{n-c^{2}}{(1+c^{2})n-c^{2}}=\frac{1}{1+c^{2}}-\frac{c^{2}-\frac{c^{2}}{1+c^{2}}}{(1+c^{2})n-c^{2}}\le \frac{1}{1+c^{2}}.
\]
As $c\in(0,1)$, we get $\frac{1}{1+c^{2}}<1-c^{2}/2$ and the claim is established.

\end{example}

\section{$\mf$-information}\label{sec: mutual-f-information}

The following two proposals for defining a mutual $\mathsf{f}$-information in terms of the $\mathsf f$-divergence are known:
The first 
 is (see 
\cite{ziv1973functionals}
~\cite[Eq. 3.10.1]{cohen1998comparisons})\footnote{A further generalization is given in \cite{zakai1975generalization}.}
\begin{align}
    I^{CKZ}_{\mathsf f}(A;B):= D_{\mathsf f}(P_{AB}\|P_A\times P_B),\label{IfCKZ}
\end{align}
and has been studied in the literature (e.g. see
\cite{merhav2011data,tridenski2015ziv,hsu2018generalizing,cohen1993relative}). 
Another definition is given in  \cite[Eq. 79]{polyanskiy2010arimoto}:
\begin{align}I^{PV}_{\mathsf f}(A;B):= \min_{Q_B}D_{\mathsf f}(P_{AB}\|P_A\times Q_B),\label{IfPV}
\end{align}
where the minimum is over all $Q_B$ such that $P_B\ll Q_B$.
Note that $I_{\mathsf f}^{CKZ}(A;B)=I_{\mathsf f}^{CKZ}(B;A)$ is symmetric but $I_{\mathsf f}^{PV}(A;B)$ is not symmetric in general. Moreover, when $D_{\mathsf f}(\cdot\|\cdot)$ is the KL divergence, both of these $\mathsf f$-informations reduce to Shannon's mutual information. 
\begin{example}
Let $\mf(x)=x^{\alpha}-1$ for $\alpha\in[1,2]$. Then, for random variables $A$ and $B$, we have
\begin{align}
    I^{PV}_{\mf}(A;B)&=\min_{Q_{B}}e^{(\alpha-1)D_{\alpha}(P_{AB}\|P_{A}Q_{B})}-1=e^{(\alpha-1)\min_{Q_{B}}D_{\alpha}(P_{AB}\|P_{A}Q_{B})}-1\\
&=e^{(\alpha-1)I_{\alpha}(A;B)}-1=\left(\mathbb{E}_{P_{B}}\left[\mathbb{E}_{P_{A}}\left(\frac{dP_{B|A}}{dP_{B}}\right)^{\alpha}\right]^{1/\alpha}\right)^{\alpha}-1
\end{align}
where $I_{\alpha}$ is the $\alpha$-mutual information  according to Sibson's proposal  \cite{verdu2015alpha}.
\end{example}

There is yet another definition for mutual $\mathsf f$-information in \cite[Appendix B]{mojahedian2019correlation} as follows:
\begin{align}
    I^{\text{MBGYA}}_{\mf}(A;B)=\min_{Q_B}\left\{D_{\mathsf f}(P_{AB}\|P_A\times Q_B)-D_{\mathsf f}(P_{B}\|Q_{B})\right\}\label{rmkne1}
\end{align}
where the minimum is over all $Q_B$ such that $P_B\ll Q_B$. It is clear from the definitions above that 
\[I_\mf^{CKZ}(A;B)\geq I_\mf^{PV}(A;B)\geq I_\mf^{MGBYA}(A;B).\]
\begin{example}
Let $\mathsf{f}_\alpha(t)=t^{\alpha}-1$ for $1\le \alpha\le 2$. Then it is shown in \cite[Theorem 33]{mojahedian2019correlation} that
\begin{align*}
I^{\text{MBGYA}}_{\mathsf f_\alpha}(A;B)=\frac{1}{\alpha-1}\left(\mathbb{E}_{P_{B}}\left[\mathbb{E}_{P_{A}}\left[\left(\frac{dP_{B|A}}{dP_{B}}\right)^{\alpha}-1\right]\right]^{\frac{1}{\alpha}}\right)^{\alpha}.
\end{align*}
When $\alpha=2$, interestingly this definition  coincides with \cite[Definition 1]{DBLP:conf/isit/IssaG18} (derived independently) in the context of exploration bias.
\end{example}

\subsection{Properties of mutual f-information}\label{f-information-appendix}
 \begin{definition}
 Let $\rho(X;Y)$ be a mapping that assigns a non-negative real number to arbitrary random variables $X$ and $Y$. We say that $\rho$ is a measure of dependence if it satisfies the faithfulness and data processing properties defined as follows:  
  we say that  $\rho(\cdot;\cdot)$  satisfies the \emph{faithfulness property} if $\rho(X;Y)=0$ if and only if $X$ and $Y$ are independent. We say that $\rho(\cdot,\cdot)$ satisfies the \emph{data processing property} if $\rho(X;Y)\ge \rho(W;Z)$ whenever $W-X-Y-Z$ forms a Markov chain.\label{defMC}
 \end{definition}
 
\begin{theorem}\label{dpi_f_information}\cite{polyanskiy2010arimoto, mojahedian2019correlation}
For any convex function $\mf$, $I^{PV}_{\mf}$, $I^{CKZ}_{\mf}$ are measures of dependence (see Definition \ref{defMC}). For any function $\mathsf f\in\mathcal{F}$ as defined in Definition \ref{def:class-C}, the mutual $\mathsf f$-information $I^{\text{MGBYA}}_{\mf}$ is a measure of dependence. 
\end{theorem}

\begin{theorem}\label{new properties_I_f}
Assume that $dP_{Y|X}(y,x):=\frac{dP_{YX}}{dP_{X}}(y,x)$ exists and we call it $p(y|x)$.
\begin{enumerate}[(i).]
\item For every convex function $\mathsf f$, $I^{PV}_{\mf}(X;Y)$ and $I^{\text{MBGYA}}_{\mf}(X;Y)$ are concave in $P_{X}$ when $p(y|x)$ is fixed.
Equivalently,
    $I^{PV}_{\mathsf f}(X;Y)\geq I^{PV}_{\mathsf f}(X;Y|Q)$ and $I^{MBGYA}_{\mathsf f}(X;Y)\geq I^{MBGYA}_{\mathsf f}(X;Y|Q)$ for any $p_{X,Q}p_{Y|X}$.

    \item  For any convex function $\mf$, $I^{CKZ}_{\mathsf f}(X;Y)$ and $I^{PV}_{\mathsf f}(X;Y)$ are convex functions of $p(y|x)$ when $P_{X}$ is fixed. Equivalently, $I^{CKZ}_{\mathsf f}(X;Y)\leq I_{\mathsf f}^{CKZ}(X;Y|Q)$ and
    $I^{PV}_{\mathsf f}(X;Y)\leq I_{\mathsf f}^{PV}(X;Y|Q)$ for any $p_{Q}p_{X}p_{Y|X,Q}$. The same convexity statement holds for $I^{MGBYA}_{\mathsf f}(X;Y)$ with the further assumption that $\mf\in\mathcal{F}$.
    \item Assume that $\mathsf f\in \mathcal{F}$. Let $X^{n}=(X_1, \cdots, X_n)$ be a sequence of $n$ independent random variables. Assume that $U$ and $X^{n}$ are arbitrarily distributed. Then,
\begin{align}
    I^{CKZ}_{\mathsf f}(X^n;U)&\geq \sum_i I^{CKZ}_{\mathsf f}(X_i;U),\\
    I^{MBGYA}_{\mathsf f}(X^n;U)&\geq \sum_i I^{MBGYA}_{\mathsf f}(X_i;U).
\end{align}
\item 
For $\mf\in\mathcal{F}$ and every $p(x,y)$, we have
\[I(X;Y)\cdot \left[\lim_{t\to \infty}t\mathsf f''(t)\right]\geq
I_\mf^{CKZ}(X;Y)\geq I_\mf^{PV}(X;Y)\geq I_\mf^{\text{MGBYA}}(X;Y).\]
\end{enumerate}

\end{theorem}
The proof of Theorem \ref{new properties_I_f} is given in Appendix \ref{proof_new properties_I_f}.


\begin{definition}\label{def_f_entropy} Assume that random variable $X$ takes value in a discrete set $\mathcal{X}$. The $\mathsf f$-entropy of $X$ is defined as follows:
 \begin{align}
    H^{CKZ}_{\mathsf f}(X)\triangleq I^{CKZ}_{\mathsf f}(X;X)=\mathsf f(0)\left(1-\sum_{x\in\mathcal{X}}P^{2}_{X}(x)\right)+\sum_{x\in\mathcal{X}}P^{2}_{X}(x)\mathsf f\left(\frac{1}{P_{X}(x)}\right)
\end{align}
\begin{align}
    H^{PV}_{\mathsf f}(X)\triangleq I^{PV}_{\mathsf f}(X;X)=\min_{Q_{X}}\left\lbrace \mathsf f(0)\left(1-\sum_{x\in\mathcal{X}}P_{X}(x)Q_{X}(x)\right)+\sum_{x\in\mathcal{X}}P_{X}(x)Q_{X}(x)\mathsf f\left(\frac{1}{Q_{X}(x)}\right)\right\rbrace
\end{align}
\begin{align}
    &H^{\text{MBGYA}}_{\mathsf f}(X)\triangleq I^{\text{MBGYA}}_{\mf}(X;X)\nonumber\\
    &=\min_{Q_{X}}\left\lbrace \mathsf f(0)\left(1-\sum_{x\in\mathcal{X}}P_{X}(x)Q_{X}(x)\right)+\sum_{x\in\mathcal{X}}P_{X}(x)Q_{X}(x)\mathsf f\left(\frac{1}{Q_{X}(x)}\right)-\sum_{x\in\mathcal{X}}Q_{X}(x)\mf\left(\frac{P_{X}(x)}{Q_{X}(x)}\right)\right\rbrace
\end{align}
\end{definition}

We will prove some properties of $\mathsf f$-entropy. We use $H_{\mathsf f}$ if the statement is true for both $H^{PV}_{\mathsf f}(X)$ and $H^{CKZ}_{\mathsf f}(X)$.

\begin{theorem}\label{property_f_entropy}[Properties of $\mathsf f$-entropy]
\begin{enumerate}[(i).]
    \item For every convex function $\mathsf f$ defined on $I=[0,+\infty)$, $H^{\text{MBGYA}}_{\mathsf f}(X)$ is a concave function of $p(x)$.  With the extra condition $\mathsf f^{'''}(x)\le 0$ for $x>0$, $H^{CKZ}_{\mathsf f}(X)$ is a concave function of $p(x)$. Moreover, $\mathsf{f}\in \mathcal{F}$ implies that $\mathsf f^{'''}(x)\le 0$ for $x>0$.
    \item The function $p(x)\mapsto H^{\text{MBGYA}}_{\mathsf f}(X)$ is maximized at the uniform distribution for every convex function $\mf(x)$. A similar statement holds $H^{CKZ}_{\mathsf f}(X)$ if $\mathsf f^{'''}(x)\le 0$ for $x>0$.
    \item
    Let $Y$ take value in a finite set $\mathcal{Y}$, then for convex $\mathsf f:\mathbb{R}_{+}\to \mathbb{R}$, 
\begin{align}\label{I_f_H_f}
    I^{CKZ}_{\mathsf f}(X;Y)&\le H^{CKZ}_{\mathsf f}(Y),\\
    I^{PV}_{\mathsf f}(X;Y)&\le H^{PV}_{\mathsf f}(Y),\\
    I^{\text{MBGYA}}_{\mathsf f}(X;Y)&\le H^{\text{MBGYA}}_{\mathsf f}(Y).
\end{align}
\end{enumerate}
\end{theorem}
The proof of Theorem \ref{property_f_entropy} is given in Appendix \ref{proof_property_f_entropy}.

\section{Lossy compression with mutual $\mf$-information}\label{sec: lossy compression}

\tikzset{every picture/.style={line width=0.75pt}} 
\begin{figure}
\centering
\begin{tikzpicture}[x=0.75pt,y=0.75pt,yscale=-1.3,xscale=1.3]

\draw   (238,28) .. controls (238,23.58) and (241.58,20) .. (246,20) -- (300,20) .. controls (304.42,20) and (308,23.58) .. (308,28) -- (308,52) .. controls (308,56.42) and (304.42,60) .. (300,60) -- (246,60) .. controls (241.58,60) and (238,56.42) .. (238,52) -- cycle ;
\draw   (362,28) .. controls (362,23.58) and (365.58,20) .. (370,20) -- (424,20) .. controls (428.42,20) and (432,23.58) .. (432,28) -- (432,52) .. controls (432,56.42) and (428.42,60) .. (424,60) -- (370,60) .. controls (365.58,60) and (362,56.42) .. (362,52) -- cycle ;
\draw    (182,40.2) -- (236,40.2) ;
\draw [shift={(238,40.2)}, rotate = 180] [color={rgb, 255:red, 0; green, 0; blue, 0 }  ][line width=0.75]    (10.93,-3.29) .. controls (6.95,-1.4) and (3.31,-0.3) .. (0,0) .. controls (3.31,0.3) and (6.95,1.4) .. (10.93,3.29)   ;
\draw    (307,40.2) -- (359,40.2) ;
\draw [shift={(361,40.2)}, rotate = 180] [color={rgb, 255:red, 0; green, 0; blue, 0 }  ][line width=0.75]    (10.93,-3.29) .. controls (6.95,-1.4) and (3.31,-0.3) .. (0,0) .. controls (3.31,0.3) and (6.95,1.4) .. (10.93,3.29)   ;
\draw    (432,39.2) -- (484,39.2) ;
\draw [shift={(486,39.2)}, rotate = 180] [color={rgb, 255:red, 0; green, 0; blue, 0 }  ][line width=0.75]    (10.93,-3.29) .. controls (6.95,-1.4) and (3.31,-0.3) .. (0,0) .. controls (3.31,0.3) and (6.95,1.4) .. (10.93,3.29)   ;

\draw (159,32) node [anchor=north west][inner sep=0.75pt]    {$X^{n}$};
\draw (489,28.4) node [anchor=north west][inner sep=0.75pt]    {$\hat{X}^{n}$};
\draw (249,34) node [anchor=north west][inner sep=0.75pt]   [align=center] {Encoder};
\draw (373,34) node [anchor=north west][inner sep=0.75pt]   [align=center] {Decoder};
\draw (323,21.4) node [anchor=north west][inner sep=0.75pt]    {$M$};

\end{tikzpicture}

\caption{The setup for the lossy compression problem} \label{fig:M1}
\end{figure}
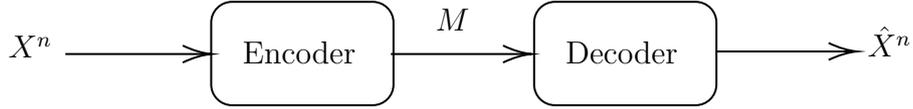

Consider a memoryless source $X\sim \mu_X$ and a distortion function
$d(x,\hat x)$ for $x\in\mathcal{X}$ and $\hat x\in\hat{\mathcal{X}}$ where $\mathcal{X}$ and $\hat{\mathcal{X}}$ are the source and reconstruction alphabet sets. While the literature commonly assumes that the reconstruction alphabet set is identical with the source alphabet set $\hat{\mathcal{X}}=\mathcal{X}$,  we do not make this assumption here. Moreover, the distortion function $d(x,\hat{x})$ is generally assumed to be non-negative. Here, unlike the source-coding literature, we allow $d(x,\hat{x})$ to take negative values. The  same standard proofs (of the rate-distortion theory) go through when $\hat{\mathcal{X}}\neq \mathcal{X}$ or when $d(x,\hat{x})$ becomes negative.

An $(n,R,D)$ \emph{lossy source code} consists of an encoder $\mathcal{E}:\mathcal{X}^n\mapsto \{1,2,\cdots, 2^{nR}\}$ and a decoder $\mathcal{D}:\{1,2,\cdots, 2^{nR}\}\mapsto\hat{\mathcal{X}}^n$ such that the reconstruction sequence \[\hat{X}^n=\mathcal{D}(\mathcal{E}(X^n))\] satisfies the expected requirement
\[\frac1n\sum_{i=1}^{n}\mathbb{E}[d(X_i, \hat X_i)]\leq D.
\]
A rate-distortion pair $(R,D)$ is said to be achievable if for every $\epsilon>0$, one can find an $(n,R,D+\epsilon)$ lossy source code for some blocklength $n$. Given a rate $R$, let $\D(R)$ be the minimum $D$ such that the rate-distortion pair
$(R,D)$ is achievable. Similarly, given a distortion $D$, let $\RR(D)$ be the minimum $R$ such that the rate-distortion pair
$(R,D)$ is achievable.
The following characterization of the rate-distortion function $\D(R)$ is known \cite[Theorem 3.5]{el2011network}:\footnote{Some technical conditions are needed for \eqref{eqnRDnew2} and \eqref{eqnRDcurve1} when $d(x,\hat{x})$ is unbounded. For instance, a sufficient condition is existence of $\hat{x}_{0}\in \hat{\mathcal{X}}$ such that $\mathbb{E}[d(X,\hat{x}_{0})]<\infty$ \cite[Theorems 7.2.4
\& 7.2.5]{berger2003rate}.}
\begin{align}
\RR(D)&=
\inf_{P_{\hat X|X}:~
\mathbb{E}[d(X,\hat{X})]\leq D
}I(X;\hat{X}).\label{eqnRDnew2}\\
\D(R)&=\inf_{P_{\hat X|X}:~I(X;\hat X)\le R}\mathbb{E}[d(X,\hat X)],\label{eqnRDcurve1}
 \end{align}

One can also formally define a variant of the rate-distortion function by replacing Shannon's mutual information with other measures of correlation. For instance, using mutual $\mathsf f$-information defined in \eqref{IfCKZ} or \eqref{IfPV}, we define:
\begin{align}\RR^{CKZ}_{\mathsf f}(D)&=
\inf_{P_{\hat X|X}:~
\mathbb{E}[d(X,\hat{X})]\leq D
}I^{CKZ}_{\mathsf f}(X;\hat{X}),\label{eqnRDfCKZ}
 \\
 \D^{CKZ}_{{\mathsf{f}}}(R)&=\inf_{P_{\hat X|X}:~I^{CKZ}_{\mathsf f}(X;\hat{X})\le R}\mathbb{E}[d(X,\hat{X})],\label{eqnDfv1a}
  \end{align}
  and similarly,
   \begin{align}
 \RR^{\text{MBGYA}}_{\mathsf f}(D)&=
\inf_{P_{\hat X|X}:~
\mathbb{E}[d(X,\hat{X})]\leq D
}I^{\text{MBGYA}}_{\mathsf f}(X;\hat{X}),\label{eqnRDfPV}
\\
 \D^{\text{MBGYA}}_{{\mathsf{f}}}(R)&=\inf_{P_{\hat X|X}:~I^{\text{MBGYA}}_{\mathsf f}(X;\hat{X})\le R}\mathbb{E}[d(X,\hat{X})].\label{eqnDfv2a}
 \end{align}
 We call \eqref{eqnRDfCKZ} and \eqref{eqnRDfPV}  $\mf$-rate-distortion functions. Note that for the special case of $\mf(t)=t\ln(t)$, the $\mf$-rate-distortion functions in \eqref{eqnRDfCKZ} and \eqref{eqnRDfPV}  reduce to the ordinary rate-distortion functions in 
\eqref{eqnRDnew2}. 
\begin{remark}
 From Theorem \ref{new properties_I_f} Part (iv), we obtain the following inequalities for any $\mf\in\mathcal{F}$:
\begin{align}
   &\RR(D)\cdot\left[\lim_{t\to \infty}t\mathsf f''(t)\right] \ge \RR^{CKZ}_{\mathsf f}(D)\ge \RR^{PV}_{\mathsf f}(D)\ge \RR^{\text{MGBYA}}_{\mathsf f}(D),\\
   &\D(R)\cdot\left[\lim_{t\to \infty}t\mathsf f''(t)\right]\ge \D^{CKZ}_{{\mathsf{f}}}(R)\ge \D^{PV}_{{\mathsf{f}}}(R)\ge \D^{\text{MGBYA}}_{{\mathsf{f}}}(R).
\end{align}
\end{remark}
The ordinary rate-distortion function in \eqref{eqnRDnew2} with Shannon's mutual information is a meaningful quantity with an operational interpretation as the solution to the lossy compression problem. What about  $\RR^{CKZ}_{\mathsf f}(D)$ or $\RR^{PV}_{\mathsf f}(D)$?\footnote{
We are only interested in the $\mf$-rate-distortion function in the context of the lossy source coding problem. The $\mf$-rate-distortion function is also of interest in other applications such as privacy or security; interested readers can refer to \cite[Section IV.A, Section IV.C]{du2017principal} for such applications of $\mf$-rate-distortion functions.
} Ziv and Zakai in \cite{ziv1973functionals}  consider the classical proof for the lossy source coding and attempt to mimic the proof by replacing Shannon's mutual information with the mutual $\mf$-information in the proof steps. They show that the function $\RR^{CKZ}_{\mathsf f}(D)$ is useful in obtaining new infeasibility results for the lossy source coding problem when blocklength $n=1$. In fact, the bound obtained using $\mf$-rate-distortion functions can strictly improve over the ordinary rate-distortion function \eqref{eqnRDnew2}
 \emph{when we restrict to codes with blocklength $n=1$.} Authors in \cite{ziv1973functionals} simply use the fact that mutual $\mf$-information (as defined in
\eqref{IfCKZ} or \eqref{IfPV}) satisfies the data-processing property. 
However, we would like to highlight that the argument in \cite{ziv1973functionals} does not yield a computationally tractable bound for arbitrary $n>1$; this is due to the fact that \cite{ziv1973functionals} only yields a bound in the multi-letter form for $n>1$ (not a single-letter bound). This limitation stems from the fact that $\mf$-mutual-information does not enjoy the chain-rule property of Shannon's mutual information.  

Even though mutual $\mf$-information does not have the chain rule property of Shannon's mutual information, in this paper we state new properties for mutual $\mf$-information which could be utilized in lieu of the chain rule to obtain a \emph{single-letter} bound. Using these new properties, we relate $\RR^{CKZ}_{\mathsf f}(D)$ and $\RR^{MBGYA}_{\mathsf f}(D)$ to the lossy source coding problem when blocklength $n$ is arbitrary. More specifically, our result allows us to generalizes the result in \cite{ziv1973functionals} to arbitrary blocklength $n$.

\begin{theorem}
\label{thmrd3}
For any $(n,R,D)$ lossy source code, there is some $P_{\hat X|X}$ such that 
\begin{itemize}
    \item $\mathbb{E}_{P_{\hat{X}|X}P_{X}}[d(X,\hat{X})]\leq D$
    \item For any arbitrary function $\mf\in \mathcal{F}$ as defined in Definition \ref{def:class-C}, we have
\begin{align}\frac1nH^{CKZ}_{\mathsf f}(K)&\geq 
I^{CKZ}_{\mathsf f}(X;\hat X),\label{nnn1}
\\
\frac1n H^{\text{MBGYA}}_{\mathsf f}(K)&\geq 
I^{\text{MBGYA}}_{\mathsf f}(X;\hat X),
\end{align}
where $K$ is a uniform random variable over $\{1,2,\cdots, 2^{nR}\}$. In particular, equation \eqref{nnn1} can be equivalently written as
\begin{align}\label{eq_RD_f_lossy}
    \frac1n\left\{
\mathsf f(0)\left(1-2^{-nR}\right)+
2^{-nR} \mf\left(2^{nR}\right)\right\}\geq 
I^{CKZ}_{\mathsf f}(X;\hat X).
\end{align}
\end{itemize}
 \end{theorem}

 \begin{proof}[Proof of Theorem \ref{thmrd3}]
Take an encoder $\mathcal{E}:\mathcal{X}^n\mapsto \{1,2,\cdots, 2^{nR}\}$ and a decoder $\mathcal{D}:\{1,2,\cdots, 2^{nR}\}\mapsto\hat{\mathcal{X}}^n$ such that the reconstruction sequence \[\hat{X}^n=\mathcal{D}(\mathcal{E}(X^n))\] satisfies
\[\frac1n\sum_{i=1}^{n}\mathbb{E}[d(X_i, \hat X_i)]\leq D.
\]
Let $M=\mathcal{E}(X^n)\in\{1,2,\cdots, 2^{nR}\}$ be the compression of $X^n$. That is, $M$ and $K$ are defined on the same alphabet set, and $K$ is uniformly distributed. Take a time-sharing random variable $Q$ uniform over $\{1,2,\cdots, n\}$ and independent of previously defined variables. Then,
\[\mathbb{E}_{P_{Q}P_{X^{n}\hat{X}^{n}}}[d(X_Q, \hat X_Q)]=\frac1n\sum_{i=1}^n\mathbb{E}_{P_{X_{i}\hat{X}_{i}}}[d(X_i, \hat X_i)]\leq D,
\]
where $(X_Q, \hat X_Q)$ is a function of the sequence $(X^n,\hat X^n)$ and the random time index $Q$.
Let $(X,\hat X)=(X_Q, \hat{X}_Q)$. Then, the joint distribution of $(X,\hat X)$ will be the average of the joint distributions of $(X_i, \hat X_i)$ for $i=1,2,\cdots,n$. Since $X^n$ is i.i.d., the marginal distribution of $X=X_Q$ will be the same as the marginal distribution of the $X_i$'s. Moreover,  $P_{\hat X|X}$ can be written explicitly as
\begin{align}
    P(\hat{X}\in \mathcal A|X=x)=\sum_{i=1}^n\frac1n P(\hat X_{i}\in \mathcal A|X_i=x).
\end{align}
Then 
\[
\mathbb{E}_{P_{X,\hat{X}}}[d(X, \hat X)]=\mathbb{E}_{P_{Q}P_{X^{n},\hat{X}^{n}}}[d(X_Q, \hat X_Q)]=\frac1n\sum_{i=1}^n\mathbb{E}_{P_{X_{i}\hat{X}_{i}}}[d(X_i, \hat X_i)]\leq D.
\]

For either the CKZ or MBGYA notions of the mutual $\mf$-information, we have
\begin{align}
H_\mf(K)&\geq 
    H_\mf(M)\label{jf0}
    \\&=I_\mf(M;M)\label{jf1}
    \\&\geq
    I_\mf(X^n;\hat X^n)\label{jf2}
    \\&\geq
    \sum_{i=1}^n I_\mf(X_i;\hat X^n)\label{jf3}
    \\&\geq
    \sum_{i=1}^n I_\mf(X_i;\hat X_i)\label{jf4}
    \\&\geq
    n\cdot I_\mf(X_Q;\hat X_Q),\label{jf5}
\end{align}
where \eqref{jf0} follows from property (ii) of Theorem \ref{property_f_entropy}, \eqref{jf1} follows from the definition of $\mf$-entropy, \eqref{jf2}, \eqref{jf4} from data processing property of $\mf$-information (Theorem \ref{dpi_f_information}), and finally \eqref{jf3} follows from property (iii) of Theorem \ref{new properties_I_f}.

\end{proof}

 \begin{corollary}\label{newcorcard} We obtain
 \begin{align}
 \frac1n H^{\text{CKZ}}_{\mathsf f}(K)&\geq
\RR^{\text{CKZ}}_{\mathsf f}(D),\label{eqneqn1}
\\
 \frac1n H^{\text{MBGYA}}_{\mathsf f}(K)&\geq
\RR^{\text{MBGYA}}_{\mathsf f}(D),\label{eqneqn2}
\end{align}
where $K$ is a uniform random variable over $\{1,2,\cdots, 2^{nR}\}$.
Furthermore, for the CKZ notion to compute the minimum in $\RR^{\text{CKZ}}_{\mathsf f}(D)$, it suffices to compute the minimum over all conditional distributions $P_{\hat X|X}$ for $\hat{X}\in \hat{\mathcal{X}}$ such that the support of $\hat X$ has size at most $|\mathcal{X}|+1$.
 
 \end{corollary}
 \begin{proof}
 Inequalities \eqref{eqneqn1} and \eqref{eqneqn2} follow immediately from Theorem \ref{thmrd3}. It remains to show the bound on the size of the support of $\hat{X}$ when computing $\RR^{\text{CKZ}}_{\mathsf f}(D)$. The cardinality bounds on the auxiliary random variable $\hat{W}$  comes from the standard Caratheodory-Bunt \cite{bun34} arguments and is omitted. We just point out that similar to the ordinary mutual information,  $\mf$-information $I^{CKZ}_{\mathsf f}(X;\hat{X})$ has the  following property: fix a conditional distribution $P_{X|\hat X}$ and consider the set of marginal distributions $P_{\hat X}$ that induce a given marginal distribution on $P_X$. Then, 
 \[
 I^{CKZ}_{\mathsf f}(X;\tilde{X})=\mathbb{E}_{P_XP_{\tilde{X}}}\mf\left(\frac{dP_{X|\tilde{X}}}{dP_X}(X,\hat X)\right)
 \]
 is linear in the marginal distribution $P_{\hat X}$. 

 \end{proof}
\begin{remark} In the asymptotic regime when the blocklength $n$ tends to infinity, a full characterization of the achievable rate-distortion pairs $(R,D)$ is known and given in \eqref{eqnRDnew2}. As a sanity check, let us restrict the above theorem to the case of $n$ tending to infinity. Assume that a rate-distortion pair $(R,D)$ is achievable as $n$ tends to infinity. Letting $n$ converge to infinity we obtain
\[\lim_{n\rightarrow\infty}
\frac1n\left\{
\mathsf f(0)\left(1-2^{-nR}\right)+
2^{-nR}\mathsf f\left(2^{nR}\right)\right\}\geq 
\RR^{CKZ}_{\mathsf f}(D)\]
or equivalently,
\[R\lim_{t\rightarrow\infty}
\frac{\mathsf{f}(0)(t-1)+\mathsf f\left(t\right)}{t\ln(t)}\geq 
\RR^{CKZ}_{\mathsf f}(D).\]
If $R=\RR(D)$ lies on the rate-distortion curve,  from Lemma \ref{newlmma3} in Section \ref{class_conv_func_F} we obtain
\[\RR(D)\left[\lim_{t\to \infty}t\mathsf f''(t)\right]\geq
\RR^{CKZ}_{\mathsf f}(D).\]
 The correctness of the above inequality can be verified using part (iv) of Theorem 
 \ref{new properties_I_f}.
\end{remark}

\begin{example}
Let $X\sim B(1/2)$ be a uniform binary source (i.e., $X\in\{0,1\}$ is uniform). Let $\mathcal{\hat X}=\mathcal{X}=\{0,1\}$ and $d(x,\hat{x})=\boldsymbol{1}(x\neq\hat{x})$ be the Hamming distance. So, the average distortion is in fact the probability of mismatch: $\mathbb{E}[d(X,\hat{X})]=\mathbb{P}[X\neq \hat{X}]$. The classical rate-distortion function, in this case, is given by \cite[Example 11.1]{el2011network}:
\[
\RR(D)=1-H_2(D),\qquad \forall~D\in[0,0.5].
\]
where $H_2(x)=-x\log_{2}(x)-(1-x)\log_{2}(1-x)$ is the binary entropy function. The choice of $\mf(x)=x\log_{2}(x)$ in  \eqref{eq_RD_f_lossy} yields that $R\ge \RR(D)$. This lower bound holds for universally any arbitrary $n$. We next show that better lower bounds can be found for finite values of $n$.

Let $\mf(x)=x^2-1$ which belongs to $\mathcal{F}$. Then,  $\RR_{\mf}^{CKZ}(D)$ is achieved by a binary symmetric channel $P_{\hat{X}|X}(1|0)=P_{\hat{X}|X}(0|1)=D$ and
\[
\RR_{\mf}^{CKZ}(D)=(2D-1)^2,\qquad \forall~D\in[0,0.5].
\]
For this choice of $\mf$, Theorem \ref{thmrd3} implies that for any $(n,R,D)$ lossy source code, we have
\[R\geq LB_{\mf}(D,n):=\frac{1}{n}\log_{2}(n\RR^{CKZ}_{\mf}(D)+1).\]

From the rate-distortion theory, we know that $R\geq \RR(D)$ (which holds for arbitrary values of $n$). Figure \ref{fig:new_LB_Rf} plots the lower bounds $LB_{\mf}(D,n)$ and $\RR(D)$ for different values of $n$.
 As the figure indicates, the lower bound $LB_{\mf}(D,n)$ is a non-trivial lower bound for a finite blocklength lossy source code. 
 \end{example}

 \begin{claim}\label{claim1}
 We have that $LB_{\mf}(D,n)\ge \RR(D)$ for $\frac{1}{2}-\frac{1}{2\sqrt{n}}\le D\le\frac{1}{2}$ and arbitrary block length $n\ge 1$. 
 \end{claim}
 The proof is given in Appendix \ref{proof of claim1}.

\begin{figure}
	   \centering
	\includegraphics[scale=1,width=0.8\linewidth]{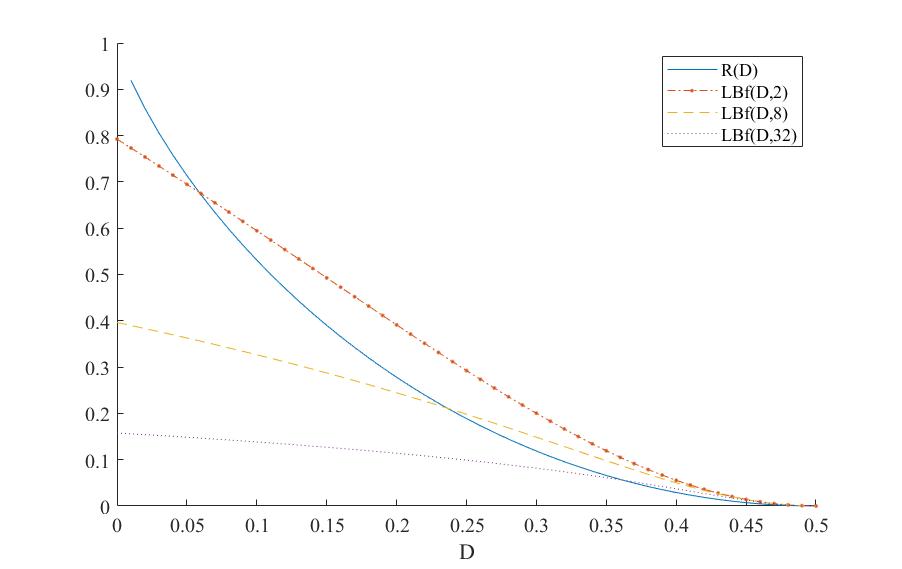}
	\caption{ Lower bounds on achievable rate $R$ versus $D$ for a code of length $n$ for a binary source $X\sim Bernoulli(1/2)$ and Hamming distortion. 
	The new lower bound $LB_{\mf}(D,n)=\frac{1}{n}\log_{2}(n\RR^{CKZ}_{\mf}(D)+1)$ for $\mf(x)=x^2-1$ is plotted along with $\RR(D)$  for $n=2,8,32$.}
	\label{fig:new_LB_Rf}
\end{figure}

 The following theorem shows single-letterization of the $\mf$-rate distortion function when $\mf\in\mathcal{F}$. 

\begin{theorem}\label{thmNewRD}
    Take some arbitrary $\mf\in\mathcal{F}$. Let $X^n$ is i.i.d.\ according to some $P_X$ for some arbitrary natural number $n$. Then, we have
\begin{align}
 \inf_{P_{\hat X^n|X^n}:~I^{CKZ}_{\mathsf f}(X^n;\hat{X}^n)\le nR}\frac1n\sum_{i=1}^n\mathbb{E}[d(X_i,\hat{X}_i)]\geq \D^{CKZ}_{{\mathsf{f}}}(R).
  \end{align}
  Similarly,
  \begin{align}
 \inf_{P_{\hat X^n|X^n}:~I^{\text{MBGYA}}_{\mathsf f}(X^n;\hat{X}^n)\ge nR}\frac1n\sum_{i=1}^n\mathbb{E}[d(X_i,\hat{X}_i)]\geq \D^{\text{MBGYA}}_{{\mathsf{f}}}(R).
  \end{align}
\end{theorem}
The proof is given in Appendix \ref{proof_thmNewRD}.


\section{Lower bounds on the ($\mf$-)rate-distortion function}\label{sec:lower bound on rate-dist}

The ($\mf$-)rate-distortion function is expressed in terms of an optimization problem. It does not have an explicit expression. However, there are some explicit lower bounds on
the rate-distortion function such as Shannon's lower bound. Such lower bounds are of independent interest and have been previously studied in the literature. A new motivation for studying such lower bounds stems from a connection that we provide between lower bounds on the rate-distortion function and upper bounds on the generalization error of learning algorithms. This connection is discussed later in Section \ref{subsec: lossy_comp_to_gen}.

In this section we discuss explicit lower bounds on the ($\mf$-)rate-distortion function. We  provide new bounds that outperform the existing bounds in some cases. We begin our discussion with the ordinary rate-distortion function and extend the result to ($\mf$-)rate-distortion function afterwards.

\subsection{Review of the existing bounds and ideas in the literature}
The following explicit lower bound on the rate-distortion is known:

\begin{theorem}\cite[Theorem 55]{koliander2016entropy},\cite[Lemma 2]{riegler2018rate}\label{th5_generalized_SLB}
Consider a random variable $X\sim P_X$ distributed on the measure space $(\mathcal{X},\A,\mu)$, a measurable space $(\mathcal{Y},\B)$, and a
distortion function $d:\mathcal{X}\times \mathcal{Y}\to [0,\infty]$ satisfying  1) $\inf_{y\in\mathcal{Y}}d(x,y)=0$ for all $x\in\mathcal{X}$. 2) there exists a finite set $\mathcal{B}\subseteq\mathcal{Y}$ such that $\mathbb{E}[\min_{y\in\mathcal{B}}d(X,y)]=0$. Suppose that $\mu$ is a reference measure for random variable $X$ (i.e. $P_{X}\ll \mu$), then
    \begin{align}
        \RR(D)\ge h_{\mu}(X)-\inf_{\lambda\ge 0}\{\lambda D+\ln\nu(\lambda)\}
    \end{align}
    where
    \[
   h_{\mu}(X)=-\mathbb{E}\left[\ln\left(\frac{dP_{X}}{d\mu}(X)\right)\right],\quad \nu(\lambda)=\sup_{y\in{\mathcal{Y}}}\int e^{-\lambda d(x,y)}d\mu(x), \quad \forall \lambda\in [0,\infty).
    \]
\end{theorem}

\begin{remark}
Theorem \ref{th5_generalized_SLB} holds for all values of distortion $D\geq 0$. In \cite{marton1971asymptotics}, Marton studied the asymptotic of $\RR(D)$ when $D\to 0$, showing that for discrete sources 
\[
\RR(D)\ge H(P_{X})+c\cdot D\log_{2}{D}+\mathcal O(D)
\]
where $c$ is a constant determined by the source.
\end{remark}

Another known idea (originally developed by Shannon) for finding a lower bound on the rate-distortion function when $\mathcal{X}=\hat{\mathcal{X}}=\mathbb{R}^k$ is as follows:
\begin{align*}
    I(X;\hat X)&=h(X)-h(X|\hat X)
    \\&=h(X)-h(X-\hat X|\hat X)
    \\&\geq h(X)-h(X-\hat X)
    \\&\geq h(X)-\sup_{p_{\hat X|X}:\mathbb{E}[ d(X,\hat X)]\leq D}h(X-\hat X).
\end{align*}

Therefore, for any arbitrary $d(x,\hat x)$, we deduce that
\[
\RR(D)\geq h(X)-\sup_{p_{\hat X|X}:\mathbb{E}[ d(X,\hat X)]\leq D}h(X-\hat X).
\]

Assume that we have a distortion measure $d(\cdot,\cdot)$ of the form \[d(x,\hat x)=L(x-\hat x)\qquad\forall x,\hat x\in \mathbb{R}^{k}\] for some function $L:\mathbb{R}^{k}\mapsto \mathbb{R}$. Then, we obtain Shannon's lower bound on the rate-distortion function:\cite[Lemma 4.6.1]{gray2012source}
\[
\RR(D)\ge h(X)-\sup_{p_{N}:~\mathbb{E}[L(N)]\le D}h(N).
\]
Moreover, the solution to the maximization problem
\begin{align}\label{non-gaussian-max-entropy}
    \max_{p_{X}:~\mathbb{E}[L(X)]\le D}h(X)
\end{align}
is of the form 
\[p^*(x)=\frac{e^{-bL(x)}}{\int e^{-bL(x)}dx}\] where the constant $b$ is chosen such that
\[
\int L(x)p^*(x)dx= D
\]

Note that Theorem \ref{th5_generalized_SLB} recovers this lower bound with the choice of $\mu$ being the Lebesgue measure.
\begin{corollary}
  
\cite[Section 4.8]{gray2012source} 
 Let $X\sim P_{X}$ on $\mathbb{R}^{k}$. Assume that
 $\mathcal{X}=\hat{\mathcal{X}}=\mathbb{R}^k$, and take a distortion function $d(x,\hat{x})=\|x-\hat{x}\|^{r}_{r}$ for all $x,\hat{x}$.
 Then we have
 \begin{align}
     \RR(D)\ge h(X)-\ln\left(\left(\frac{D re}{k}\right)^{\frac{k}{r}}\Gamma(1+k/r)V_{k}\right)
 \end{align}
 where $V_{k}$ is the volume of the unit radius ball in $\mathbb{R}^k$, i.e.,
 \[
 V_{k}\triangleq \mathsf{Vol}\left[\{x\in\mathbb{R}^{k}:\,\|x\|_{r}\le 1\}\right]. \]\label{th_shannon_lower_bound}
  \end{corollary}

\subsection{A new lower bound on the rate-distortion function}
In this subsection, we give a new lower bound on the rate-distortion function that can be strictly better than the lower bound of Theorem \ref{th5_generalized_SLB} (taken from \cite[Theorem 55]{koliander2016entropy},\cite[Lemma 2]{riegler2018rate}). Moreover, this bound is useful when we study the generalization error of learning algorithms.
Consider a  random variable $X$. Then, we have
\begin{align}
	\D(R)&=\inf_{P_{\hat{X}|X}: I(X;\hat{X})\le R} \mathbb{E}\left[d(X,\hat{X})\right]
	\\
	&=
	\inf_{P_{\hat{X}|X}}\sup_{\lambda\geq 0}\mathbb{E}\left[d(X,\hat{X})\right]-\lambda R+\lambda D(P_{\hat{X}X}\|P_{\hat{X}}P_{X}) 
	\\
	&\ge
\sup_{\lambda\geq 0}\inf_{P_{\hat{X}|X}}\mathbb{E}\left[d(X,\hat{X})\right]-\lambda R+\lambda D(P_{\hat{X}X}\|P_{\hat{X}}P_{X})\label{eqneqn30}\\
&=
\sup_{\lambda\geq 0}\inf_{P_{\hat{X}|X}}\inf_{Q_{\hat{X}}}\mathbb{E}_{P_{\hat X X}}\left[d(X,\hat{X})\right]-\lambda R+\lambda D(P_{\hat{X}X}\|Q_{\hat{X}}P_{X})\label{dual_char_RDnnn}\\
&=\sup_{\lambda\geq 0}\inf_{Q_{\hat{X}}}-\lambda R-\lambda\mathbb{E}_{X\sim P_{X}}\left[\ln\left\{\mathbb{E}_{
\substack{
\hat X\sim Q_{\hat{X}}
\\\hat X \text{ ind. of } X}
}\left[e^{-\frac{d(X,\hat{X})}{\lambda}}\right]\right\}\right]\label{eqneqnc}
\end{align}
where \eqref{eqneqn30} follows from minimax theorem and \eqref{dual_char_RDnnn} follows from the following equality:
\begin{align*}
D(P_{\hat{X}X}\|Q_{\hat{X}}P_{X})=D(P_{\hat{X}X}\|P_{\hat{X}}P_{X})+D(P_{\hat{X}}\|Q_{\hat{X}}).
\end{align*}
To obtain \eqref{eqneqnc}, note that for every fixed $Q_{\hat X}$, the minimizing $P_{\hat{X}|X}$ in \eqref{dual_char_RDnnn} is the Gibbs measure:
\[
dP^{*}_{\hat{X}|X}(\hat{x}|x):=\frac{dQ_{\hat{X}}(\hat{x})e^{-\frac{d(x,\hat{x})}{\lambda}}}{\mathbb{E}_{Q_{\hat{X}}}\left[e^{-\frac{d(x,\hat{X})}{\lambda}}\right]},
\]
Then \eqref{eqneqnc} follows from substituting $P^{*}_{\hat{X}|X}(\hat{x}|x)$ in \eqref{dual_char_RDnnn}.
Applying Jensen's inequality on  \eqref{eqneqnc}, we get the following theorem:
\begin{theorem}\label{th3} Take an arbitrary source $X\sim P_{X}$ defined on a set $\mathcal{X}$ and a distortion function $d(x,\hat{x})\in\mathbb{R}$ (which may be negative). Then, the following two statements hold:
  We have
    \begin{align}\label{lower_psi_1_ineq}
     \D(R)\ge \sup_{\lambda\ge0}\inf_{Q_{\hat{X}}}\left\lbrace-\frac{R}{\lambda}-\frac{1}{\lambda}\ln\left[\mathbb{E}_{P_{X}Q_{\hat{X}}}\left[e^{-\lambda d(X,\hat{X})}\right]\right]\right\rbrace\ge \sup_{\lambda\ge0}\left\lbrace-\frac{R}{\lambda}-\frac{1}{\lambda}\phi(-\lambda)\right\rbrace
    \end{align}
    where
$\phi(\cdot)$ is a function defined as follows: 
    \[
\phi(\lambda)=\sup_{\hat{x}} \ln\mathbb{E}_{P_{X}}\left[e^{\lambda d(X,\hat{x})}\right].
    \]
Note that if for some $\lambda$,  $\mathbb{E}_{P_{X}}\left[e^{\lambda d(X,\hat{x})}\right]=\infty$ for some $\hat x$, we set $\phi(\lambda)=\infty$.
\end{theorem}
\begin{remark}
The lower bound in \eqref{lower_psi_1_ineq} can be used to find bounds on the generalization error of learning algorithms, while {\color{blue}Assumption 2) of} Theorem \ref{th5_generalized_SLB}  precludes the application of lower bound in Theorem \ref{th5_generalized_SLB} in the context of generalization error.
\end{remark}

\begin{corollary}
\label{sub-gaussian-D(R)-lower-bound}
Assume that $d(X,\hat{x})$ satisfies
\[
\phi(\lambda)=\sup_{\hat{x}} \ln\mathbb{E}_{P_{X}}\left[e^{\lambda d(X,\hat{x})}\right]\le\lambda^{2}\sigma^{2}/2.
\]
Then
\[
\D(R)\ge\sup_{\lambda\ge0}\left\lbrace-\frac{R}{\lambda}-\lambda\sigma^{2}/2\right\rbrace=-\sqrt{2\sigma^{2}R}.
\]
\end{corollary}

\begin{example}\label{remark6}
The lower bound given in part (i) of Theorem \ref{th3} can be stronger than the lower bound of Theorem \ref{th5_generalized_SLB} (taken from \cite[Lemma 2]{riegler2018rate}). For instance, assume that $\mathcal{X}=\{1,2,3\}$ and $P_{X}(1)=0.5, P_{X}(2)=0.01, P_{X}(3)=0.49$.
Let $d(x,\hat{x})=\boldsymbol{1}\{x\neq \hat{x}\}$. Theorem \ref{th5_generalized_SLB} gives a lower bound on $\RR(D)$. It implies the following lower bound on $\D(R)$:
\begin{align}
\text{\rm{LB}}_{1}(R)=\sup_{\lambda\ge 0}\left\{-\frac{R}{\lambda}
-\frac{1}{\lambda}\left(\ln\sup_{\hat{x}}\left(\sum_{x}e^{-{\lambda d(x,\hat{x})}}\right)-H(P_{X})\right)\right\}.
\end{align}
On the other hand, the lower bound of Theorem \ref{th3} evaluates to
\begin{align}
\text{\rm{LB}}_{2}(R)=\sup_{\lambda\ge 0}\left\{-\frac{R}{\lambda}
-\frac{1}{\lambda}\ln\sup_{\hat{x}}\mathbb{E}_{P_{X}}\left[e^{-{\lambda d(X,\hat{x})}}\right]\right\}.
\end{align}
These two lower bounds are plotted in Figure \ref{fig:compare_new_LB_Rf}. This plot indicates that
the lower bound given in part (i) of Theorem \ref{th3} can be stronger than the lower bound of Theorem \ref{th5_generalized_SLB}.

\begin{figure}
	   \centering
	\includegraphics[scale=0.8,width=0.7\linewidth]{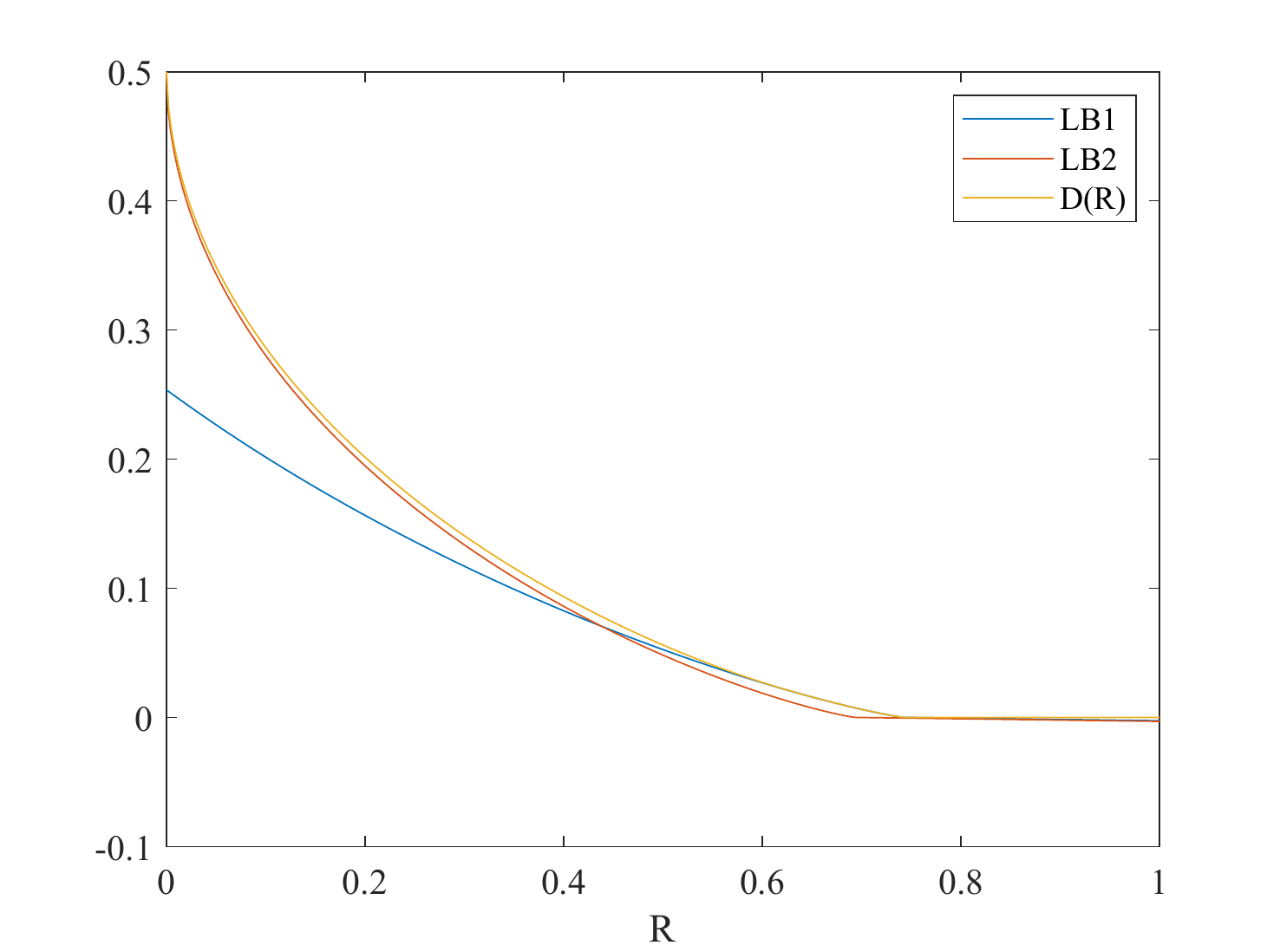}
	\caption{With the setting of Remark \ref{remark6}, $\rm{LB}_{1}(R)$ is the lower bound of distortion-rate function in Theorem \ref{th5_generalized_SLB} (taken from \cite[Lemma 2]{riegler2018rate}) and $\rm{LB}_{2}(R)$ is the lower bound proposed in part (i) of Theorem \ref{th3}.}
	\label{fig:compare_new_LB_Rf}
\end{figure}
\end{example}

\subsection{Lower bound on the $\mf$-rate-distortion function}

The lower bounds of the previous section can be generalized to the $\mf$-rate-distortion functions defined in \eqref{eqnDfv1a} and \eqref{eqnDfv2a} (note that for the special case of $\mf(t)=t\ln(t)$, the $\mf$-rate-distortion functions in \eqref{eqnRDfCKZ} and \eqref{eqnRDfPV}  reduce to the ordinary rate-distortion functions in
\eqref{eqnRDnew2}). Note that the function $\mf:[0,\infty)\to\mathbb{R}\cup\{\infty\}$ is not necessarily in class function $\mathcal{F}$ in this section. We also use $\mf^{*}:\mathbb{R}\to\mathbb{R}\cup\{\infty\}$ to denote the convex conjugate function of $\mf$ defined as follows:
\begin{align}\mf^{*}(x)=\sup_{t\geq0}\{tx-\mf(t)\}, \qquad\forall x\in\mathbb{R}.\label{eqnfs}\end{align}
For instance, if $\mf(t)=t\ln(t)$, we have $\mf^*(t)=e^{t}-t+1$.

Since $I^{CKZ}_{{\mf}}(\hat{X};X)\geq I^{PV}_{{\mf}}(\hat{X};X)$, we have $\D^{CKZ}_{{\mf}}(R)\geq \D^{PV}_{{\mf}}(R)$ for every rate $R$. Therefore, any lower bound on $\D^{PV}_{{\mf}}(R)$ is also a lower bound on $\D^{CKZ}_{{\mf}}(R)$. In what follows, we only consider $\D^{PV}_{\mf}(R)$:
\begin{align}
    -\D^{PV}_{{\mf}}(R) 
    &=\sup_{P_{\hat{X}|X}:~I_{\mf}^{PV}(X;\hat{X})\le R}\mathbb{E}_{P_{\hat{X}X}}[-d(X,\hat{X})]
    \\&= \sup_{Q_{\hat{X}}}\sup_{\substack{P_{\hat{X}|X}:\\~D_{{\mf}}(P_{\hat{X}X}\|Q_{\hat{X}}P_{X})\le R}}\mathbb{E}_{P_{\hat{X}X}}[-d(X,\hat{X})]\nonumber\\
    &\le \sup_{Q_{\hat{X}}}\sup_{\substack{\nu_{\hat{X}X}\ll Q_{\hat{X}}P_{X}:\\~D_{{\mf}}(\nu_{\hat{X}X}\|Q_{\hat{X}}P_{X})\le R}}\mathbb{E}_{\nu_{\hat{X}X}}[-d(X,\hat{X})]
    \\&= \sup_{Q_{\hat{X}}}~\mathbb{F}^{D_{\mf}}_{P_{X} Q_{\hat{X}},R}[-d(X,\hat{X})],\label{LB_f_RD_F}
\end{align}
where for any arbitrary random variable $Z\sim \mu$ and any function $g(\cdot)$ we define
\begin{align}\label{F_r_D_f}
\mathbb{F}^{D_{{\mf}}}_{\mu,R}[g(Z)]\triangleq\sup_{\substack{\nu\ll\mu:\\~D_{\mf}(\nu\|\mu)\le R}}\mathbb{E}_{Z\sim\nu}[g(Z)].
\end{align}
The definition of smoothed expectation \eqref{F_r_D_f} also appears in the context of distributionally robust optimization problems \cite{ben2013robust, bauso2017distributionally, duchi2018learning, birrell2020distributionally}. In \cite{ben2013robust}, the authors give an equivalent characterization of smoothed expectation. Theorem \ref{th_lowerbound_Fr_Df_equiv_Orlicz} in Appendix \ref{f-distortion_norm_append} gives a similar characterization. We state and prove the following theorem which is derived by Theorem \ref{th_lowerbound_Fr_Df_equiv_Orlicz}.

\begin{theorem}\label{cor_f-dist_lower_bound}
Let $\mf:[0,\infty)\to \mathbb{R}\cup\{\infty\}$ be a convex function. Take an arbitrary source $X\sim P_{X}$ defined on a set $\mathcal{X}$ and a distortion function $d(x,\hat{x})\in\mathbb{R}$ (which may be negative). Then, the following two statements hold:

     We have
    \begin{align}\label{lower_general_f_ineq}
     \D^{PV}_{\mf}(R)\ge  \sup_{\lambda\ge 0,~a\in\mathbb{R}}\left\lbrace-\frac{1}{\lambda}\left[a+R\right]-\frac{1}{\lambda}\phi_{\mf}(-\lambda,a)\right\rbrace,
    \end{align}

    where
$\phi(\cdot,\cdot)$ is a function defined on $[0,\infty)\times\mathbb{R}$ as follows:
    \[
   \phi_{\mf}(\lambda,a)=\sup_{\hat{x}} \mathbb{E}_{P_{X}}{\mf}^{*}\left(\lambda{d(X,\hat{x})}-a\right).
    \]
\end{theorem}
The proof is given in Section \ref{proof_cor_f-dist_lower_bound}.
\begin{remark}
The lower bound in \eqref{lower_general_f_ineq} yields the lower bound in \eqref{lower_psi_1_ineq} when specializing to $\mf(x)=x\ln x-x+1$ for $x\ge 0$ ($\mf(0):=0$). For this function $\mf^{*}(y)=e^{y}-1$ for $y\in\mathbb{R}$. The lower bound in Part (a) of Theorem \ref{cor_f-dist_lower_bound} yields:
  \begin{align}
     \D^{PV}_{\mf}(R)&\ge  \sup_{\lambda\ge 0,~a\in\mathbb{R}}\left\lbrace-\frac{1}{\lambda}\left[a+R\right]-\frac{1}{\lambda}\sup_{\hat x} \mathbb{E}_{P_{X}}{\mf}^{*}\left(-\lambda{d(X,\hat{x})}-a\right)\right\rbrace\nonumber\\
     &= \sup_{\lambda\ge 0}\left\lbrace-\frac{1}{\lambda}R-\frac{1}{\lambda}\inf_{a\in\mathbb{R}} \{a+\sup_{\hat x}\{\mathbb{E}_{P_{X}}e^{-\lambda{d(X,\hat{x})}}\}e^{-a}-1\}\right\rbrace\nonumber\\
          &= \sup_{\lambda\ge 0}\left\lbrace-\frac{1}{\lambda}R-\frac{1}{\lambda} \ln\sup_{\hat{x}}\{\mathbb{E}_{P_{X}}e^{-\lambda{d(X,\hat{x})}}\}\right\rbrace\label{eqnError}\\
     &=\sup_{\lambda\ge 0}\left\lbrace-\frac{1}{\lambda}R-\frac{1}{\lambda}\sup_{\hat x}\ln\mathbb{E}_{P_{X}}\left[e^{-\lambda d(X,\hat{x})}\right]\right\rbrace.
    \end{align}
\end{remark}


\section{Generalization error of learning algorithms with bounded input/output information}\label{sec: gen}
Consider a learning problem with an instance space $\mathcal{Z}$, a hypothesis space $\mathcal{W}$ and a loss function $\ell:\mathcal{W}\times \mathcal{Z}\to \mathbb{R}$. Assume that the test and training samples are produced (in an i.i.d. fashion) from an unknown distribution $\mu$ on $\mathcal{Z}$ respectively. A training  dataset of size $n$ is shown by the $n$-tuple, $S=(Z_{1},Z_{2},\cdots,Z_{n})\in\mathcal{Z}^n$ of i.i.d. random elements according to an unknown distribution $\mu$. 
A learning algorithm is characterized by a probabilistic mapping $\A(\cdot)$ (a Markov Kernel) that maps training data $S$ to the  random variable $W=\A(S)\in\mathcal{W}$ as the output hypothesis. The population risk of a hypothesis $w\in \mathcal{W}$ is computed on the test distribution $\mu$ as follows: 
\begin{align}
	L_{\mu}(w)\triangleq\mathbb{E}_{\mu}[\ell(w,Z)]=\int_{\mathcal{Z}}\ell(w,z)\mu(dz), \qquad\forall w\in\mathcal{W}.
\end{align}
The goal of learning is to ensure that under any data-generating distribution $\mu$, the population risk of the output hypothesis $W$ is small, either
in expectation or with high probability. Since $\mu$ is unknown, the learning algorithm cannot directly compute $L_{\mu}(w)$ for any $w\in \mathcal{W}$, but can compute the empirical risk of $w$ on the training dataset $S$ as an approximation, which is defined as
\begin{align}
	L_{S}(w)\triangleq\frac{1}{n}\sum_{i=1}^{n}\ell(w,Z_{i}).
\end{align}
The true objective of the learning algorithm, $L_{\mu}(W)$, is unknown to the learning algorithm while the empirical risk $L_{S}(W)$ is known.  The generalization gap is defined as the difference between these two quantities as
\begin{align}\label{old_notion}
	\mathrm{gen}_{\mu}(W,S)= L_{\mu}(W)-L_{S}(W),
\end{align}
where $W=\A(S)$ is the output of the algorithm $\A$ on the input $S\sim (\mu)^{\otimes n}$. In common algorithms such as empirical risk minimization (ERM) and gradient descent, $L_{S}(W)$ is minimized \cite{shalev2010learnability, hardt2016train}. Therefore, to control $L_{\mu}(W)$ we need to bound $\mathrm{gen}_{\mu}(W,S)$ from above (in expectation or with high probability). Observe that $\mathrm{gen}_{\mu}(W,S)$, as defined in \eqref{old_notion}, is a random variable and a function of $(S,W)$. The generalization error is the expected value of $\mathrm{gen}_{\mu}(W,S)$:
\begin{align}\label{new_notion2}
	\mathrm{gen}\left(\mu,\A\right)= \mathbb{E}\left[L_{\mu}(W)-L_{S}(W)\right].
\end{align}

 Designing algorithms with low generalization error is a key challenge in machine learning. The following upper bound on the generalization error is given in \cite{xu2017information} (see also \cite{russo2019much}):
\begin{theorem}\cite{xu2017information}\label{th1_without_mis}
	Suppose $\ell(w,Z)$ is $\sigma^{2}$-sub-Gaussian under $Z\sim \mu$ for all $w\in\mathcal{W}$. Take an arbitrary algorithm $\A$ that runs on a training dataset $S$. Then the generalization error is bounded as
	\[
	\mathrm{gen}(\mu,\A)\le\sqrt{\frac{2\sigma^{2}}{n}I\big(S;\A(S)\big)}.
	\]
\end{theorem}
In \cite{bu2020tightening}, a strengthened version of Theorem \ref{th1_without_mis} is given as follows:
\begin{theorem}\cite{bu2020tightening}\label{thm05}
    Suppose that the  loss function $\ell(w,Z)$ is $\sigma^{2}$-sub-Gaussian under the distribution $\mu$ on $Z$ for any $w\in \mathcal{W}$. We have:
    \begin{align}
    \mathrm{gen}(\mu,\A)\le\frac{1}{n}\sum_{i=1}^{n}\sqrt{{2\sigma^{2}}I\big(Z_{i};\A(S)\big)}.
    \end{align}
\end{theorem}

\subsection{From lossy compression to generalization error}\label{subsec: lossy_comp_to_gen}
Theorem \ref{th1_without_mis} provides an upper bound on the  generalization error in terms of  $I(S;\A(S))$.
Let us now write the \emph{sharpest} possible bound on the generalization error given the knowledge of $I(S;\A(S))$. For any $R>0$, we define
\begin{align}
\U_1(R)\triangleq \sup_{P_{W|S}:~I(S;\hat W)\leq R}
\mathbb{E}\left[\mathrm{gen}_{\mu}(\hat W,S)\right]\label{eqnRD1}
\end{align}
where 
\begin{align}\label{gen_gap}
 \mathrm{gen}_{\mu}(w,s)=L_{\mu}(w)-L_{s}(w)
\end{align}
and the supremum in \eqref{eqnRD1} is over all Markov kernels $P_{\hat W|S}$ with a bounded input/output mutual information and $S\sim (\mu)^{\otimes n}$. Given an algorithm $W=\mathcal{A}(S)$, its generalization error is bounded from above by $\U_1(I(S;W))$ since in defining $\U_1(\cdot)$ we take supremum over \emph{all} algorithms with a maximum input/output mutual information.

We claim that $\U_1(R)$ is related to the rate-distortion function. To see this, consider a rate-distortion problem where the input symbol space is $\mathcal{S}$, the reproduction space is $\mathcal{W}$ and the following distortion function between a symbol $w$ and an input symbol $s$ is used:\footnote{
While the literature commonly takes the reproduction space to be the same as the input symbol space, the rate-distortion theory does not formally require that.
}
\begin{align}\label{gen_gap1}
 -\mathrm{gen}_{\mu}(w,s)=L_{s}(w)-L_{\mu}(w).   
\end{align}
 
With this definition, from \eqref{eqnRD1}, we obtain
\begin{align}
-\U_1(R)&
=\inf_{P_{\hat W|S}:~I(S;\hat W)\leq R}
\mathbb{E}\left[-\mathrm{gen}_{\mu}(\hat W,S)\right]\label{eqnRD2}
\end{align}
which is in the rate-distortion form. 
With $\U_1(R)$ defined as in \eqref{eqnRD1}, it follows that for any arbitrary algorithm $\A$ with $I\big(S;\A(S)\big)\leq R$ we have
\[
\mathrm{gen}\left(\mu,\A\right)\le \U_1(R).
\]
This upper bound does not require any sub-Gaussianity assumption on the loss function. From this viewpoint, Theorem \ref{th1_without_mis} 
is just a convenient and explicit lower bound on a rate-distortion function under an extra assumption on the loss function.
In fact, Corollary \ref{sub-gaussian-D(R)-lower-bound}  shows that
\begin{align}
	 \mathcal{U}_1(R)
		\le	\sqrt{\frac{2\sigma^{2}}{n}R},
	\end{align}
yielding Theorem \ref{th1_without_mis}.

Note that by a similar argument, if instead of an upper bound on $I(S;\A(S))\leq R$, we know $I(Z_i;\A(S))\le R_{i}$ for $1\leq i\leq n$, the following upper bound on the generalization error follows:
\begin{align}
\mathrm{gen}\left(\mu,\A\right)=\mathbb{E}[\mathrm{gen}_{\mu}(W,S)]=\frac{1}{n}\sum_{i=1}^{n}\mathbb{E}\left[\mathrm{gen}_{\mu}(W,Z_{i})\right]
\leq \frac1n\sum_{i=1}^n{\U}_{2}(R_{i})\label{bnu2}
\end{align}
where \begin{align}\text{\rm{gen}}_{\mu}(w,z_i)=\mathbb{E}_{Z \sim \mu}\left[\ell(Z,w)\right]-\ell(z_i,w)\label{gen_gap2}\end{align} and
\begin{equation}
\mathcal{U}_2(R)\triangleq \sup_{P_{\hat W|Z}:~I(\hat W;Z)\leq R}
\mathbb{E}\left[\mathrm{gen}_{\mu}(\hat{W},Z)\right]\label{eqnRDN22}
\end{equation}
where $Z\in\mathcal{Z}$ is distributed according to $\mu$. As before, $-\mathcal{U}_2(R)$ is in the rate-distortion form. Then, Corollary \ref{sub-gaussian-D(R)-lower-bound} yields Theorem \ref{thm05} (from \cite{bu2020tightening}). Computation of $\U_2(\cdot)$ is significantly easier than $\U_1(\cdot)$ as the supremum is taken over a smaller set in $\U_2(\cdot)$. Computational aspects of the above bounds are discussed in Section \ref{NewSecComputational}. 

In the rest of this subsection, we discuss several ideas to improve the above bounds based on $\mathcal{U}_1(R)$ and $\mathcal{U}_2(R)$.
We motivate our discussion with the following example: consider the problem of learning the mean of a Gaussian random vector $Z\sim \mathcal{N}(\beta,\sigma^{2}I_{d})$ with loss function $\ell(w,z)=\|w-z\|^{2}$. This example has been considered in \cite{bu2020tightening,sefidgaran2022rate}. The generalization error of the ERM algorithm $W_{ERM}=\frac{1}{n}\sum_{i=1}^{n}Z_{i}$ can be computed exactly as
\begin{align}
    \text{\rm{gen}}(P_{Z},P_{W|S})=\frac{2\sigma^{2}d}{n}.\label{eqnkkkk2}
\end{align}

Note that the output of ERM algorithm $W_{ERM}=\frac{1}{n}\sum_{i=1}^{n}Z_{i}$ and $W_{ERM}\sim \mathcal{N}(\beta,\frac{\sigma^{2}}{n}I_{d})$. Thus $I(W_{ERM};S)=\infty$ and there is no upper bound on generalization error by assuming just the knowledge of $I(W_{ERM};S)$. Next, note that 
\begin{align}\label{mutual_GME}
    I(W_{ERM};Z_{i})=\frac{d}{2}\ln\frac{n}{n-1}
\end{align}
where it is obtained in \cite{bu2020tightening}.  Assuming that we know $I(W;Z_{i})$ and 
$\ell(W,Z)$ is $\sigma^{2}$-sub-Gaussian under $P_{W}P_{Z}$, Bu et al \cite{bu2020tightening} obtained an upper bound on the generalization error of order $1/\sqrt{n}$. Note that the correct order is $1/n$. 
What if we want to derive an upper bound on the generalization error just with the knowledge of $I(W;Z_{i})=\frac{d}{2}\ln\frac{n}{n-1}$?
As we saw above, if we just know the values of $I(\A(S);Z_i)=\frac{d}{2}\ln\frac{n}{n-1}$ for $i=1,2,\cdots, n$, we obtain
\begin{align}
\mathrm{gen}\left(\mu,\A\right)
\leq \frac1n\sum_{i=1}^n{\U}_{2}\left(\frac{d}{2}\ln\frac{n}{n-1}\right)=\U_2\left(\frac{d}{2}\ln\frac{n}{n-1}\right).
\label{eqnkkkk1}\end{align}
 We claim that $\mathcal{U}_2(R)=\infty$ for this example; therefore, \eqref{eqnkkkk1} does not yield any meaningful bound.  To see this, note that 
\[\mathrm{gen}_{\mu}(\hat w,z)=\sigma^2d+\|\beta-\hat w\|^2-\|\hat w-z\|^2=\sigma^2d+\|\beta\|^2-\|z\|^2+2\hat w^T(z-\beta).
\]
Take some constant $c$ and let $\hat W=c(Z-\beta)+T$ where $T\sim\mathcal{N}(0,\tilde{\sigma}^{2}I_{d})$ is independent of $Z$. For every $c$, by letting $\tilde{\sigma}^{2}\rightarrow\infty$ we can ensure that $I(\hat W;Z)\leq R$. Next, by letting $c\rightarrow\infty$ we get that $\mathbb{E}[\mathrm{gen}_\mu(\hat W,Z)]\rightarrow\infty$. Note that in this construction, the variance of $\hat W$ tends to infinity. 


To sum this up, assuming just the knowledge of $I(W;Z_{i})$, the bound in \eqref{bnu2} does not result in any meaningful upper bound on the generalization error. If we additionally know that 
$\ell(W,Z)$ is $\sigma^{2}$-sub-Gaussian under $P_{W}P_{Z}$, Bu et al \cite{bu2020tightening} obtained a meaningful bound, but with the suboptimal order $1/\sqrt{n}$. 

One idea is to assume more knowledge about the algorithm than just $I(W;Z_{i})$. This knowledge can be incorporated as a constraint on the algorithm as follows in order to get a nontrivial bound:
\begin{align}
    \sup_{\substack{P_{\hat W|Z}:~I(\hat W;Z)\leq R,\\{\text{\& another constraint}}}}
\mathbb{E}\left[\text{\rm{gen}}_{\mu}(Z,\hat W)\right]\label{addconstt}
\end{align}
where the additional information is reflected in the constraint that restricts the domain of the optimization problem above.
For the ERM algorithm, we know that \[\mathsf{Var}(W_{ERM})=\mathbb{E}[\|W_{ERM}-\beta\|^2]=\frac{\mathsf{Var}(Z)}{n}=\frac{\sigma^2d}{n}\]
where the variance of a vector is defined as the trace of its covariance matrix.\footnote{Here, we used the fact that $\mathbb{E}[W_{ERM}]=\beta$ as ERM is an unbiased estimate of the mean.}
If we restrict to algorithms whose variance is at most $c\sigma^2d/n$ for some constant $c$, the upper bound
$
\U_2\left(\frac{d}{2}\ln\frac{n}{n-1}\right)$ in \eqref{eqnkkkk1}
can be replaced by 
$
\hat{\U}_2\left(\frac{d}{2}\ln\frac{n}{n-1}\right)$ where
\begin{align}
\hat{\U}_2(R):= \sup_{\substack{P_{\hat W|Z}:~I(\hat W;Z)\leq R,\\ \mathsf{Var}(\hat{W})\leq c\sigma^2d/n}}
\mathbb{E}\left[\text{\rm{gen}}_\mu(Z,\hat W)\right].\label{qn3jkd}
\end{align}
We show in Appendix \ref{rate_of_consistency} that \[
\hat{\U}_2\left(\frac{d}{2}\ln\frac{n}{n-1}\right)=\mathcal{O}\left(\frac1n\right)\] Therefore, the upper bound is order optimal. In fact, 
if we restrict the domain of supremum in \eqref{qn3jkd} by imposing $\mathsf{Var}(\hat{W})\leq \dfrac{\mathsf{Var}(Z)}{n}=\dfrac{\sigma^2d}{n}$ (i.e., setting $c=1$), we get the exact value of the generalization error in \eqref{eqnkkkk2}! 
\begin{remark} In \cite{sefidgaran2022rate}, the authors get the order $1/n$ by assuming that $\ell(\cdot,z)$ is Lipschitz for every $z\in\mathcal{Z}$ i.e., $|\ell(w,z)-\ell(w',z)|\le L\|w-w'\|$ for every $z\in\mathcal{Z}$ and $w,w'\in\mathcal{W}$ (in addition to a $\sigma$-sub-Gaussian assumption on the loss function). However, note that $\ell(w,z)=\|w-z\|^{2}$ in our example is not Lipschitz. For the ERM algorithm, we should note that techniques based on VC-dimension \cite{boucheron2005theory} (for binary loss function) and algorithmic stability \cite[Example 3]{bousquet2002stability} (for the least square loss on bounded hypothesis space) also yield bounds of $\mathcal{O}(1/\sqrt{n})$.
\end{remark}

The above example illustrates that one idea to improve the rate-distortion bound on the generalization error is through restricting the domain of $\U_1(\cdot)$ and $\U_2(\cdot)$ by adding other convex constraints to the set of randomized learning algorithms $P_{W|S}$. In Appendix \ref{further-ideas} we explore two ideas along this line. In particular, an application of supermodular $\mf$-divergences is discussed.

\subsubsection{Related work}
As discussed above, one of our contributions is the novel connection between the generalization error and the rate-distortion theory. This connection is at a formal level and follows from the similarity of the expression of the generalization error and that of the rate-distortion function. 

Note that the lossy source coding problem can be understood as the vector quantization problem which is essentially a clustering task (an unsupervised learning problem). Connections between the rate-distortion function and generalization error in supervised learning algorithms are reported in \cite{bu2021population}, \cite{sefidgaran2022rate} and \cite{xu2022minimum}. However, our approach is different from these works.
In \cite{bu2021population}, the authors utilize the rate-distortion function to find bounds on the generalization error of the algorithm after model compression. This suggests that model compression can be interpreted as a regularization technique to avoid overfitting.
In \cite{sefidgaran2022rate}, the authors derive new upper bounds on the generalization error (in-expectation and tail bound). In particular, by utilizing a different auxiliary algorithm $\hat W$, they derive the following upper bound on the generalization error of $P_{W|S}$:
\begin{align}\label{rate-dist-upper-bound}
    \text{\rm{gen}}(\mu, P_{W|S}) \leq \sqrt{\frac{2\sigma^2\mathcal{R}({D})}{n}}+{D}.
\end{align}
where $\mathcal{R}({D}) = \inf \limits_{P_{\hat{W}|S}:\mathbb{E} \left[d(W,\hat{W};S)\right]\leq {D}}I(S;\hat{W})$ and $d(w,\hat{w};s)=\text{\rm{gen}}_{\mu}(w,s)-\text{\rm{gen}}_{\mu}(w',s)$. While the mutual information upper bound in Theorem \ref{th1_without_mis} is infinite if $W$ is a deterministic function of $S$ (deterministic algorithm) and $W$ is a continuous random variable, the bound in \eqref{rate-dist-upper-bound} is a finite upper bound in this case for $D>0$.
Finally, in \cite{xu2022minimum}, the authors analyze the performance of Bayesian learning under generative models by
defining and upper-bounding the minimum excess risk. Minimum excess risk (MER) can be viewed as a rate-distortion problem with distortion measure between two decision rules \cite{hafez2021rate}.


\subsubsection{Computational aspects of the rate-distortion bound}\label{NewSecComputational}
Assuming that an upper bound $R$ on the $I\big(S;\A(S)\big)$ is known, computing the upper bound $\U_1(R)$ is still a concern if the sample size $n$ is large despite the fact that computing $\U_1(R)$ is a convex optimization problem and there are efficient algorithms for solving it for discrete sample spaces \cite{blahut1972computation}. The following theorem addresses this computational concern by providing an upper bound in a ``single-letter" form:
\begin{theorem}\label{Th2_without_mismatch} For any arbitrary loss function $\ell(w,z)$, and algorithm $\A$ that runs on a training dataset $S$ of size $n$, we have
\begin{align}\mathcal{U}_1\left(R\right)\le \mathcal{U}_2\left(\frac{R}{n}\right)\label{eqnNNN}
\end{align}
where $\mathcal{U}_2(R)$ is defined in \eqref{eqnRDN22}. Consequently,
\[\mathrm{gen}\left(\mu,\A\right)\le \mathcal{U}_2\left(\frac{I(S;\A(S))}{n}\right)
\]
Furthermore, to compute the maximum in \eqref{eqnRDN22}, it suffices to compute the maximum over all conditional distributions $P_{\hat W|Z}$ for $\hat{W}\in \mathcal{W}$ such that the support of $\hat W$ has size at most $|\mathcal{Z}|+1$.
\end{theorem}

\begin{remark}
In Appendix \ref{further-ideas}, we show that ``auxiliary loss functions" can be utilized to tighten the gap between $\mathcal{U}_{1}(R)$ and $\mathcal{U}_{2}(R/n)$.
\end{remark}

 A generalized version of Theorem \ref{Th2_without_mismatch} is given later in Theorem \ref{Th_gen_Uf_If}. Therefore, a separate proof is not given for the above theorem here.

Note that $\mathcal{U}_2(R)$ is still in terms of a rate-distortion function and does not admit an explicit closed-form expression in general. However, the Blahut-Arimoto algorithm can be used to compute it \cite{blahut1972computation} even when the cardinality of instance space $\mathcal{Z}$ is  infinite (see also \cite{blahut1972computation2}).\footnote{
Rate-distortion theory for continuous or abstract alphabets is discussed at length in the literature, e.g. see 
\cite{csiszar1974extremum, rose1994mapping}. See also  \cite{berger1998lossy} for a survey.} It is shown in \cite{arimoto1972algorithm} that Blahut-Arimoto algorithm in discrete case has only two kinds of convergence speeds depending on the source distribution. One is the exponential convergence, which is a fast convergence, and the other is the convergence of order $O(1/N)$ where $N$ is number of the iterations.

Note that the bound in Theorem \ref{th1_without_mis} is in a very explicit form.
Moreover, the bound in Theorem \ref{th1_without_mis} depends only on mutual information $I\big(S;\A(S)\big)$ while the bounds in Theorem \ref{Th2_without_mismatch}  depends on  $\mu$, and $I\big(S;\A(S)\big)$ (as $\mathcal{U}_1(\cdot)$ and $\mathcal{U}_2(\cdot)$ depend on $\mu$). However, one can obtain a bound from Theorem \ref{Th2_without_mismatch} that does not depend on  $\mu$ by maximizing the bound in Theorem \ref{Th2_without_mismatch} over all distributions $\mu$. We show that even after this maximization, the bound in Theorem \ref{Th2_without_mismatch} is still an improvement over Theorem \ref{th1_without_mis}. To see this, first take some arbitrary $\mu$ such that $\ell(w,Z)$ is $\sigma^{2}$-sub-Gaussian  for every $w\in\mathcal{W}$ under the distribution $\mu$ on $Z$. Then, observe that the bound in Theorem \ref{Th2_without_mismatch} is always less than or equal to the bound in Theorem \ref{th1_without_mis} since $
	 \mathcal{U}_2(R)
		\le	\sqrt{2\sigma^{2}R}$. On the other hand, the following examples show that the maximum of 
		$\mathcal{U}_2\left(R\right)$ over all distributions $\mu$ with the $\sigma^{2}$-sub-Gaussian property  can be strictly less than the bound in Theorem \ref{th1_without_mis}.

\begin{figure}
	   \centering
	\includegraphics[scale=1,width=0.95\linewidth]{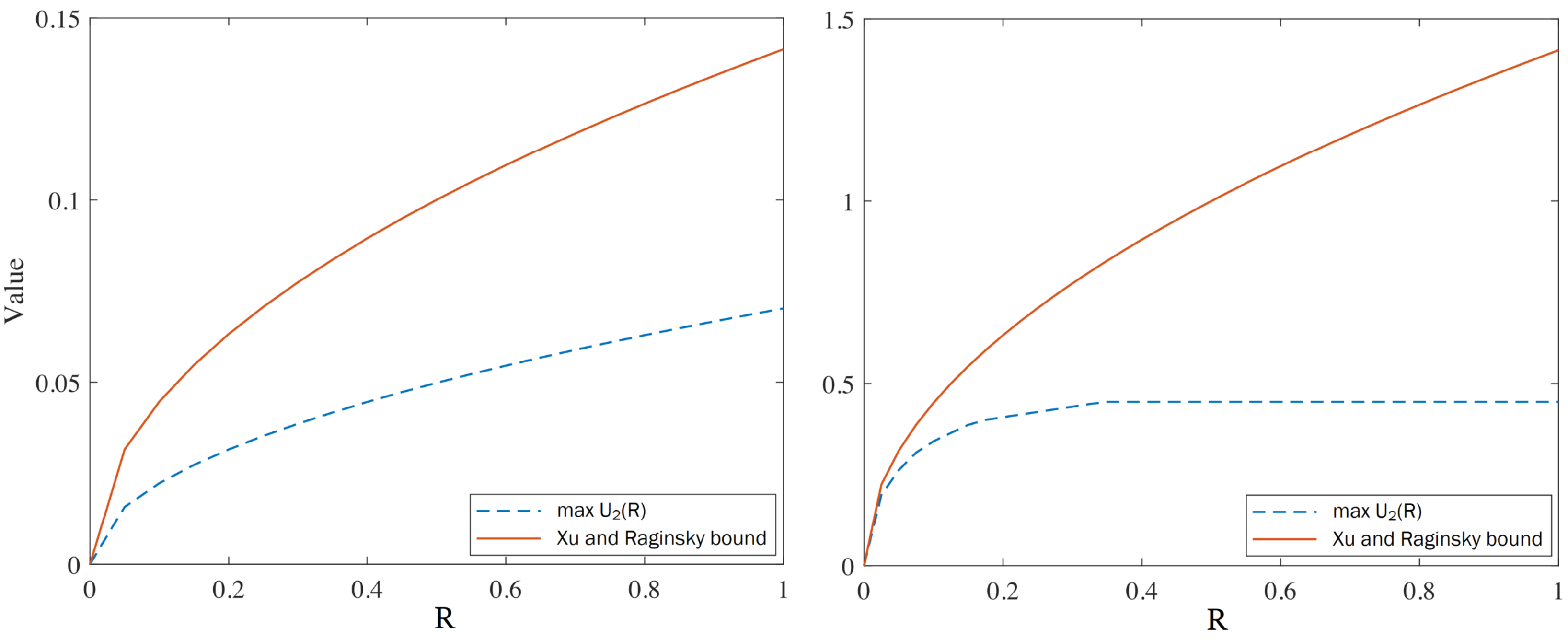}
	\caption{Left picture: The bound in Theorem \ref{th1_without_mis} versus the maximum of the upper bound in Theorem \ref{Th2_without_mismatch} over all distributions $\mu$ when $\mathcal{W}=\mathcal{Z}=\{0,1\}$ and  $\ell(w,z)=w\cdot z$. Right picture: The bound in Theorem \ref{th1_without_mis} versus the maximum of the upper bound in Theorem \ref{Th2_without_mismatch} over all distributions $\mu$ when $\mathcal{W}=[0,1],\,\mathcal{Z}=\{0,1\}$ and  $\ell(w,z)=|w-z|$.}
	\label{fig:screenshot001}
\end{figure}
\begin{example}\label{ex1}
Figure \ref{fig:screenshot001} depicts the bound  in Theorem \ref{th1_without_mis} versus the maximum of the bound in Theorem \ref{Th2_without_mismatch}  over all distributions $\mu$ on $\lbrace0,1\rbrace$ for two particular loss functions. Note that the distortion function itself depends on the choice of $\mu$ and this makes it difficult to find a closed form expression for the maximum of the bound in Theorem \ref{Th2_without_mismatch} over all distributions $\mu$. In the left image, we consider $\mathcal{W}=\mathcal{Z}=\lbrace0,1\rbrace$ and a learning problem on a data set $S$ with the size $n=1$ with loss function $\ell(w,z)=w\cdot z$. In the right picture, we consider $\mathcal{W}=[0,1],\,\mathcal{Z}=\lbrace0,1\rbrace$ and a learning problem on a data set $S$ with the size $n=1$ with loss function $\ell(w,z)=|w-z|$.
\end{example}

\subsection{Generalization error for algorithms with bounded input/output mutual $\mf$-information}
Theorem \ref{th1_without_mis} shows that 
the generalization error of a learning algorithm can be bounded from above in terms of the mutual information between the input and output of the algorithm \cite{russo2019much,xu2017information}. Similar bounds are obtained in \cite{bu2020tightening,lopez2018generalization,wang2019information,hellstrom2020generalization,aminian2020jensen,esposito2020robust,esposito2019generalization,jiao2017dependence,asadi2018chaining} for various generalizations and extensions using other measures of dependence. 

In this paper we are interested in the generalization error using the mutual $\mf$-information instead of Shannon's mutual information. As before, the \emph{sharpest} possible upper bound on the generalization error given an upper bound $R$ on $I_{\mf}(S;\A(S))$ is
 \begin{align}\label{eqnRDf1}
     \U^{I_{\mf}}_{1}(R)\triangleq\sup_{P_{W|S}:~I_{\mf}(S;W)\le R}\mathbb{E}[\text{\rm{gen}}_\mu(S,\hat W)],
 \end{align}
 where $\text{\rm{gen}}_\mu(s,w)$ is defined in \eqref{gen_gap} and 
 the training data $S=(Z_{1},Z_{2},\cdots,Z_{n})\in\mathcal{Z}^n$
 is a sequence of $n$ i.i.d.\ repetitions of samples generated according to a distribution $\mu$. It follows that for any algorithm $\A$ satisfying $I_{\mf}\big(S;\A(S)\big)\leq R$, we have
\[
\mathrm{gen}\left(\mu,\A\right)
\leq {\U}^{I_{\mf}}_{1}(R).\]

Similarly, assuming $I_\mf(Z_i;W)\le R_{i}$ for $1\leq i\leq n$, we obtain the following bound 
\begin{align}
\mathrm{gen}\left(\mu,\A\right)
\leq \frac1n\sum_{i=1}^n{\U}^{I_{\mf}}_{2}(R_{i}),\label{eqnU2fnew}\end{align}
where
\begin{align}\label{eqnRDf2}
      \U^{I_{\mf}}_{2}(R)\triangleq\sup_{P_{\hat W|Z}:~I_{\mf}(Z;\hat W)\le R}\mathbb{E}[\text{\rm{gen}}_\mu(Z,\hat W)]
\end{align}
where  where $\text{\rm{gen}}_\mu(z,w)$ is defined in \eqref{gen_gap2}.

As before, the bound  $\U^{I_{\mf}}_{2}(R)$ is easier to compute than $\U^{I_{\mf}}_{1}(R)$ because the optimization problem in \eqref{eqnRDf2} is for a single symbol $Z$ whereas the  optimization problem in \eqref{eqnRDf1} is for a sequence $S$ of $n$ symbols.

The following bound follows from the characterization in \eqref{eqnRD2} and Theorem \ref{cor_f-dist_lower_bound} for $R=I_\mf(W;S)$.
\begin{corollary}\label{U1_explicit_upper_bound}
We have
\begin{align}
   \U^{I_{\mf}}_{1}(R)\le \Psi^{\mf^{*}}_{L_{S}(\cdot)}(R)
\end{align}
where $\mf^*$ is the convex conjugate function of $\mf$ defined in \eqref{eqnfs}
and
\[
\Psi^{\mf^{*}}_{L_{S}(\cdot)}(x)\triangleq \inf_{\lambda\ge0,\,a\in\mathbb{R}}\left\lbrace\frac{a+x}{\lambda}+\frac{1}{\lambda}\sup_{w\in\mathcal{W}}\mathbb{E}_{P_{S}}\left[\mf^{*}\left(\lambda L_{S}(w)-\lambda L_{\mu}(w)-a\right)\right]\right\rbrace.
\]
In this corollary $I_{\mf}$ can be either the PV or CKZ notions defined in \eqref{IfPV} and \eqref{IfCKZ}.
\end{corollary}
\begin{example}
Let $\mf(x)=x\ln x-x+1$ and $\ell(w,Z)$ be $\sigma^{2}$-sub-Gaussian under $P_{Z}$ for all $w\in\mathcal{W}$. Then $\mf^{*}(y)=e^{y}-1$ and
\begin{align}
&\Psi^{\mf^{*}}_{L_{S}(\cdot)}(I(S;W))= \inf_{\lambda\ge0,\,a\in\mathbb{R}}\left\lbrace\frac{a+I(S;W)}{\lambda}+\frac{1}{\lambda}\sup_{w\in\mathcal{W}}\mathbb{E}_{P_{Z}}\left[\exp\left(\lambda L_{S}(w)-\lambda L_{\mu}(w)-a\right)-1\right]\right\rbrace\nonumber\\
&=\inf_{\lambda\ge0}\left\lbrace \frac{I(S;W)}{\lambda}+\frac{1}{\lambda}\ln\sup_{w\in\mathcal{W}}\mathbb{E}_{P_{Z}}\left[\exp\left(\lambda L_{S}(w)-\lambda L_{\mu}(w)\right)\right]\right\rbrace\nonumber\\
&\le \inf_{\lambda\ge0}\left\lbrace \frac{I(S;W)}{\lambda}+\frac{\lambda\sigma^{2}}{2n}\right\rbrace=\sqrt{\frac{2\sigma^{2}I(S;W)}{n}},
\end{align}
recovering Theorem \ref{th1_without_mis}.
\end{example}
Thus, from Corollary \ref{U1_explicit_upper_bound} that 
\begin{align}
   \text{\rm{gen}}(\mu,P_{W|S})\le \Psi^{\mf^{*}}_{L_{S}(\cdot)}(I_{\mf}(S;W))\label{eqnpsimupw}
\end{align}

Similarly, from \eqref{eqnU2fnew} and Theorem \ref{cor_f-dist_lower_bound} we obtain
\begin{corollary}\label{th27}
Let $\mf:[0,\infty)\to \mathbb{R}\cup\{\infty\}$ be a convex function.
Then
\begin{align}
\text{\rm{gen}}(\mu,P_{W|S})\le \frac{1}{n}\sum_{i=1}^{n}\U^{I_{\mf}}_{2}(I_{\mf}(Z_{i};W))\le \frac{1}{n}\sum_{i=1}^{n}\Psi^{\mf^{*}}_{\ell(\cdot,Z)}(I_{\mf}(Z_{i};W))\label{f_mutual_upperbound_gen_individual}
\end{align}
where
\begin{align}\label{def_psi}
    \Psi^{\mf^{*}}_{\ell(\cdot,Z)}(x)\triangleq\inf_{\lambda\ge0,\,a\in\mathbb{R}}\left\lbrace\frac{a+x}{\lambda}+\frac{1}{\lambda}\sup_{w\in\mathcal{W}}\mathbb{E}_{P_{Z}}\left[\mf^{*}\left(\lambda\ell(w,Z)-\lambda\mathbb{E}_{P_{Z}}\ell(w,Z)-a\right)\right]\right\rbrace.
\end{align}
In this corollary $I_{\mf}$ can be either the PV or CKZ notions defined in \eqref{IfPV} and \eqref{IfCKZ}.
\end{corollary}

Corollary \ref{th27} can be used to derive explicit bounds on the generalization error in the following example.

\begin{example}
\begin{enumerate}[(i).]
    
    \item Assume $\mf(x)=(x-1)^{2}$ for $x\in [0,\infty)$. Assume $\ell(w,Z)$ is a random variable with finite variance $Var(\ell(w,Z))\le \sigma^{2}$ for all $w\in\mathcal{W}$. Then,
\begin{align*}
    \text{\rm{gen}}(\mu,P_{W|S})\le \frac{1}{n}\sum_{i=1}^{n}\sqrt{\sigma^{2}\chi^{2}(P_{WZ_{i}}\|P_{W}P_{Z_{i}})}\le \sqrt{\frac{\sigma^{2}}{n}\chi^{2}(P_{WS}\|P_{W}P_{S})}.
\end{align*}
where $\chi^{2}(\cdot\|\cdot)$ is defined in \eqref{def:chi-squared}.
\item Assume $\ell(w,Z)$ is a $\sigma^{2}$-sub-Gaussian random variable under $P_{Z}$ for all $w\in\mathcal{W}$ and $\mf(x)=(\alpha+\bar{\alpha}x)\ln\left(\alpha+\bar{\alpha}x\right)$ defined on the domain $\mathbb{R}_{+}$ for $\alpha\in[0,1]$ and $\alpha+\bar{\alpha}=1$. Then
\begin{align}
     \text{\rm{gen}}(\mu,P_{W|S})&\le \frac{1}{n}\sum_{i=1}^{n}\sqrt{\frac{2\sigma^{2}}{\bar{\alpha}^{2}}D\left(\bar{\alpha} P_{WZ_{i}}+\alpha P_{W}P_{Z_{i}}\|P_{W}P_{Z_{i}}\right)}\nonumber\\
     &\le \sqrt{\frac{2\sigma^{2}}{\bar{\alpha}^{2}n}D\left(\bar{\alpha} P_{WS}+\alpha P_{W}P_{S}\|P_{W}P_{S}\right)} ,\quad\forall \alpha\in[0,1].
\end{align}
\end{enumerate}
\label{claim_example_main_theorem}
\end{example}
Proof of Example \ref{claim_example_main_theorem} is given in Section \ref{proof_claim_example_main_theorem}. The bound 
\begin{align*}
    \text{\rm{gen}}(\mu,P_{W|S})\le \sqrt{\frac{\sigma^{2}}{n}\chi^{2}(P_{WS}\|P_{W}P_{S})}.
\end{align*}
in part (i) may be compared to the Xu-Raginsky bound \cite{xu2017information} (Theorem 
\ref{th1_without_mis}):
\begin{align*}
    \text{\rm{gen}}(\mu,P_{W|S})\le \sqrt{\frac{\sigma^{2}}{n}D_{KL}(P_{WS}\|P_{W}P_{S})}.
\end{align*}
 The assumption of $\sigma^2$-subgaussian is replaced by a weaker variance assumption. However, the bound we obtain is in terms of $\chi^2$-information rather than Shannon's mutual information. The $\chi^2$-information is always greater than or equal to Shannon's mutual information and is utilized in statistics (e.g. in the chi-square test of independence).

{\color{black}
\subsection{Upper bound on the generalization error using $\mf$ in the function class $\mathcal{F}$}\label{Orlicz_norm_based_bounds}
The following theorem generalizes Theorem \ref{Th2_without_mismatch} and is helpful from a computational perspective.
\begin{theorem}\label{Th_gen_Uf_If}
Let $\mf\in\mathcal{F}$. Then for the CKZ and MBGYA notions of mutual $\mf$-information, we have
\begin{align}\label{single_letter_I_f_eq}
   \U^{I_{\mf}}_{1}(R)\le \U^{I_{\mf}}_{2}(R/n), \qquad\forall R\geq 0.
\end{align}
Consequently, for any algorithm $\A$ we have
\[
\mathrm{gen}\left(\mu,\A\right)
\leq \U^{I_{\mf}}_{2}\left(\frac{I_{\mf}\big(S;\A(S)\big)}{n}\right).\]
Furthermore, to compute the maximum in $\U^{I_{\mf}}_{2}$ for the CKZ notion of the $\mf$-information, it suffices to compute the maximum over all conditional distributions $P_{\hat W|Z}$ for $\hat{W}\in \mathcal{W}$ such that the support of $\hat W$ has size at most $|\mathcal{Z}|+1$.

\end{theorem}
The proof is given in Section \ref{Proof_Th_gen_Uf_If}. 

For a finite hypothesis space $\mathcal{W}$, one can find an explicit upper bound on the generalization error that holds \emph{for any arbitrary learning algorithm} and depends only on the sample size $n$ as follows:
\begin{corollary}\label{Proposition: decay rate K_alpha} Assume that $\mf\in \mathcal{F}$. Then
\[
\mathrm{gen}\left(\mu,\A\right)\le \Psi^{\mf^{*}}_{\ell(\cdot,Z)}\left(\frac{H_{\mf}(W_{unif})}{n}\right)
\]
where $H_{\mf}$ is defined in Definition \ref{def_f_entropy} and $\Psi^{\mf^{*}}_{\ell(\cdot,Z)}(x)$ is defined in \eqref{def_psi}, and $W_{unif}$ is a random variable with uniform distribution on $\mathcal{W}$.
\end{corollary}
\begin{proof}
From Theorem \ref{Th_gen_Uf_If}, we get
\begin{align}
    \mathrm{gen}\left(\mu,\A\right)&\le \U^{I_{\mf}}_{1}(I_{\mf}(S;\A(S)))\le \U^{I_{\mf}}_{2}(I_{\mf}(S;\A(S))/n)\nonumber\\
    &\overset{(a)}{\le} \U^{I_{\mf}}_{2}(H_{\mf}(\A(S))/n)\overset{(b)}{\le}\U^{I_{\mf}}_{2}(H_{\mf}(W_{unif})/n)
\end{align}
(a) comes from Part (iii) of Theorem \ref{property_f_entropy} and (b) comes from Part (ii) of Theorem \ref{property_f_entropy}.
\end{proof}
\begin{example}
Let $\mf(x)=(\alpha+(1-\alpha)x)\ln(\alpha+(1-\alpha)x)$ defined on the domain $\mathbb{R}_{+}$ for $\alpha\in[0,1]$. We know that $\mf\in\mathcal{F}$. Let $\ell(w,Z)$ be $\sigma^{2}$-sub-Gaussian under $P_{Z}$ for $w\in\mathcal{W}$. So from Proposition \ref{Proposition: decay rate K_alpha}
\[
\mathrm{gen}\left(\mu,\A\right)\le \Psi^{\mf^{*}}_{\ell(\cdot,Z)}\left(\frac{H_{\mf}(W_{unif})}{n}\right)=\sqrt{\frac{2\sigma^{2}}{(1-\alpha)^{2}}\frac{H_{\mf}(W_{unif})}{n}}
\]
where in this example
\[
H_{\mf}(W_{unif})=\left(1-\frac{1}{|\mathcal{W}|}\right)\alpha\ln\alpha+\left(\frac{\alpha}{|\mathcal{W}|}+1-\alpha\right)\ln(\alpha+(1-\alpha)|\mathcal{W}|).
\]
In particular, when $\alpha=0$, we get $\mathrm{gen}\left(\mu,\A\right)\le \sqrt{2\sigma^{2}\frac{\ln|\mathcal{W}|}{n}}$. However, one can minimize 
\[\frac{H_{\mf}(W_{unif})}{(1-\alpha)^{2}}
\]
over $\alpha$ and get strictly better bounds. The minimum of the above expression occurs at $\alpha^{*}=1/(|\mathcal{W}|-1)$ and substituting $\alpha^{*}$, we get
\[
\mathrm{gen}\left(\mu,\A\right)\le \sqrt{2\sigma^{2}\frac{(|\mathcal{W}|-1)^{2}}{|\mathcal{W}|\cdot(|\mathcal{W}|-2)}\frac{\ln(|\mathcal{W}|-1)}{n}}
\]
slightly improving over $\sqrt{2\sigma^{2}\frac{\ln|\mathcal{W}|}{n}}$.

\end{example}

\section{Future work}
The following questions are left for future work:
\begin{itemize}
    \item Proposition \ref{prop_super-modular_D_f} shows that
$D_\mf$ is super-modular for any function $\mf\in\mathcal{F}$. Is the converse true? One can show that $D_\mf$ is super-modular if and only if for any product distribution $q_{X_1}q_{X_2}$ and any function $Z(x_1,x_2)$ satisfying $\mathbb{E}_{q}[Z]=1$ we have
\begin{align}
&\mathbb{E}_{X_1,X_2}\left[\mathsf f\left(Z\right)\right]
\geq
\mathbb{E}_{X_1}\left[\mathsf f\left(\mathbb{E}_{X_2}Z\right)\right]+\mathbb{E}_{X_2}\left[\mathsf f\left(\mathbb{E}_{X_1}Z\right)\right].\label{erjklnn}
\end{align}
\item It is known that  $I_\mf^{PV}$ has applications in hypothesis testing and channel coding  \cite[Theorem 8]{polyanskiy2010arimoto}. Can we find similar applications for $I_\mf^{MBGYA}$, perhaps to bound the exponent of certain conditional events?

\item In Theorem \ref{new properties_I_f}, it is shown that $I^{CKZ}$ and $I^{\text{MBGYA}}$ satisfy $AB$-property mentioned in the introduction. Does $I^{PV}_\mf(X;Y)$ satisfy the $AB$-property? The proof does not seem to go through.

\item To improve the rate-distortion bound on the generalization error, we discuss a number of ideas in Appendix \ref{further-ideas}. Further exploration of these ideas is left as future work.
\end{itemize}




}


\bibliographystyle{IEEEtran}
\bibliography{reference}

\section{Proofs of the results}
\label{sec:proofs}
In the following sections, we present the proofs of the results stated in the previous sections in their order of appearance. 
\subsection{Proof of Proposition \ref{prop_super-modular_D_f}}\label{proof_prop_super_modular_D_f}

It is known that $\mf$-divergence is related to $\Phi$-entropy: given two distributions $P$ and $Q$, let $g(x)=\frac{dP}{dQ}(x)$. Then,
\begin{align}
D_\mf(P\|Q)&=\mathbb{E}_{Q}\left[\mf\left(\frac{dP}{dQ}(X)\right)\right]
\\&=\int \mf\left(g(x)\right)dQ(x)-\mf(1)
\\&=\int \mf\left(g(x)\right)dQ(x)-\mf\left(\int g(x)dQ(x)\right):=\mathbb{H}_{\mf}(g(X)).
\end{align}
Let $\Phi(x)=\mf(x)$. Then, $D_\mf(P\|Q)$ is the $\Phi$-entropy of the function $g(x)$. 

In \cite[Theorem 14.1]{boucheron2013concentration}, the authors show that for $\Phi\in\mathcal{F}$ and mutually independent random variables $X_1, X_2, \cdots, X_n$, the $\Phi$-entropy becomes subadditive, i.e. for any arbitrary function $r(x_1,\cdots, x_n)$ we have
\begin{align}\label{phi-ent_tensor}
    \mathbb{H}_{\Phi}(r(X_{1},\ldots,X_{n}))\le \mathbb{E}\sum_{i=1}^{n}\mathbb{H}^{(i)}_{\Phi}(r(X^{\setminus i}))
\end{align}
where $\mathbb{H}^{(i)}_{\Phi}(r(X^{\setminus i}))=\mathbb{E}^{(i)}[\Phi(Z)]-\Phi(\mathbb{E}^{(i)}[Z])$ and $Z=r(X_{1},\ldots,X_{n})$ and $\mathbb{E}^{(i)}$ denotes conditional expectation conditioned on the $(n-1)$-vector $X^{\setminus i}=(X_{1},\ldots, X_{i-1},X_{i+1},\ldots,X_{n})$. 
To obtain \eqref{def-sub-mod},  we proceed as follows: we write \eqref{phi-ent_tensor} for $\Phi(x)=\mf(x)$, $n=2$, and $r_{x_{3}}(x_{1},x_{2})=\frac{dP_{X_{1},X_{2},X_{3}}}{dQ_{X_{1}}dQ_{X_{2}}dQ_{X_{3}}}(x_{1},x_{2},x_{3})$. By averaging over $x_3$, we obtain:
\begin{align}\mathbb{E}_{Q_{X_{3}}}[\mathbb{H}_{\Phi}(r_{X_{3}}(X_{1},X_2))]\le \mathbb{E}_{Q_{X_{1}}Q_{X_{3}}}[\mathbb{H}^{(1)}_{\Phi}(r_{X_{3}}(X^{\setminus 1}))]+\mathbb{E}_{Q_{X_{2}}Q_{X_{3}}}[\mathbb{H}^{(2)}_{\Phi}(r_{X_{3}}(X^{\setminus 2}))].\label{eqnendg}\end{align}
Next, note that
\begin{align}
   &\mathbb{E}_{Q_{X_{3}}}[ \mathbb{H}_{\Phi}(r_{X_{3}}(X_{1},X_2))]\nonumber\\
   &=\int \Phi\left(\frac{dP_{X_{1},X_{2},X_{3}}}{dQ_{X_{1}}dQ_{X_{2}}dQ_{X_{3}}}(x_{1},x_{2},x_{3})\right)dQ_{X_{1}}dQ_{X_{2}}dQ_{X_{3}}-\int\Phi\left(\frac{dP_{X_{3}}}{dQ_{X_{3}}}(x_{3})\right)dQ_{X_{3}}\nonumber\\
    &= D_{\Phi}(P_{X_{1},X_{2},X_{3}}\|Q_{X_{1}}Q_{X_{2}}Q_{X_{3}})-D_{\Phi}(P_{X_{3}}\|Q_{X_{3}}).
  \end{align} 
  Next,
\begin{align}&\mathbb{E}_{Q_{X_{1}}Q_{X_{3}}}[\mathbb{H}^{(1)}_{\Phi}(r_{X_{3}}(X^{\setminus 1}))]=\int \Phi\left(\frac{dP_{X_{1},X_{2},X_{3}}}{dQ_{X_{1}}dQ_{X_{2}}dQ_{X_{3}}}(x_{1},x_{2},x_{3})\right)dQ_{X_{1}}(x_{1})dQ_{X_{2}}(x_{2})dQ_{X_{3}}(x_{3})\nonumber\\
   &\hspace{5cm}-\int \Phi\left(\frac{dP_{X_{1}X_{3}}}{dQ_{X_{1}}dQ_{X_{3}}}(x_{1},x_{3})\right)  dQ_{X_{1}}dQ_{X_{2}}\nonumber\\
    &=D_{\Phi}(P_{X_{1},X_{2},X_{3}}\|Q_{X_{1}}Q_{X_{2}}Q_{X_{3}})- D_{\Phi}(P_{X_{1}X_{3}}\|Q_{X_{1}X_{3}})
\end{align}
and similarly $\mathbb{E}_{Q_{X_{2}}Q_{X_{3}}}[\mathbb{H}^{(2)}_{\Phi}(r_{X_{3}}(X^{\setminus 2}))]=D_{\Phi}(P_{X_{1},X_{2},X_{3}}\|Q_{X_{1}}Q_{X_{2}}Q_{X_{3}})- D_{\Phi}(P_{X_{2}X_{3}}\|Q_{X_{2}X_{3}})$. Therefore, \eqref{eqnendg} reduces to  \eqref{def-sub-mod}.

\subsection{Proof of Lemma \ref{newlmma3}}\label{section_class_F_proofs}

For every $\mathsf f\in\mathcal{F}$ defined on $[0,\infty)$ we have $\mathsf f^{''}(t)+t\mathsf f^{'''}(t)\ge 0$ for any $t\geq 0$ \cite[Exercise 2.3.29]{borwein2010convex}. Hence $t\to t\mathsf f^{''}(t)$ is a non-decreasing function.
From $\mf\in\mathcal{F}$, $\mf''(t)> 0$. Moreover, since $\mf$ is not a linear function,  $t\mathsf f''(t)$ is not zero for all $t$. Thus, $\lim_{t\to \infty}t\mathsf f^{''}(t)>0$. Also, from \cite[Exercise 2.3.29]{borwein2010convex},  we obtain that $\mathsf f^{'''}(x)\le 0$. Hence $\mathsf f^{''}$ is a non-increasing function.

Since $t\to t\mathsf f^{''}(t)$ is a non-decreasing function, we obtain
$t\mathsf f^{''}(t)\geq \mf''(1)$. 
 Therefore,
    \[\mf''(t)\ge\frac{\mf''(1)}{t}\qquad\forall t\geq 1.\]
    Hence,
    \[\mf'(t)\ge\mf'(1)+\int_{1}^{t}\frac{\mf''(1)}{\tau}d\tau.\]
This shows that $\lim_{t\rightarrow\infty}\mf'(t)=\infty$. This, in turn, shows that $\lim_{t\rightarrow\infty}\mf(t)=\infty$.

To show \eqref{const1tlogt} using L'Hospital's rule we have
\[
\lim_{t\to \infty}\frac{\mathsf f(t)}{t\ln t}=
\lim_{t\to \infty}\frac{\mathsf f(t)}{t\ln t-t}=\lim_{t\to \infty}\frac{\mathsf f^{'}(t)}{\ln t}=\lim_{t\to \infty}t\mathsf f''(t).
\]

To show \eqref{const2t2}, using L'Hospital's rule
\[\lim_{t\to \infty}\frac{t^2}{\mathsf f(t)}=
\lim_{t\to \infty}\frac{2t}{\mathsf f'(t)}=
\lim_{t\to \infty}\frac{2}{\mathsf f''(t)}.
\]
\subsection{Proof of the claim in Remark \ref{remark 004}}\label{proof_remark 004}
The exponent of Sanov's bound can be simplified as follows:
\begin{align}
    E_{SA}&=\frac{1}{n}\inf_{\substack{Q_{X,W}:\\\mathbb{E}_{Q_{X,W}}[\ell(X,W)]\ge \delta}}\{nD(Q_{X,W}\|P_{X}Q_{W})-(n-1)D(Q_{W}\|P_{W})\}\nonumber\\
    &=\frac{1}{n}\inf_{Q_{X,W}}\sup_{\lambda\ge 0}\lambda(\delta-\mathbb{E}_{Q_{X,W}}[\ell(X,W)])+nD(Q_{X,W}\|P_{X}Q_{W})-(n-1)D(Q_{W}\|P_{W})\nonumber\\
    &=\frac{1}{n}\sup_{\lambda\ge 0}\inf_{Q_{X,W}}\lambda(\delta-\mathbb{E}_{Q_{X,W}}[\ell(X,W)])+nD(Q_{X,W}\|P_{X}Q_{W})-(n-1)D(Q_{W}\|P_{W})\label{mn21}\\
    &=\sup_{\lambda\ge 0}\lambda\delta+\inf_{Q_{X,W}}\mathbb{E}\left[\ln\frac{Q_{W}\otimes Q^{\otimes n}_{X|W}}{P_{W}\otimes P^{\otimes n}_{X}\exp(\lambda \ell(X,W))}\right]\nonumber\\
    &=\frac{1}{n}\sup_{\lambda\ge 0}\lambda\delta+\inf_{Q_{W},Q_{X|W}}\mathbb{E}_{Q_{W}}\left[\mathbb{E}_{Q_{X|W}}\left[\ln\frac{Q_{W}}{P_{W}}+n\ln\frac{Q_{X|W}}{P_{X}\exp(\frac{\lambda}{n} \ell(X,W))}\Bigg{|}W\right]\right]\nonumber\\
    &=\frac{1}{n}\sup_{\lambda\ge 0}\lambda\delta+\inf_{Q_{W}}\mathbb{E}_{Q_{W}}\left[\ln\frac{Q_{W}}{P_{W}}+n\ln\frac{1}{\mathbb{E}[\exp(\frac{\lambda}{n} \ell(X,W))|W]}\right]\label{mn22}\\
    \label{mn23}
    &=\frac{1}{n}\sup_{\lambda\ge 0}\lambda\delta+\ln\frac{1}{\mathbb{E}_{P_{W}}\left[\mathbb{E}^{n}[\exp(\frac{\lambda}{n} \ell(X,W))|W]\right]}\\
    &=\sup_{\tilde{\lambda}\ge 0}\tilde{\lambda}\delta-\frac{1}{n}\ln\mathbb{E}_{{W}}\left[\left(\mathbb{E}[e^{\tilde{\lambda}\ell(X,W)}|W]\right)^{n}\right].\label{E_SA}
\end{align}
Note that 
\begin{align*}
    &nD(Q_{X,W}\|P_{X}Q_{W})-(n-1)D(Q_{W}\|P_{W})=D(Q_{X,W}\|P_{X}Q_{W})+(n-1)D(Q_{X|W}P_{W}\|P_{X}Q_{W})
\end{align*}
 is a convex function of $Q_{X,W}$. Equation \eqref{mn21} follows from Sion's minimax theorem \cite[Theorem 1]{hartung1982extension} as the objective function in \eqref{mn21} is a (concave) linear function of $\lambda$ and a convex function of $Q_{X,W}$. $\{Q_{X,W}:\sum_{w.x}Q_{X,W}(x,w)=1\}$ is also a compact and convex set (compactness comes from finite alphabet sets of $X$ and $W$) and $\{0\le\lambda\}$ is a convex set. Equations \eqref{mn22} and \eqref{mn23} come from solving the optimization problems over $Q_{X|W}$ and $Q_{W}$ respectively and substituting the optimizers $Q_{X|W}=P_{X}\frac{e^{\frac{\lambda}{n}\ell(X,W)}}{\mathbb{E}_{P_{X}}[e^{\frac{\lambda}{n}\ell(X,W)}]}$ and $Q_{W}=P_{W}\frac{\mathbb{E}^{n}[\exp(\frac{\lambda}{n} \ell(X,W))|W]}{\mathbb{E}_{P_{W}}\left[\mathbb{E}^{n}[\exp(\frac{\lambda}{n} \ell(X,W))|W]\right]}$.





\subsection{Proof of Theorem \ref{thm2nsa}}
\label{proof_of_thm2nsa}
First, note that given any random variable $Z\in\mathcal{Z}$ and any  $A\subseteq \mathcal{Z}$ where $\mathbb{P}_{Z}(A)>0$, we have
\[
\mathbb{P}_{Z}(A)\mf\left(\frac{1}{\mathbb{P}_{Z}(A)}\right)+(1-\mathbb{P}_{Z}(A))\mf(0)=D_{\mf}(P_{Z|\{Z\in A\}}\|P_{Z})
\]
where $\frac{dP_{Z|\{Z\in A\}}}{dP_{Z}}=\frac{\boldsymbol{1}\{Z\in A\}}{P_{Z}(A)}$.
Let
\[
A\triangleq\{(x_{1},\ldots,x_{n},w)\in\mathcal{X}^{\otimes n}\times \mathcal{W}:\,\frac{1}{n}\sum_{i=1}^{n}\ell(x_{i},w)\ge \delta\}.
\]
Then,
\begin{align}
\gamma\mf\left(\frac{1}{\gamma}\right)+(1-\gamma)\mf(0)=D_{\mf}(P_{(S,W)|\{(S,W)\in A\}}\|P_{S,W}).\label{ldkjn1}
\end{align}
Note that under $\tilde{P}_{S,W}=P_{(S,W)|\{(S,W)\in A\}}$, we have
\[\mathbb{E}\left[\frac{1}{n}\sum_{i=1}^{n}\ell(X_{i},W)\right]\ge \delta.
\]
Thus, 
\begin{align}
    D_{\mf}(P_{(S,W)|\{(S,W)\in A\}}\|P_{S,W})\ge \inf_{\substack{R_{S,W}:\\
    \mathbb{E}_R\left[\frac{1}{n}\sum_{i=1}^{n}\ell(X_{i},W)\right]\ge \delta
    }}D_{\mf}(R_{S,W}\|P_{S,W}).\label{ldkjn2}
\end{align}
Take some arbitrary $R_{S,W}$ and let \[\bar{R}_{X,W}(x,w):=\frac1n\sum_iR_{X_i,W}(x,w).\]
Note that $\bar{R}_W=R_W$, and $\mathbb{E}_R\left[\frac{1}{n}\sum_{i=1}^{n}\ell(X_{i},W)\right]=\mathbb{E}_{\bar R}\left[\ell(X,W)\right]$.
From Corollary \ref{corr:def-sub-mod22}  we have
\begin{align}
D_{\mf}(R_{S,W}\|P_{S,W})&\geq -(n-1)D_{\mf}(R_W\|P_W)+\sum_{i}D_{\mf}(R_{X_i,W}\|P_{X_i,W})
\\&\geq 
-(n-1)D_{\mf}(\bar R_W\|P_W)+nD_{\mf}(\bar R_{X,W}\|P_{X,W})
\end{align}
 where in the last step we used Jensen's inequality and the convexity of $D_\mf$ (See Theorem \ref{theorem_properties_D_f}). Thus, corresponding to every arbitrary $R_{S,W}$ there exists some $\bar{R}_{X,W}$ such that 
\begin{align}
D_{\mf}(R_{S,W}\|P_{S,W})&\geq 
-(n-1)D_{\mf}(\bar R_W\|P_W)+nD_{\mf}(\bar R_{X,W}\|P_{X,W})
\end{align}
and $\mathbb{E}_R\left[\frac{1}{n}\sum_{i=1}^{n}\ell(X_{i},W)\right]=\mathbb{E}_{\bar R}\left[\ell(X,W)\right]$.
 Since $\bar{R}_{X,W}$ with these properties can be found for any arbitrary $R_{S,W}$, we obtain the following inequality:
\begin{align}
     \inf_{\substack{R_{S,W}:\\
    \mathbb{E}_R\left[\frac{1}{n}\sum_{i=1}^{n}\ell(X_{i},W)\right]\ge \delta
    }}D_{\mf}(R_{S,W}\|P_{S,W})\geq
     \inf_{\substack{\bar R_{X,W}:\\
    \mathbb{E}_{\bar R}\left[\ell(X,W)\right]\ge \delta
    }}\left\{-(n-1)D(\bar R_W\|P_W)+nD_{\mf}(\bar R_{X,W}\|P_{X,W})\right\}.
    \label{ldkjn3}
\end{align}
From \eqref{ldkjn1}-\eqref{ldkjn3}, we get the desired result. 


\subsection{Proof of Theorem \ref{new properties_I_f}}\label{proof_new properties_I_f}

Proof of (i):
\begin{align*}
    I^{PV}_{\mathsf f}(X;Y)= \min_{Q_Y}\mathbb{E}_{P_{X}Q_{Y}}\mathsf f\left(\frac{dP_{Y|X}}{dQ_{Y}}\right)
\end{align*}
is a minimum over some linear functions of $P_{X}$ and hence is a concave function. 

Next, consider $I^{MBGYA}_{\mathsf f}(X;Y)$.
Note that 
\[
\mathbb{E}_{P_{X}Q_{Y}}\left[\mathsf f\left(\frac{dP_{Y|X}}{dQ_{Y}}(Y,X)\right)\right]
\]
is a linear function of $P_{X}$. The expression
\[
\mathbb{E}_{Q_{Y}}\left[\mf\left(\frac{dP_{Y}}{dQ_{Y}}(Y)\right)\right]
\]
is convex in $P_{Y}$. It is also convex in $P_{X}$ since $P_{Y}$ is a linear function of $P_{X}$ (as $P_{Y|X}$ is fixed). So
\begin{align}
    I^{\text{MBGYA}}_{\mathsf f}(X;Y)= \min_{Q_Y}\mathbb{E}_{P_{X}Q_{Y}}\left[\mathsf f\left(\frac{dP_{Y|X}}{dQ_{Y}}(Y,X)\right)\right]-\mathbb{E}_{Q_{Y}}\left[\mf\left(\frac{dP_{Y}}{dQ_{Y}}(Y)\right)\right].
\end{align}
is a minimization over a sum of a linear function of $P_{X}$ and a concave function of $P_{X}$. Therefore, it is concave in $P_X$.

\vspace{0.5cm}

Proof of (ii): for simplicity of exposition, we prove the statement for discrete variables. Take channels $p_i(y|x)$ and non-negative weights $\lambda_i$ adding up to one. Let $\bar p(y|x)=\sum_i\lambda_i p_i(y|x)$ be the average channel. 

We begin with the CKZ definition of mutual $\mathsf f$-information. Define the perspective function of $\mathsf{f}$ as \begin{align}g(t,z)= t\mathsf{f}(z/t)\label{def-pfsc}, \qquad t>0, z\geq 0.\end{align}  It is known that $\mf$ is a convex function, $g$ is a jointly convex function of $(t,z)$; the Hessian equals
\begin{align*}
    \begin{pmatrix}
    &\frac{1}{t}\mathsf f^{''}\left(\frac{z}{t}\right)&-\frac{z}{t}\mathsf f^{''}\left(\frac{z}{t}\right)\\
    &-\frac{z}{t}\mathsf f^{''}\left(\frac{z}{t}\right)&\frac{z^2}{t^3}\mathsf f^{''}\left(\frac{z}{t}\right)
    \end{pmatrix}\succeq0.
\end{align*}
Applying Jensen's inequality, we obtain $g(t,z)=t\mathsf{f}(z/t)$ \begin{align*}\sum_i\lambda_i\sum_{x,y}p(x)p_i(y)\mf\left(\frac{p_i(y|x)}{p_i(y)}\right)
&\geq 
\sum_{x,y}p(x)\left\{\sum_i\lambda_ip_i(y)\right\}\mf\left(\frac{\sum_i\lambda_ip_i(y|x)}{\sum_i\lambda_ip_i(y)}\right)
\\&=
\sum_{x,y}p(x)\bar p(y)\mf\left(\frac{\bar p(y|x)}{\bar p(y)}\right).
\end{align*}
This yields the desired result. 

Next, consider the PV definition of mutual $\mathsf f$-information:
$$I^{PV}_{\mathsf f}(X;Y)= \min_{Q_Y}D_{\mathsf f}(P_{XY}\|P_{X}\times Q_{Y}).$$
Note that $D_{\mathsf f}(P_{XY}\|P_{X}\times Q_{Y})$ is a jointly convex function of $(q(y),p(y|x))$ for a fixed $p(x)$. This follows from the joint convexity property of $D_\mf$.
From elementary convex analysis,  $\min_{a\in\mathcal{A}}\phi(a,b)$ is convex in $b$ when $\phi(a,b)$ is jointly convex function of $(a,b)$ and $\mathcal{A}$ is a convex set. Therefore,
\[
I^{PV}_{\mathsf f}(X;Y)=\min_{Q_{Y}}D_{\mathsf f}(P_{XY}\|P_{X}\times Q_{Y})
\]
is jointly convex in $p(y|x)$.

Finally consider the MBGYA notion of mutual $\mf$-information. Let
\begin{align}
    J(p_{Y|X}(y)):=\sum_{x\in\mathcal{X}}p_{X}(x)\mf\left(p_{Y|X}(y,x)\right)-\mf\left(\sum_{x\in\mathcal{X}}p_{Y|X}(y,x)p_{X}(x)\right).
\end{align}
This function is jointly convex in $p_{Y|X}$. This follows from the following fact about functions in $\mathcal{F}$: for non-negative weights $\lambda_i$ adding up to one, from \cite[Exercise 14.2]{boucheron2013concentration} we have that  the function
\[(z_1, \cdots, z_n)\mapsto 
\sum_{i=1}^{n}\lambda_{i}\mf(z_{i})-\mf\left(\sum_{i=1}^{n}\lambda_{i}z_{i}\right)
\]
is jointly convex. Using the convexity property of the perspective function (defined above in \eqref{def-pfsc}), we obtain that
\[
\sum_{y\in\mathcal{Y}}q_{Y}(y)J\left(\frac{p_{Y|X}(y)}{q_{Y}(y)}\right)
\]
is a jointly convex function of $q_{Y}$ and $p_{Y|X}$. Then
\[
I_{\mf}^{\text{MBGYA}}(X;Y)=\min_{q_{Y}}\sum_{y\in\mathcal{Y}}q_{Y}(y)J\left(\frac{p_{Y|X}(y)}{q_{Y}(y)}\right)
\]
is convex function in $p_{Y|X}$. Thus, the claim is established.

\vspace{0.5cm}

Proof of (iii): by induction on $n$, it suffices to show that for any independent $X_1$ and $X_2$, we have
\[I_{\mathsf f}(X_1,X_2;U)\geq  I_{\mathsf f}(X_1;U)+I_{\mathsf f}(X_2;U).\]
To show the above inequality when $\mathsf f$ belongs to $\mathcal F$ we proceed as follows: for the MBGYA notion of mutual $\mf$-information, it suffices to show that for any arbitrary $Q_U$, we have
\begin{align}
D_\mf(P_{X_1,X_2,U}\|P_{X_1,X_2}Q_{U})-D_\mf(P_{U}\|Q_{U})
&\geq\sum_{i=1}^2\left\{D_\mf(P_{X_i,U}\|P_{X_i}Q_{U})-D_\mf(P_{U}\|Q_{U})\right\}.\label{eqnsd34}
\end{align}
This follows from independence of $X_1$ and $X_2$ and the super-modularity of $D_\mf$ as defined in \eqref{def-sub-mod}. 

For the CKZ  definition of $\mathsf f$-information, it suffices to consider \eqref{eqnsd34} for $Q_U=P_U$.

\vspace{0.5cm}

Proof of (iv): the inequality $ I_\mf^{CKZ}(X;Y)\geq I_\mf^{PV}(X;Y)\geq I_\mf^{MGBYA}(X;Y)$ follows from the definitions of mutual $\mf$-information.  It remains to show that \[I(X;Y)\cdot \left[\lim_{t\to \infty}t\mathsf f''(t)\right]\geq
I_\mf^{CKZ}(X;Y).\]
 If $\lim_{t\to \infty}t\mathsf f''(t)=\infty$, there is nothing to prove. Assume that $k\triangleq \lim_{t\to \infty}t\mathsf f''(t)<\infty$. From Lemma \ref{newlmma3}, $t\mathsf f^{''}(t)$ is non-decreasing. Thus,
 \begin{align}\label{2050}
    \forall t\in[0,\infty), \,\,\mf^{''}(t)\le \frac{k}{t}.
\end{align}
We claim that
\begin{align}
    \mf(z)-\mf^{'}(1)(z-1)\le k(z\ln(z)-z+1), \qquad \forall z.\label{eqncondforaz}
\end{align}
Take some $1\le z<\infty$. 
Integrating both hands sides of \eqref{2050} twice, each time from $1$ to $z$ yields \eqref{eqncondforaz} for all $z\geq 1$. The proof of \eqref{eqncondforaz} for  $0<z<1$ is similar. Integrating both hands sides of \eqref{2050} from $z$ to $1$ yields
\[
\mf^{'}(1)-\mf^{'}(z)\le -k\ln(z).
\]
Integrating again from $z$ to $1$ yields \eqref{eqncondforaz} for all $z\leq 1$. 
Finally, observe that the upper bound on $\mf(\cdot)$ in \eqref{eqncondforaz} immediately implies that
\[I_\mf^{CKZ}(X;Y)\leq k\cdot I(X;Y).\]
\subsection{Proof of Theorem \ref{property_f_entropy}}\label{proof_property_f_entropy}
Proof of (i): 
\begin{align*}
    &H^{\text{MBGYA}}_{\mathsf f}(X)=\\
&\min_{Q_{X}}\left\lbrace \mathsf f(0)\left(1-\sum_{x\in\mathcal{X}}P_{X}(x)Q_{X}(x)\right)+\sum_{x\in\mathcal{X}}P_{X}(x)Q_{X}(x)\mathsf f\left(\frac{1}{Q_{X}(x)}\right)-\sum_{x\in\mathcal{X}}Q_{X}(x)\mathsf f\left(\frac{P_{X}(x)}{Q_{X}(x)}\right)\right\rbrace
\end{align*}
is a minimum over a concave function of $P_{X}$. Therefore, it is concave in $P_{X}$. 

Next, consider $H^{CKZ}_{\mathsf f}(X)$ under the assumption that $\mathsf f^{'''}(x)\le 0$. Assume that $\mathcal{X}=\{x_1, x_2, \cdots, x_n\}$ and $p_i=\mathbb{P}[X=x_i]$ for $1\leq i\leq n$. We have
\begin{align}
    H^{CKZ}_{\mathsf f}(X)&=\sum_{i=1}^{n}p^{2}_{i}\mathsf f\left(\frac{1}{p_{i}}\right)+\left(1-\sum_{i=1}^{n}p_{i}^{2}\right)\mathsf f(0)
    \\&=\sum_{i=1}^n L_{\mathsf{f}}(p_{i})
\end{align}
where 
\[
L_{\mathsf{f}}(x)\triangleq x^2\mathsf{f}\left(\frac{1}{x}\right)+\left(\frac{1}{n}-x^2\right)\mathsf{f}(0).
\]
To show concavity of $H^{CKZ}_{\mathsf f}(X)$, it suffices to show concavity of $L_{\mathsf{f}}(x)$ for $x\in[0,1]$.  The second derivative of $L_{\mathsf{f}}$ equals:
\begin{align}
  2\mathsf f\left(\frac{1}{x}\right)-\frac{2}{x}\mathsf f^{'}\left(\frac{1}{x}\right)+\frac{1}{x^2}\mathsf f^{''}\left(\frac{1}{x}\right)-2\mf(0)
\end{align}
Note that for all $b\le a$
\[
\mathsf f(a)+(b-a)\mathsf f^{'}(a)+\frac{(b-a)^{2}}{2}\mathsf f^{''}(a)\le \mathsf f(b).
\]
This holds because the residual term in the second order Taylor expansion of $\mathsf f(b)$ at the point $a$ for $0\le b\le a$ is 
\[
R_{2}(b)=\int_{a}^{b}\frac{\mathsf f^{(3)}(t)}{2!}(b-t)^{2}dt
\]
which is non-negative due to $\mathsf f^{(3)}(a)\le 0$ for all $a>0$.
 Letting $a=1/x$ and $b=0$, we obtain
\[
2\mathsf f\left(\frac{1}{x}\right)-\frac{2}{x}\mathsf f^{'}\left(\frac{1}{x}\right)+\frac{1}{x^2}\mathsf f^{''}\left(\frac{1}{x}\right)-2\mathsf f(0)\le 0
\]
as desired. 

Finally, one of the equivalent forms of  $\mathsf{f}\in \mathcal{F}$ presented in \cite[Proposition 11]{beigi2018phi} is that $B_{\mathsf f}:(x,y)\to \mathsf f(x)-\mathsf f(y)-\mathsf f^{'}(x)(x-y)$ is a jointly convex function. Therefore, it is also a separately convex function. If $B_{\mathsf f}$ is seperately convex on $(0,\infty)\times (0,\infty)$, then $\mathsf f^{'''}(x)\le 0$ for $x\in (0,\infty)$  \cite[Exercise 2.3.29]{borwein2010convex}. This completes the proof of part (i).

\vspace{0.5cm}
Proof of (ii): Since the $\mf$-entropy is concave in $p(x)$ and symmetric, it must be maximized at the uniform distribution.

\vspace{0.5cm}
Proof of (iii): The inequality $I^{CKZ}_{\mathsf f}(X;Y)\leq H^{CKZ}_\mf(Y)$ is proved in \cite[Lemma 5]{jiao2017dependence}. The proof for $I^{PV}_{\mathsf f}(X;Y)\leq H^{PV}_\mf(Y)$ and $I^{MBGYA}_{\mathsf f}(X;Y)\leq H^{MBGYA}_\mf(Y)$ are similar. Note that
\begin{align}
    0\le \frac{P_{XY}(x,y)}{P_{X}(x)Q_{Y}(y)}\le \frac{1}{Q_{Y}(y)}
\end{align}
Therefore, Jensen's inequality yields
\begin{align*}
    \mathsf f\left(\frac{P_{XY}(x,y)}{P_{X}(x)Q_{Y}(y)}\right)\le \left(1-\frac{P_{XY}(x,y)}{P_{X}(x)}\right)\mathsf f(0)+\frac{P_{XY}(x,y)}{P_{X}(x)}\mathsf f\left(\frac{1}{Q_{Y}(y)}\right).
\end{align*}
Now averaging over $x,y$ under $P_{X}\times Q_{Y}$, we get
\begin{align*}
    D_{\mathsf f}(P_{XY}\|P_{X}\times Q_{Y})
    &\le \mathsf f(0)\left(1-\sum_{y\in\mathcal{Y}}P_{Y}(y)Q_{Y}(y)\right)+\sum_{y\in\mathcal{Y}}P_{Y}(y)Q_{Y}(y)\mathsf f\left(\frac{1}{Q_{Y}(y)}\right),
\end{align*} 
and
\begin{align*}
    &D_{\mathsf f}(P_{XY}\|P_{X}\times Q_{Y})-D_{\mf}(P_{Y}\|Q_{Y})\\
    &\le \mathsf f(0)\left(1-\sum_{y\in\mathcal{Y}}P_{Y}(y)Q_{Y}(y)\right)+\sum_{y\in\mathcal{Y}}P_{Y}(y)Q_{Y}(y)\mathsf f\left(\frac{1}{Q_{Y}(y)}\right)-\sum_{y\in\mathcal{Y}}Q_{Y}(y)\mathsf f\left(\frac{P_{Y}(y)}{Q_{Y}(y)}\right).
\end{align*}
Taking the minimum of both sides over $Q_{Y}$, we obtain $I^{PV}_{\mathsf f}(X;Y)\leq H^{PV}_\mf(Y)$ and $I^{MBGYA}_{\mathsf f}(X;Y)\leq H^{MBGYA}_\mf(Y)$. Inequality $I^{CKZ}_{\mathsf f}(X;Y)\leq H^{CKZ}_\mf(Y)$ follows by setting $Q_{Y}\equiv P_{Y}$. It suffices to consider $X$ being a discrete random variable because of  the generalization of the Gel'fand-Yaglom-Peres theorem for $\mathsf f$-divergence in \cite[Proposition 1]{gilardoni2009gel}.

\subsection{Proof of Claim \ref{claim1}}\label{proof of claim1}
 To show this, we first prove the claim for $n=1$. Note that for $0\le D\le \frac{1}{2}$ we have
 \begin{align}
     \RR(D)&=1-H_{2}(D)=1+D\log_{2}(D)+(1-D)\log_{2}(1-D)\nonumber\\
     &\overset{(a)}{\le}1+\log_{2}\left(D^{2}+(1-D)^{2}\right)=\log_{2}\left(1+(2D-1)^{2}\right)
     \\&=LB_{\mf}(D,1)
 \end{align}
 where (a) comes from Jensen's inequality for the logarithm function. For $n\ge 2$, we use the Taylor expansion of $1-H_2(D)$ around $D=1/2$:
\begin{align*}
    &1-H_2(D)=\frac{1}{\ln 2}\sum_{k=1}^{\infty}\frac{(2D-1)^{2k}}{2k(2k-1)} 
\\ &=\frac{1}{\ln2}\left[\frac{(2D-1)^{2}}{2}+\frac{(2D-1)^{4}}{12}+\frac{(2D-1)^{6}}{30}+\frac{(2D-1)^{8}}{56}\right]+
\frac{1}{\ln 2}\sum_{k=5}^{\infty}\frac{(2D-1)^{2k}}{2k(2k-1)}.\end{align*}
Next,\begin{align*}
    LB_{\mf}(D,n)&=\frac{1}{n}\log_{2}(n\RR^{CKZ}_{\mf}(D)+1)
    \\
    &=\frac{1}{n}\log_{2}(n(2D-1)^2+1)
    \\&=\frac{1}{\ln2}\sum_{k=1}^{\infty}\frac{n^{k}(2D-1)^{2k}(-1)^{k+1}}{n k}
    \\&=\frac{1}{\ln2}\left[\frac{n(2D-1)^{2}}{n}-\frac{n^{2}(2D-1)^{4}}{2n}+\frac{n^{3}(2D-1)^{6}}{3n}-\frac{n^{4}(2D-1)^{8}}{4n}\right]
    \\&\qquad+
    \frac{1}{\ln2}\sum_{k=5}^{\infty}\frac{n^{k}(2D-1)^{2k}(-1)^{k+1}}{n k}.
\end{align*}
To show $LB_{\mf}(D,n)\geq \RR(D)$, it suffices to show the following two claims:
\begin{align}
    &\frac{n(2D-1)^{2}}{n}-\frac{n^{2}(2D-1)^{4}}{2n}+\frac{n^{3}(2D-1)^{6}}{3n}-\frac{n^{4}(2D-1)^{8}}{4n}\nonumber\\
&-\frac{(2D-1)^{2}}{2}-\frac{(2D-1)^{4}}{12}-\frac{(2D-1)^{6}}{30}-\frac{(2D-1)^{8}}{56}\ge 0\label{meq1d0}
\end{align}
and for any odd $k\geq 5$
\begin{align}
\frac{n^{k}(2D-1)^{2k}}{nk}-\frac{n^{k+1}(2D-1)^{2k+2}}{n(k+1)}-\frac{(2D-1)^{2k}}{2k(2k-1)}-\frac{(2D-1)^{2k+2}}{(2k+2)(2k+1)}\ge 0.\label{meq1d2}
\end{align}

Inequality \eqref{meq1d2} is equivalent with 
\begin{align}\label{im_eq_ex1}
    \frac{n^{k-1}}{k}-\frac{n^{k}(2D-1)^{2}}{k+1}-\frac{1}{2k(2k-1)}-\frac{(2D-1)^{2}}{(2k+2)(2k+1)}\ge 0.
\end{align}
Since $k\ge5$, we have $\frac{1}{2k(k+1)}\ge \frac{1}{2k(2k-1)}$ and $\frac{1}{2k(k+1)}\ge \frac{1}{(2k+2)(2k+1)}$ and it suffices to show that
\[
(k+1)n^{k-1}-kn^{k}(2D-1)^{2}-\frac{1}{2}-\frac{1}{2}(2D-1)^{2}\ge 0
\]
which holds as long as $|D-\frac{1}{2}|\le \frac{1}{2\sqrt{n}}$.

Next, consider \eqref{meq1d0}:
\begin{align*}
    &1-\frac{n(2D-1)^{2}}{2}+\frac{n^{2}(2D-1)^{4}}{3}-\frac{n^{3}(2D-1)^{6}}{4}\\
&-\frac{1}{2}-\frac{(2D-1)^{2}}{12}-\frac{(2D-1)^{4}}{30}-\frac{(2D-1)^{6}}{56}\ge 0.
\end{align*}
Since $|D-\frac{1}{2}|\le \frac{1}{2\sqrt{n}}$, for some constant $c\le1$, we assume that $D=\frac{1}{2}-\frac{c}{2\sqrt{n}}$. Then we get
\begin{align*}
    1-\frac{c^{2}}{2}+\frac{c^{4}}{3}-\frac{c^{6}}{4}-\frac{1}{2}-\frac{c^{2}}{12n}-\frac{c^{4}}{30n^{2}}-\frac{c^{6}}{56n^{3}}\ge 0.
\end{align*}
For $n\ge 2$, it suffices to show that
\begin{align*}
    &\frac{1}{2}-\frac{c^{2}}{2}+\frac{c^{4}}{3}-\frac{c^{6}}{4}-\frac{c^{2}}{24}-\frac{c^{4}}{120}-\frac{c^{6}}{448}\ge \frac{1}{2}-\frac{c^{2}}{2}+\frac{c^{4}}{3}-\frac{c^{4}}{4}-\frac{c^{2}}{24}-\frac{c^{4}}{120}-\frac{c^{4}}{448}\\
    &\ge \frac{1}{2}(1-c^{2})-\frac{c^{2}}{24}+c^{4}\left(\frac{1}{3}-\frac{1}{4}-\frac{1}{120}-\frac{1}{448}\right)\ge \frac{1}{2}(1-c^{2})-\frac{c^{2}}{24}+c^{4}\left(\frac{1}{24}\right)\\
    &\ge (1-c^{2})(\frac{1}{2}-\frac{c^{2}}{24})\ge0
\end{align*}

The last inequality holds for $c\le 1$. The claim is established.


\subsection{Proof of Theorem \ref{thmNewRD}}\label{proof_thmNewRD}

We prove the statement for the CKZ notion of mutual $\mf$-information.  The proof for the MBGYA notion of mutual $\mf$-information is similar. 

Take some arbitrary $P_{\hat X^n|X^n}$ satisfying $I^{CKZ}_{\mathsf f}(X^n;\hat{X}^n)\le nR$. Using Part (iii) of Theorem \ref{new properties_I_f} and Theorem \ref{dpi_f_information}, we have
\begin{align}\label{ineq_im_I_fnn}
    nR\geq I^{CKZ}_{\mathsf f}(X^n;\hat{X}^n)\geq \sum_{i}I^{CKZ}_{\mf}(X_i;\hat{X}^n)\geq \sum_{i}I^{CKZ}_{\mf}(X_i;\hat{X}_i).
\end{align}
We also have
\begin{align}
    \frac1n\sum_{i=1}^n\mathbb{E}[d(X_i,\hat{X}_i)]&\geq 
    \frac1n\sum_{i=1}^n\D^{CKZ}_{{\mathsf{f}}}(I^{CKZ}_{\mf}(X_i;\hat{X}_i))\label{eqnppp1nn}
    \\&
    \geq 
    \D^{CKZ}\left(\frac1n\sum_{i=1}^nI^{CKZ}_{\mf}(X_i;\hat{X}_i)\right)\label{eqnppp2nn}
    \\&\geq \D^{CKZ}\left(R\right)\label{eqnppp3nn}
\end{align}
where \eqref{eqnppp1nn} follows from
the definition of $\D^{CKZ}$, \eqref{eqnppp2nn} follows from convexity of $\D^{CKZ}(\cdot)$ justified below, and \eqref{eqnppp3nn}
follows from \eqref{ineq_im_I_fnn} and the fact that $\D^{CKZ}(\cdot)$ is a decreasing function. Convexity of $\D^{CKZ}(\cdot)$ follows from the fact that Mutual $\mf$-information $I^{CKZ}_{\mf}(A;B)$ is convex in $P_{B|A}$ for a fixed distribution on $P_{A}$ (Theorem \ref{new properties_I_f}, Part (ii)).

Since $P_{\hat X^n|X^n}$ was an arbitrary conditional distribution satisfying $I^{CKZ}_{\mathsf f}(X^n;\hat{X}^n)\le nR$, we deduce from \eqref{eqnppp1nn}-\eqref{eqnppp3nn} the desired inequality.

\subsection{Proof of Theorem \ref{cor_f-dist_lower_bound}}\label{proof_cor_f-dist_lower_bound}

It follows from \eqref{LB_f_RD_F} that
\begin{align}
   \D^{PV}_{\mf}(R)&\geq  -\sup_{Q_{\hat{X}}}\mathbb{F}^{D_{\mf}}_{P_{X}Q_{\hat{X}},R}[-d(X,\hat{X})]
   \\&\overset{(a)}{\ge}-\sup_{Q_{\hat{X}}}\inf_{\lambda\ge0,a\in\mathbb{R}}\left\{\lambda [R+a]+\lambda\mathbb{E}_{P_{X}Q_{\hat{X}}}\mf^{*}\left(-\frac{d(X,\hat{X})}{\lambda}-a\right)\right\}\\
   &\geq-\inf_{\lambda\ge0,a\in\mathbb{R}}\sup_{Q_{\hat{X}}}\left\{\lambda [R+a]+\lambda\mathbb{E}_{P_{X}Q_{\hat{X}}}\mf^{*}\left(-\frac{d(X,\hat{X})}{\lambda}-a\right)\right\}\label{10056a}
   \\
   &= -\inf_{\lambda\ge0,a\in\mathbb{R}}\left\{\lambda [R+a]+\lambda\sup_{\hat{x}}\mathbb{E}_{P_{X}}\mf^{*}\left(-\frac{d(X,\hat{x})}{\lambda}-a\right)\right\}\label{10056}\\
   &\overset{(b)}{=}\sup_{\lambda\ge0,a\in\mathbb{R}}\left\{-\lambda [R+a]-\lambda\phi_{\mf}(-\lambda,a)\right\}
\end{align}
where (a) comes from Part (i) of Theorem \ref{th_lowerbound_Fr_Df_equiv_Orlicz} and (b) follows from the definition of $\phi_{\mf}(\lambda,a)$ in Theorem \ref{cor_f-dist_lower_bound}.

\begin{lemma}
Lower bound in \eqref{lower_general_f_ineq} can be lower bounded by
\begin{align}
   \sup_{\lambda\ge 0,~a\in\mathbb{R}}\left\lbrace-\frac{1}{\lambda}\left[a+R\right]-\frac{1}{\lambda}\phi_{\mf}(-\lambda,a)\right\rbrace\ge -2\max\{1,R\}\cdot \sup_{\hat{x}}\|d(X,\hat{x})\|_{L_{\mf^{*}}}
\end{align}
where $\|\cdot\|_{L_{\mf^{*}}}$ is the 
Orlicz norm as defined in Definition \ref{orlicz_space_def}.
\end{lemma}
\begin{proof}

The lower bound in the part (ii) for $\mf$-rate-distortion function follows from
\begin{align}\label{lbfh_psi_1_norm}
 \D^{PV}_{\mf}(R)&\geq  -\sup_{Q_{\hat{X}}}\mathbb{F}^{D_{\mf}}_{P_{X}Q_{\hat{X}},R}[-d(X,\hat{X})]\geq  -\sup_{Q_{\hat{X}}}\mathbb{F}^{D_{\mf}}_{P_{X}Q_{\hat{X}},R}[|d(X,\hat{X})|]\\
 &\geq -2\sup_{Q_{\hat{X}}}\left[\max\{1,R\}\cdot \|d(X,\hat{X})\|_{L_{\mf^{*}}}\right]\label{lbh_F_Df_1_norm},
\end{align}
where in \eqref{lbh_F_Df_1_norm} we have $(X,\hat X)\sim P_{X}Q_{\hat{X}}$ and use Part (b) of Theorem \ref{th_lowerbound_Fr_Df_equiv_Orlicz}. Next, observe that for any $Q_{\hat{X}}$ we have
\[ \|d(X,\hat{X})\|_{L_{\mf^{*}}}\leq \sup_{\hat{x}}\|d(X,\hat{x})\|_{L_{\mf^{*}}}
\]
which holds because $\mathbb{E}_{P_{X}Q_{\hat{X}}}\left[\mf^{*}\left(\frac{d(X,\hat{X})^{2}}{t^{2}}\right)\right]\le \sup_{\hat{x}}\mathbb{E}_{P_{X}}\left[\mf^{*}\left(\frac{d(X,\hat{x})^{2}}{t^{2}}\right)\right]$ for each $t\ge0$ due to the fact that $\mf^{*}$ is a non-decreasing function.
\end{proof}



\subsection{Proof of Theorem \ref{Th_gen_Uf_If}}
\label{Proof_Th_gen_Uf_If}

Let $\tw=(\tilde{w}_1, \tilde{w}_2, \cdots, \tilde{w}_n)\in\mathcal{W}^n$ be a sequence of length $n$. Let
\begin{align}
\bar{\U}^{I_{\mf}}_{1}(R)&\triangleq \sup_{P_{\tW|S}:~I_{\mf}(S;\tW)\leq R}
\frac1n\sum_{i=1}^n\mathbb{E}\left[L_{\mu}(\tilde{W}_i)-\ell(\tilde{W}_i,Z_i)\right]\label{tildeU}
\end{align}
where $S=(Z_1, Z_2, \cdots, Z_n)$ and
$\tW=(\tilde{W}_1, \tilde{W}_2, \cdots, \tilde{W}_n)$. Observe that if the entries of the vector $\tW$ are all equal, the expression in \eqref{tildeU} reduces to the one  in \eqref{eqnRDf1}. Therefore, in \eqref{tildeU} we are taking the supremum over a larger set. Thus, $\bar{\U}^{I_{\mf}}_{1}(R)\geq \U^{I_{\mf}}_{1}(R)$. 
We claim that
$\bar{\U}^{I_{\mf}}_{1}(R)\leq\U^{I_{\mf}}_{2}(R/n)$. This follows from Theorem \ref{thmNewRD} since $-\bar{\U}^{I_{\mf}}_{1}(R)$ and $-\U^{I_{\mf}}_{2}(R/n)$ can be expressed as rate-distortion functions. The cardinality bound in the statement of the theorem follows from Corollary \ref{newcorcard}.


\subsection{Proof of Example \ref{claim_example_main_theorem}}
\label{proof_claim_example_main_theorem}

Proof of (i):  The convex conjugate of $\mf(x)=(1-x)^{2}$ for $x\in\mathbb{R}$ is $\mf^{*}(y)=y+\frac{y^{2}}{4}$. Then,
\[
D_{\mf}(P\|Q)=\chi^{2}(P\|Q)=\mathbb{E}_{Q}\left[\left(\frac{dP}{dQ}\right)^{2}\right]-1
\]
\[
\mathbb{E}_{P_{Z}}\left[\lambda(\ell(w,Z)-\mathbb{E}\ell(w,Z))-a+\frac{(\lambda\ell(w,Z)-\lambda\mathbb{E}\ell(w,Z)-a)^{2}}{4}\right]\le \frac{\lambda^{2}\sigma^{2}}{4}+\frac{a^{2}}{4}-a
\]
So 
\begin{align*}
    &\Psi^{\mf^{*}}_{\ell(\cdot,Z)}(\chi^{2}(P_{WZ_{i}}\|P_{W}P_{Z_{i}}))=\inf_{\lambda\ge0,a\in \mathbb{R}}\left\{\frac{1}{\lambda}\chi^{2}(P_{WZ_{i}}\|P_{W}P_{Z_{i}})+\frac{a}{\lambda}+\frac{1}{\lambda}\left[\frac{\lambda^{2}\sigma^{2}}{4}+\frac{a^{2}}{4}-a\right]\right\}\\
    &\inf_{\lambda\ge0}\left\{\frac{1}{\lambda}\chi^{2}(P_{WZ_{i}}\|P_{W}P_{Z_{i}})+\frac{\lambda\sigma^{2}}{4}\right\}=\sqrt{\sigma^{2}\chi^{2}(P_{WZ_{i}}\|P_{W}P_{Z_{i}})}.
\end{align*}
Then using Theorem \ref{th27},
\begin{align*}
    \text{\rm{gen}}(\mu,P_{W|S})\le \frac{1}{n}\sum_{i=1}^{n}\sqrt{\sigma^{2}\chi^{2}(P_{WZ_{i}}\|P_{W}P_{Z_{i}})}.
\end{align*}
Then from Jensen's inequality and supermodularity property of $D_{\mf}$ (See Corollary \ref{corr:def-sub-mod22} for constant $W$) for $\mf(x)=x^{2}-x$.
\begin{align}
    \frac{1}{n}\sum_{i=1}^{n}\sqrt{\sigma^{2}\chi^{2}(P_{WZ_{i}}\|P_{W}P_{Z_{i}})}\le \sqrt{\frac{1}{n}\sum_{i=1}^{n}\sigma^{2}\chi^{2}(P_{WZ_{i}}\|P_{W}P_{Z_{i}})}\le \sqrt{\frac{\sigma^{2}}{n}\chi^{2}(P_{WS}\|P_{W}P_{S})}
\end{align}

Proof of (ii): \[
D_{\mf}(P\|Q)=D\left({\bar{\alpha} P+\alpha Q}\|Q\right)
\]
Its convex conjugate function is $\mf^{*}(y)=e^{\frac{y-\bar{\alpha}}{\bar{\alpha}}}-\frac{\alpha}{\bar{\alpha}}y$ for $y\in\mathbb{R}$. We use the theorem \ref{th27} to characterize an upper bound on generalization error.
\[
\mathbb{E}_{{Z}}[e^{\frac{1}{\bar{\alpha}}[\lambda\ell(w,Z)-\lambda\mathbb{E}_{{Z}}\ell(w,Z)]-\frac{a}{\bar{\alpha}}-1}]-\frac{\alpha}{\bar{\alpha}}\mathbb{E}_{{Z}}(\lambda\ell(w,Z)-\lambda\mathbb{E}_{{Z}}\ell(w,Z)-a)\le e^{-\frac{a}{\bar{\alpha}}-1}e^{\frac{\lambda^{2}}{2\bar{\alpha}^{2}}\sigma^{2}}+\frac{\alpha}{\bar{\alpha}}a\quad\forall \lambda,a\in \mathbb{R}.
\]
Then
\begin{align*}
     &\Psi^{\mf^{*}}_{\ell(\cdot,Z)}(I^{CKZ}_{\mf}(W;Z_{i}))\triangleq\inf_{\lambda\ge0, a\in\mathbb{R}}\left\lbrace\frac{1}{\lambda}I^{CKZ}_{\mf}(W;Z_{i})+\frac{a}{\lambda}+\frac{1}{\lambda}\phi(\lambda,a)\right\rbrace\\
     &=\inf_{\lambda\ge0}\left\lbrace\frac{1}{\lambda}I^{CKZ}_{\mf}(W;Z_{i})+\frac{1}{\lambda}\frac{\lambda^{2}\sigma^{2}}{2\bar{\alpha}^{2}}\right\rbrace=\sqrt{\frac{2\sigma^{2}}{\bar{\alpha}^{2}}D\left(\bar{\alpha} P_{WZ_{i}}+\alpha P_{W}P_{Z_{i}}\|P_{W}P_{Z_{i}}\right)}.
\end{align*}

Obviously $\alpha=0$ is deduced to Example 1. We will also look at regime $\alpha\to 1$.
Define $Q_{WZ_{i}}^{\alpha}\triangleq \bar{\alpha} P_{WZ_{i}}+\alpha P_{W}P_{Z_{i}}$ and then $Q_{WZ_{i}}^{1}=P_{W}P_{Z_{i}}$, we get
\begin{align}
    \text{\rm{gen}}(\mu,P_{W|S})&\le \lim_{ \alpha\to 1}\left\{\frac{1}{n}\sum_{i=1}^{n}\sqrt{\frac{2\sigma^{2}}{\bar{\alpha}^{2}}D\left(Q_{WZ_{i}}^{\alpha}\|Q_{WZ_{i}}^{1}\right)}\right\}=\frac{1}{n}\sum_{i=1}^{n}\sqrt{2\sigma^{2}\lim_{ \alpha\to 1}\frac{D\left(Q_{WZ_{i}}^{\alpha}\|Q_{WZ_{i}}^{1}\right)}{(1-\alpha)^{2}}}\nonumber\\
    &\overset{(a)}{=}\frac{1}{n}\sum_{i=1}^{n}\sqrt{2\sigma^{2}\chi^{2}(P_{WZ_{i}}\|P_{W}P_{Z_{i}})}
\end{align}
(a) comes from the fact that the fisher information $J_{\alpha}$ defined on $Q_{WZ_{i}}^{\alpha}$ in $\alpha=1$ is equal to $\lim_{ \alpha\to 1}\frac{D\left(Q_{WZ_{i}}^{\alpha}\|Q_{WZ_{i}}^{1}\right)}{(1-\alpha)^{2}}$. So $$J_{1}=\mathbb{E}_{Q_{WZ_{i}}^{1}}\left[\left(\frac{\frac{\partial dQ_{WZ_{i}}^{\alpha}}{\partial \alpha}}{dQ_{WZ_{i}}^{1}}\right)^{2}\right]=\mathbb{E}_{P_{W}P_{Z_{i}}}\left[\left(1-\frac{dP_{WZ_{i}}}{dP_{W}dP_{Z_{i}}}\right)^{2}\right]=\chi^{2}(P_{WZ_{i}}\|P_{W}P_{Z_{i}})$$.

\appendices


\section{Improving the rate-distortion upper bound for mean estimation}\label{rate_of_consistency}

With the knowledge $\mathsf{Var}(W)\le \kappa(n)$ and $I(Z_i;\A(S))\leq R_i$, we can write the following upper bound on the generalization error:
\begin{align}
\mathrm{gen}\left(\mu,\A\right)
\leq \frac1n\sum_{i=1}^n{\hat{\U}}_{2}(R_{i}).\label{eqnf2233}
\end{align}
where
\begin{equation}
\hat{\U}_2(R)=\sup_{\substack{P_{\hat W|Z}:~I(\hat W;Z)\leq R,\\ \mathsf{Var}(\hat W)\le \kappa(n)}}
\mathbb{E}\left[L_{\mu}(\hat W)-\ell(\hat W,Z)\right].
\end{equation}
We will now compute the bound in \eqref{eqnf2233} for the example of learning the mean of a random vector $Z\in\mathbb{R}^d$ with $\mathbb{E}[Z]=\beta$ under the loss function $\ell(w,z)=\|w-z\|^{2}$. Random vector $Z$ is assumed to be a continuous random variable with finite differential entropy, but not necessarily Gaussian (we will obtain a general result and specialize it to Gaussian distributions later).
First, note that without loss of generality, we can restrict the maximum to $P_{\hat W|Z}$ for which $\mathbb{E}[\hat W]=\mathbb{E}[Z]=\beta$. The reason is that adding a constant vector $r$ to $\hat W$ does not change the generalization error, variance, or $I(\hat W;Z)$. To see this note that
\[
\mathbb{E}[\|\hat W-r-Z\|^{2}]=\mathbb{E}[\|\hat W-Z\|^{2}]+\|r\|^{2}-2r^{T}(\mathbb{E}[\hat W]-\mathbb{E}[Z])
\]
as $r$ is a constant vector. Terms involving $r$ only depend on the marginal distributions of $\hat W$ and $Z$. Then $\mathbb{E}_{P_{\hat W}P_{Z}}[\ell(\hat W-r,Z)]-\mathbb{E}_{P_{\hat W,Z}}[\ell(\hat W-r,Z)]=\mathbb{E}_{P_{\hat W}P_{Z}}[\ell(\hat W,Z)]-\mathbb{E}_{P_{\hat W,Z}}[\ell(\hat W,Z)]$. Moreover, $I(\hat{W}+r;Z)=I(\hat{W};Z)$ (as the translation does not change the differential entropy).
Thus,
\begin{align}
\hat{\U}_2(R)&=\sup_{\substack{P_{\hat W|Z}:~I(\hat W;Z)\leq R,\\ \mathsf{Var}(\hat W)\le \kappa(n)}}
\mathbb{E}\left[L_{\mu}(\hat W)-\ell(\hat W,Z)\right]\\
&=\sup_{\substack{P_{\hat W|Z}:~I(\hat W;Z)\leq R,\\ \mathsf{Var}(\hat W)\le \kappa(n)
\\ \mathbb{E}[\hat W]=\mathbb{E}[Z]
}}
\mathbb{E}\left[L_{\mu}(\hat W)-\ell(\hat W,Z)\right]\\
&\le \sup_{\substack{P_{\hat W|Z}:\mathsf{Var}(\hat W)\le \kappa(n)\\
\mathbb{E}[\hat W]=\mathbb{E}[Z]}}\mathbb{E}[L_\mu(\hat W)]-\inf_{P_{\hat W|Z}:\,I(\hat W;Z)\le R}\mathbb{E}[\ell(\hat W,Z)]\nonumber
\\
&= \sup_{\substack{P_{\hat W}:\mathsf{Var}(\hat W)\le \kappa(n)\\
\mathbb{E}[\hat W]=\mathbb{E}[Z]}}\mathbb{E}_{P_ZP_{\hat W}}[\|\hat W-Z\|_2^2]-\inf_{P_{\hat W|Z}:\,I(\hat W;Z)\le R}\mathbb{E}[\ell(\hat W,Z)]\nonumber
\\
&=\sup_{\substack{P_{\hat W}:\\ \mathsf{Var}(\hat W)\le \kappa(n)\\ \mathbb{E}[\hat W]=\mathbb{E}[Z]}}\mathbb{E}\|\hat W-\mathbb{E}[\hat{W}]\|^{2}_{2}+\mathbb{E}\|Z-\beta\|^{2}-\inf_{P_{\hat W|Z}:\,I(\hat W;Z)\le R}\mathbb{E}[\ell(\hat W,Z)]\nonumber\\
&=\kappa(n)+\mathsf{Var}(Z)-\inf_{P_{\hat{W}|Z}:\,I(\hat{W};Z)\le R}\mathbb{E}[\ell(\hat{W},Z)]\nonumber\\
    &\overset{(a)}{\leq}\kappa(n)+\mathsf{Var}(Z)-\frac{d}{2e\pi}\exp\left(\frac{2h(Z)-2R}{d}\right)\nonumber
\end{align}
where (a)  follows from Shannon's lower bound: note that Corollary \ref{th_shannon_lower_bound}  implies that
 \begin{align}
     \RR(D)&\ge h(Z)-\ln\left(\left(\frac{2De}{d}\right)^{\frac{d}{2}}\Gamma(1+d/2)V_{d}\right)
\\&=
h(Z)-\frac{d}{2}\ln\left(\frac{2\pi D e}{d}\right)
 \end{align}
where we used the fact that $V_{d}\Gamma(1+d/2)=\pi^{d/2}$. Therefore,
\[D\geq \frac{d}{2\pi e}\exp\left(\frac{2h(Z)-2\RR(D)}{d}\right).\]
To sum this up, for any arbitrary distribution $Z$  we obtain the following upper bound 
\begin{align}
\hat{\U}_2(R)
\leq \kappa(n)+\mathsf{Var}(Z)-\frac{d}{2e\pi}\exp\left(\frac{2h(Z)-2R}{d}\right).
\end{align}
Next, let us specialize this general bound to Gaussian distributions and assume that $Z\sim \mathcal{N}(\beta,\sigma^{2}I_{d})$. Then, observe that $\frac{d}{2\pi e}\exp(2h(Z)/d)= \mathsf{Var}(Z)$ \cite[Theorem 8.6.6]{thomas2006elements} and we deduce the following upper bound
\begin{align}
\hat{\U}_2(R)
\leq \kappa(n)+\mathsf{Var}(Z)\left[1-\exp\left(\frac{-2R}{d}\right)\right].
\end{align}
If for some algorithm, $\kappa(n)=c\frac{\mathsf{Var}(Z)}{n}$ for some constant $c\ge 1$ and $R\le \frac{d}{2}\ln\frac{n}{n-1}$, we get 
\begin{align}
\hat{\U}_2(R)\le \frac{(c+1)\mathsf{Var}(Z)}{n}.\label{mdfeq}
\end{align}
This yields an upper bound of order $1/n$. 
In particular if $\kappa(n)=\mathsf{Var}(Z)/n$, we obtain
\begin{align}
\hat{\U}_2(R)\le \frac{2\mathsf{Var}(Z)}{n}\label{mdfeq2}
\end{align}
which exactly matches the generalization error of the ERM algorithm on the Gaussian mean estimation problem as given in \eqref{eqnkkkk2}.

\section{Some ideas for improving the rate-distortion upper bound}
\label{further-ideas}

If only the knowledge of $I(Z_i;\A(S))\le R_{i}$ for $1\leq i\leq n$ is available to us, we obtain
\begin{align}
\mathrm{gen}\left(\mu,\A\right)=\mathbb{E}[\mathrm{gen}_{\mu}(W,S)]=\frac{1}{n}\sum_{i=1}^{n}\mathbb{E}\left[\mathrm{gen}_{\mu}(W,Z_{i})\right]
\leq \frac1n\sum_{i=1}^n{\U}_{2}(R_{i})\label{eqnOB1}
\end{align}
where $\text{\rm{gen}}_{\mu}(w,z_i)=\mathbb{E}_{Z \sim \mu}\left[\ell(Z,w)\right]-\ell(z_i,w)$ and
\begin{equation}
\mathcal{U}_2(R)\triangleq \sup_{P_{\hat W|Z}:~I(\hat W;Z)\leq R}
\mathbb{E}\left[\mathrm{gen}_{\mu}(\hat{W},Z)\right]\label{eqnRDN221}
\end{equation}
where $Z\in\mathcal{Z}$ is distributed according to $\mu$.

The above bound can be improved by adding extra constraints on the domain of the supermum. Given fixed $n$, $S=(Z_1, Z_2, \cdots, Z_n)\sim\mu^{\otimes n}$ and the hypothesis alphabet set  $\mathcal W$, define the following set:
$$\mathcal{P}=\bigcup_{P_{\hat W|S}}\{(P_{\hat W,Z_1}, P_{\hat W,Z_2}, \cdots, P_{\hat W,Z_n})\}$$
where the supremum is over all algorithms $P_{\hat W|S}$ where $\hat W\in\mathcal{W}$. In other words, $\mathcal{P}$ denotes the set of all possible marginal distributions on $(\hat W,Z_i)$ that is induced by an arbitrary algorithm $P_{\hat W|S}$.

Assume that $R_i=I(Z_i;W)$ of the original algorithm is known. Then, we can write the following bound instead of \eqref{eqnOB1}:
\begin{align}
\mathrm{gen}\left(\mu,\A\right)=\mathbb{E}[\rm{gen}_\mu(W,S)]\leq \sup_{(P_{\hat W,Z_1}, \cdots, P_{\hat W,Z_n})\in \mathcal{P}:~~ I(\hat W;Z_i)\leq R_i~\forall 1\leq i\leq n}\frac1n\sum_{i=1}^n\mathbb{E}\left[\rm{gen}_\mu(\hat W_i,Z)\right].\end{align}
The above bound yields a potential improvement over \eqref{eqnOB1} since
\begin{align*}
\sup_{\substack{(P_{\hat W,Z_1}, \cdots, P_{\hat W,Z_n})\in \mathcal{P}:\\ I(\hat W;Z_i)\leq R_i~\forall 1\leq i\leq n}}\frac1n\sum_{i=1}^n\mathbb{E}\left[\rm{gen}_\mu(\hat W_i,Z)\right]\leq
\sup_{P_{\hat W,Z_i}: I(\hat W;Z_i)\leq R_i,~\forall 1\leq i\leq n}\frac1n\sum_{i=1}^n\mathbb{E}\left[\rm{gen}_\mu(\hat W_i,Z)\right].\end{align*}
Finding a complete characterization of the convex set $\mathcal{P}$ seems difficult. However, we provide two ideas that yield some restrictions on the set $\mathcal{P}$. The first idea works only when $\mathcal{W}$ is a finite set while the second idea applies to an arbitrary set  $\mathcal{W}$.

\textbf{First idea:} 
 For a finite set $\mathcal{W}$, one restriction on the set $\mathcal{P}$ can be found using the CKZ notion of mutual $\mf$-information as follows: take some $\mf\in\mathcal{F}$. Then, we claim that for any arbitrary $P_{\hat W|S}$ we have
 \begin{align}
\mathsf f(0)\left(1-\frac{1}{|\mathcal{W}|}\right)+
\frac{1}{|\mathcal{W}|}\mf\left(|\mathcal{W}|\right)&\geq \sum_{i}I^{CKZ}_\mf(\hat W;Z_i).\label{nclnw2}
\end{align}
 In particular, when $\mf(x)=x\log(x)$, we obtain
  \begin{align}
\log|\mathcal{W}|&\geq \sum_{i}I(\hat W;Z_i).
\end{align}
Observe that  \eqref{nclnw2} implies that $\hat W$ cannot simultaneously have a high dependence on all $Z_i$'s. 
 The proof of \eqref{nclnw2} is as follows: let $K$ be a random variable that has a uniform distribution on $\mathcal{W}$. 
\begin{align}
\mathsf f(0)\left(1-\frac{1}{|\mathcal{W}|}\right)+
\frac{1}{|\mathcal{W}|}\mf\left(|\mathcal{W}|\right)&=
H^{CKZ}_\mf(K)
\\&\geq H^{CKZ}_\mf(\hat W)\label{jf0n}
\\&=I^{CKZ}_\mf(\hat W;\hat W)\label{jf1n}
\\&\geq I^{CKZ}_\mf(\hat W;Z_1,Z_2,\cdots,Z_n)
\\&\geq \sum_{i}I^{CKZ}_\mf(\hat W;Z_i)\label{jf3n} 
\end{align}
where \eqref{jf0n} follows from property (ii) of Theorem \ref{property_f_entropy}, \eqref{jf1n} follows from the definition of $\mf$-entropy and finally \eqref{jf3n} follows from property (iii) of Theorem \ref{new properties_I_f}. As a side observation, note that the derivation in \eqref{jf0n}-\eqref{jf3n} is similar to that in \eqref{jf0}-\eqref{jf5}. 

\textbf{Second idea:}
Another restriction on the set $\mathcal{P}$ can be found as follows: let $\tilde\ell(w,z)$ be an ``auxiliary" loss function; an arbitrary loss function of our choice which can be different from the original loss function $\ell(w,z)$. We show that the average risk of the ERM algorithm on the auxiliary loss function $\tilde{\ell}$ can be used to bound the generalization error of a different algorithm $\A$, which runs on the same training data as the ERM algorithm, but with the original loss function $\ell(w,z)$. Let
$$\mathsf{ERM}(z_1, \cdots, z_n)=\min_{w}\sum_{i=1}^n\frac1n\tilde\ell(w,z_i)$$
be the risk of the ERM algorithm given a training sequence $s=(z_1, z_2, \cdots, z_n)$ according to $\tilde\ell$. Let
\begin{align}v_n=\mathbb{E}_{S\sim (\mu)^{\otimes n}}\mathsf{ERM}(Z_1, \cdots, Z_n)\label{defvn}\end{align}
be the average risk of the ERM algorithm. Let us, for now, assume that $v_n$ is known to us (estimating $v_n$ is discussed in Section \ref{findingvnsec}).

Take an arbitrary algorithm $P_{\hat W|S}$. Then, the risk of this algorithm with respect to $\tilde{\ell}$ is greater than or equal the risk of the ERM algorithm, \emph{i.e.,}
\begin{align}\frac1n\sum_{i=1}^n\mathbb{E}\left[\tilde\ell(\hat W,Z_i)\right]\geq v_n\label{eqnArg11}.\end{align}
The above inequality yields a constraint on the set $\mathcal{P}$ since $\mathbb{E}\left[\tilde\ell(\hat W,Z_i)\right]$ depends only on the marginal distribution $P_{\hat W,Z_i}$, $1\leq i\leq n$. An application is given in the next subsection.

\subsection{An application}
Assume that we only know $I(S;\A(S))$. While $\mathcal{U}_1(R)$ (as defined in \eqref{eqnRD1}) is the sharpest 
possible bound on the generalization error given an upper bound $R$ on $I(S;\A(S))$, the single-letter bound 
$\mathcal{U}_2(R/n)$ in Theorem \ref{Th2_without_mismatch} is not. In fact, the following relaxation is used in  the proof of Theorem \ref{Th2_without_mismatch}: instead of producing one output hypothesis $W$ for the entire sequence $S=(Z_1, Z_2, \cdots, Z_n)$, we produce $n$ output hypothesis $\tilde{W}_1, \tilde{W}_2, \cdots, \tilde{W}_n$. To tighten the gap between  $\mathcal{U}_1(R)$ and $\mathcal{U}_2(R/n)$, we can utilize the idea of  ``auxiliary" loss functions.

Let $\tilde\ell(w,z)$ be an auxiliary loss functions. Then, for our algorithm $W=\A(S)$, the risk with respect to $\tilde{\ell}$ is greater than or equal the risk of the ERM algorithm, \emph{i.e.,}
\begin{align}\mathbb{E}\left[\sum_{i=1}^n\frac1n\tilde\ell(W,Z_i)\right]\geq v_n\label{eqnArg1}.\end{align}
Let $Q$ be a random variable, independent of all previously defined variables, and uniform on the set $\{1,2,\cdots, n\}$. Set $\tilde{Z}=Z_{Q}$. Observe that $\tilde{Z}\sim \mu$ because $Z_i\sim \mu$ for all $i$ and $Q$ is independent of $(Z_1,\cdots, Z_n)$. 
Using this definition for $\tilde{Z}$, the risk of $\A$ with respect to the loss $\tilde\ell$ equals
\begin{align}\mathbb{E}[\tilde\ell(W,\tilde{Z})]=\mathbb{E}\left[\sum_{i=1}^n\frac1n\tilde\ell(W,Z_i)\right]\label{eqnArg2}\end{align}
and the generalization error with respect to the loss $\ell$ can be characterized as
\begin{align}	\mathrm{gen}\left(\mu,\A\right)=\frac1n\sum_{i=1}^n\mathbb{E}\left[L_{\mu}(W)-\ell(W,Z_i)\right]= \mathbb{E}\left[L_{\mu}(W)-\ell(W,\tilde{Z})\right].\label{eqnArg3}\end{align}
Thus, we can write a better bound as follows:

\begin{theorem}\label{thm11}
Let
\begin{equation}
\tilde{\mathcal{U}}_2(R)\triangleq \sup_{P_{\hat W|\tilde{Z}}:~I(\hat W;\tilde{Z})\leq R,~~\mathbb{E}[\tilde\ell(\hat W,\tilde{Z})]\geq v_n}
\mathbb{E}\left[L_{\mu}(\hat W)-\ell(\hat W,\tilde{Z})\right]
\end{equation}
where $v_n$ is defined in \eqref{defvn}.
Then,
\[
\mathcal{U}_{1}(R)\le\tilde{\mathcal{U}}_2(R/n)\le {\mathcal{U}}_2(R/n).
\]
\end{theorem}
Proof of Theorem \ref{thm11} can be found in Section \ref{appenEE}.

\begin{example}
Consider the setting in Example \ref{ex1}. Left picture of Figure \ref{fig:fig2} illustrates this improvement in $\mathcal{U}_2(R/n)$ when  $\tilde{\ell}(w,z)=-\mathbf{1}[w\neq z]$ and $n=10$. Right picture of Figure \ref{fig:fig2} illustrates this improvement in $\mathcal{U}_2(R/n)$ when  $\tilde{\ell}(w,z)=(w-z)^{2}$ and $n=10$.
\end{example}

{\color{black}
A dual form of $\tilde{\mathcal{U}}_{2}(R)$ is given in the following theorem:
\begin{theorem}\label{th_dual_form_D_tilde}
We have the following upper bound on $\tilde{\mathcal{U}}_{2}(R)$:
    \begin{align}
        \tilde{\mathcal{U}}_{2}(R)\le \min_{\substack{\lambda\geq 0,\\\eta\ge0}}\lambda R-\eta v_{n}+\lambda\sup_{\hat{w}}\ln\left(\mathbb{E}_{P_{Z}}\left[e^{\frac{L_{\mu}(\hat{w})-\ell(\hat{w},Z)+\eta\tilde{\ell}(\hat{w},Z)}{\lambda}}\right]\right)
    \end{align}
\end{theorem}
\begin{proof}
\begin{align}
	\tilde{\mathcal{U}}_{2}(R)&=\sup_{\substack{P_{\hat{W}|Z}: I(Z;\hat{W})\le R,\\\mathbb{E}[\tilde{\ell}(\hat{W},Z)]\ge v_{n}}} \mathbb{E}\left[L_{\mu}(\hat{W})-\ell(\hat{W},Z)\right]\nonumber
	\\
	&{=}
	\sup_{P_{\hat{W}|Z}}\min_{\substack{\lambda\geq 0,\\\eta\ge0}}\mathbb{E}\left[L_{\mu}(\hat{W})-\ell(\hat{W},Z)\right]+\lambda[R- D(P_{\hat{W}Z}\|P_{\hat{W}}P_{Z})]+\eta[\mathbb{E}[\tilde{\ell}(\hat{W},Z)]- v_{n}]
	\nonumber\\
	&\le\min_{\substack{\lambda\geq 0,\\\eta\ge0}}\sup_{P_{\hat{W}|Z}}\mathbb{E}\left[L_{\mu}(\hat{W})-\ell(\hat{W},Z)\right]+\lambda[R- D(P_{\hat{W}Z}\|P_{\hat{W}}P_{Z})]+\eta[\mathbb{E}[\tilde{\ell}(\hat{W},Z)]- v_{n}]\nonumber\\
&\overset{(a)}{=}
\min_{\substack{\lambda\geq 0,\\\eta\ge0}}\sup_{P_{\hat{W}|Z}}\max_{Q_{\hat{W}}}\lambda R-\eta v_{n}+\mathbb{E}_{P_{\hat W Z}}\left[L_{\mu}(\hat{W})-\ell(\hat{W},Z)+\eta\tilde{\ell}(\hat{W},Z)\right]-\lambda D(P_{\hat{W}Z}\|Q_{\hat{W}}P_{Z})\label{dual_gen_RDnnn}\\
&\overset{(b)}{=}\min_{\substack{\lambda\geq 0,\\\eta\ge0}}\max_{Q_{\hat{W}}}\lambda R-\eta v_{n}+\lambda\mathbb{E}_{P_{Z}}\ln\left(\mathbb{E}_{Q_{\hat{W}}}\left[e^{\frac{L_{\mu}(\hat{W})-\ell(\hat{W},Z)+\eta\tilde{\ell}(\hat{W},Z)}{\lambda}}\right]\right)\nonumber\\
&\overset{(c)}{\le} \min_{\substack{\lambda\geq 0,\\\eta\ge0}}\max_{Q_{\hat{W}}}\lambda R-\eta v_{n}+\lambda\ln\left(\mathbb{E}_{Q_{\hat{W}}P_{Z}}\left[e^{\frac{L_{\mu}(\hat{W})-\ell(\hat{W},Z)+\eta\tilde{\ell}(\hat{W},Z)}{\lambda}}\right]\right)\nonumber\\
&\le \min_{\substack{\lambda\geq 0,\\\eta\ge0}}\lambda R-\eta v_{n}+\lambda\sup_{\hat{w}}\ln\left(\mathbb{E}_{P_{Z}}\left[e^{\frac{L_{\mu}(\hat{w})-\ell(\hat{w},Z)+\eta\tilde{\ell}(\hat{w},Z)}{\lambda}}\right]\right)
\end{align}
where (a) comes from the fact that 
\[
D(P_{\hat{W}Z}\|Q_{\hat{W}}P_{Z})=D(P_{\hat{W}Z}\|P_{\hat{W}}P_{Z})+D(P_{\hat{W}}\|Q_{\hat{W}});
\]
to obtain (b), note that for every fixed $Q_{\hat W}$, the minimizing $P_{\hat{W}|Z}$ in \eqref{dual_gen_RDnnn} is the Gibbs measure:
\[
dP^{*}_{\hat{W}|Z}(\hat{w}|z):=\frac{dQ_{\hat{W}}(\hat{w})e^{\frac{L_{\mu}(\hat{w})-\ell(\hat{w},z)+\eta \tilde{\ell}(\hat{w},z)}{\lambda}}}{\mathbb{E}_{Q_{\hat{W}}}\left[e^{\frac{L_{\mu}(\hat{W})-\ell(\hat{W},z)+\eta \tilde{\ell}(\hat{W},z)}{\lambda}}\right]}.
\]
Finally, (c) follows from Jensen's inequality for the concave function $\ln(\cdot)$. This completes the proof.
\end{proof}
}

\begin{figure}
	   \centering
	\includegraphics[scale=1,width=1\linewidth]{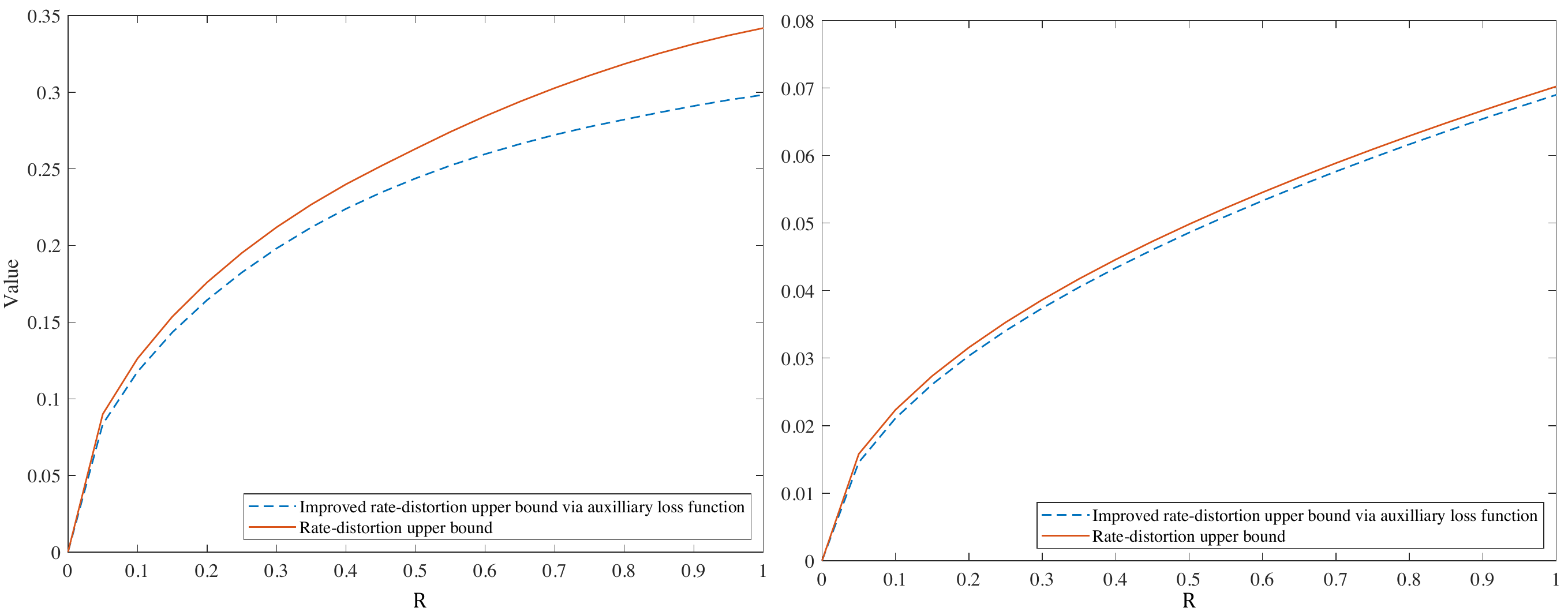}
	\caption{Left picture: The bound in  Theorem \ref{Th2_without_mismatch} and its improved version via the auxiliary loss function $\tilde{\ell}(w,z)=(w-z)^{2}$ for the learning setting $\mathcal{W}=[0,1],\,\mathcal{Z}=\{0,1\}$ and $n=10$ and the original loss function $\ell(w,z)=|w-z|$. Right picture: The bound in  Theorem \ref{Th2_without_mismatch} and its improved version via the auxiliary loss function $\tilde{\ell}(w,z)=-\mathbf{1}[w\neq z]$ for $\mathcal{W}=\mathcal{Z}=\{0,1\}$ and $n=10$ and the original loss function $\ell(w,z)=w\cdot z$.}
	\label{fig:fig2}
\end{figure}
\subsection{Estimating $v_n$}\label{findingvnsec}
In order to use the bound in Theorem \ref{thm11}, one must know the value of $v_n$. However, this is not known in practice. For instance, consider the special case of loss function $\tilde{\ell}(w,z)=(w-z)^{2}$. Given a training data $(z_1, z_2, \cdots, z_n)$, the output of the ERM algorithm with the quadratic loss is just the average of the traning data samples and $v_n$ equals \[\dfrac{n-1}{n}\mathsf{Var}_{\mu}(Z).\]
The variance of the test data is not known, but can be estimated from the training dataset itself. Below we show how to estimate  $v_n$ by running the ERM algorithm on the available training data.  
Assume that the auxiliary loss satisfies $|\tilde{\ell}(w,z)-\tilde{\ell}(w,z')|\leq c$ for all $w,z,z$. Then, we have
\begin{align*}
    \mathsf{ERM}(z_1,z_2,\cdots,z_n)&=\min_{w}\frac1n\sum_{i=1}^n\tilde\ell(w,z_i)
    \\&\leq \min_{w}\left[\frac{c}{n}+
    \frac1n\tilde\ell(w,z'_1)+
    \frac1n\sum_{i=2}^n\tilde\ell(w,z_i)
    \right]
    \\&=\frac{c}{n}+\mathsf{ERM}(z'_1,z_2,\cdots,z_n).
\end{align*}
Then McDiarmid's inequality implies high concentration around expected value for the ERM algorithm:
$$\mathbb{P}\left[\big|\mathsf{ERM}-\mathbb{E}[\mathsf{ERM}]\big|\geq t\right]\leq 2e^{-\frac{2nt^2}{c^2}}.
$$
Thus, one can find an estimate for $v_n$ with high probability based on the available training data sequence.

At the end, we remark that it is also possible to write bounds based on multiple auxiliary loss functions rather than just one.

\subsection{Proof of Theorem \ref{thm11}}\label{appenEE}
It is clear that $\tilde{\mathcal{U}}_{2}(R/n)\le \mathcal{U}_{2}(R/n)$ from their definitions.
By the definition of $v_n$ for any arbitrary $p_{W|S}$ where  $S=(Z_1, Z_2,\cdots, Z_n)$ we have
\[
\mathbb{E}\left[\sum_{i=1}^n\frac1n\tilde\ell(W,Z_i)\right]\geq v_n.
\]
It follows that
\begin{align*}
\mathcal{U}_{1}(R)&=
 \sup_{P_{W|S}:~I(W;S)\leq R}
\mathbb{E}\left[L_{\mu}(W)-L_{S}(W)\right]
\\&
=\sup_{\substack{P_{W|S}:~I(W;S)\leq R,\\ \mathbb{E}\left[\sum_{i=1}^n\frac1n\tilde\ell(W,Z_i)\right]\geq v_n}}
\mathbb{E}\left[L_{\mu}(W)-L_{S}(W)\right].
\end{align*}
For any arbitrary $p_{W|S}$ we have
\[I(W;S)\geq \sum_{i=1}^n I(W;Z_i).\]
Thus,
\begin{align*}
\mathcal{U}_{1}(R)&\leq
\sup_{\substack{P_{W|S}:~\frac1n\sum_{i=1}^n I(W;Z_i)\leq \frac{R}{n},\\ \frac1n\sum_{i=1}^n\mathbb{E}\left[\tilde\ell(W,Z_i)\right]\geq v_n}}
\frac{1}{n}\sum_{i=1}^{n}\mathbb{E}\left[L_{\mu}(W)-\ell(W,Z_{i})\right].
\end{align*}
Take some arbitrary $p_{W|S}$ and a time-sharing random variable $Q$ uniform on $\{1,2,\cdots, n\}$, independent of previously defined variables. Note that
\begin{align}
I(W;Z_Q)&\leq \nonumber
I(Q,W;Z_Q)\\
&=I(W;Z_Q|Q)\label{eqnpf11a}
\\&=\frac1n\sum_{i=1}^n I(W;Z_i)\nonumber
\\&\leq \frac{R}{n}\nonumber
\end{align}
where \eqref{eqnpf11a} follows from the fact that $Z_i$'s are iid. We also have
\[\mathbb{E}\left[\tilde\ell(W,Z_Q)\right]=
\frac1n\sum_{i=1}^n\mathbb{E}\left[\tilde\ell(W,Z_i)\right]\geq v_n,
\]
Thus, the joint distribution $p_{W,Z_Q}$ satisfies the constraints of $\tilde{\mathcal{U}}_2(R/n)$. Moreover, $Z_Q\sim\mu$ and
\[\mathbb{E}\left[L_{\mu}(W)-\ell(W,Z_{Q})\right]=
\frac{1}{n}\mathbb{E}\sum_{i=1}^{n}\left[L_{\mu}(W)-\ell(W,Z_{i})\right]
\]
Thus, we deduce that $\mathcal{U}_1(R)\leq \tilde{\mathcal{U}}_2(R/n)$ as desired.

\section{Orlicz norm and sub-Gaussian random variables}
\label{appendixBreview}
\begin{definition}\label{orlicz_space_def}
 We define the Orlicz space $L_{\psi}$ to be the set of all random variables defined on $\mathcal{X}$ such that
\[
\mathbb{E}[\psi(a|X|)]<\infty
\]
for some $a>0$. The $\psi$ space is a Banach space with respect to the norm
\[
\|X\|_{L_{\psi}}\triangleq\inf\left\lbrace t>0:\,\, \mathbb{E}[\psi(|X|/t)]\le 1\right\rbrace.
\]
\end{definition}
\begin{definition}\label{def2}
 Let $\psi_{2}(x)=e^{x^{2}}-1$. We define the sub-Gaussian random variables to be the set of all random variables defined on $\mathcal{X}$ such that
\[
\mathbb{E}[e^{a^{2}|X|^{2}}]<\infty
\]
for some $a>0$. We also call this set $L_{\psi_{2}}$. The $L_{\psi_{2}}$ space is a Banach space with respect to the norm
\[
\|X\|_{L_{\psi_{2}}}\triangleq\inf\left\lbrace t>0:\,\, \mathbb{E}\left[e^{|X|^{2}/t^{2}}\right]\le 2\right\rbrace.
\]
\end{definition}
\begin{remark}
There is a useful characterization of sub-Gaussian random variables via moment generating function which is equivalent with finiteness of $L_{\psi_{2}}$ norm. The random variable $X$ is said to be sub-Gaussian with parameter $\sigma^2$ if 
	\begin{align}
		\mathbb{E}\left[e^{s(X-\mathbb{E}[X])}\right]\le e^{\frac{\sigma^{2}s^{2}}{2}}, \qquad \forall s\in\mathbb{R}.
	\end{align}
	Using the Chernoff's bound, we obtain,
	\begin{align}
		\mathbb{P}\left(|X-\mathbb{E}[X]|>t\right)\le e^{\frac{-t^{2}}{2\sigma^{2}}}, \qquad \forall t\geq 0.
	\end{align}
\end{remark}
\begin{theorem}\label{th_lowerbound_Fr_Df_equiv_Orlicz}
Let $\mf:[0,\infty)\to \mathbb{R}$ be a convex function satisfying $\mf(1)=0$. Take a random variable $Z$ with a given distribution $\mu$ on the alphabet set $\mathcal Z$. Let 
 $$\mathcal{G}_{\mf^*}:=\left\{g:\mathcal{Z}\mapsto\mathbb{R}: \|g(Z)\|_{L_{\mf^{*}}}<\infty \right\}
$$
where $\mf^*$ is defined in \eqref{eqnfs} and $\|g(Z)\|_{L_{\mf^{*}}}$ is the
Orlicz $L_{\mf^{*}}$ norm on random variable $g(Z)$ which is defined in  Appendix \ref{appendixBreview}. Then,
\begin{enumerate}[(i).]
    \item For $g\in\mathcal{G}_{\mf^{*}}$, a characterization of $\mathbb{F}^{D_{\mf}}_{\mu,R}[g(Z)]$ using the moment generating function of $g(Z)$ is as follows:
    \begin{align}
    \mathbb{F}^{D_{{\mf}}}_{\mu,R}[g(Z)]\le\inf_{\lambda\ge0,a\in\mathbb{R}}\lambda(R+a)+\lambda \mathbb{E}_{Z\sim \mu}\mf^{*}\left(\frac{g(Z)}{\lambda}-a\right)
    \end{align}
    \item The set $\mathcal{G}_{\mf^{*}}$ is a linear space (under ordinary addition and multiplication of functions) and  $\mathbb{F}^{D_{\mf}}_{Q,R}[|g(Z)|]$ is a norm on the space $\mathcal{G}$. Moreover, this norm relates to the Orlicz $L_{\mf^{*}}$ norm on random variable $g(Z)$ as follows:
\begin{align*}
\mathbb{F}^{D_{{\mf}}}_{\mu,R}[|g(Z)|]\le 2\max\{1,R\}\|g(Z)\|_{L_{\mf^{*}}}
\end{align*}
\end{enumerate}
\end{theorem}
\subsection{Proof of Theorem \ref{th_lowerbound_Fr_Df_equiv_Orlicz} }\label{f-distortion_norm_append}
\emph{Proof of the part (a):}
Note that 
\begin{align}
	  \mathbb{F}^{D_{{\mf}}}_{\mu,R}[g(Z)]&=\sup_{\nu\ll \mu:~D_{\mf}(\nu\|\mu)\le R}\mathbb{E}_{Z\sim\nu}[g(Z)]
	  \end{align}
is a constrained supremum over all probability distributions $\nu$ satisfying $D_{\mf}(\nu\|\mu)\le R$. A probability distribution $\nu$ can be thought of as a non-negative measure satisfying
$\mathbb{E}_{Z\sim\mu}\left[\frac{d\nu}{d\mu}(Z)\right]=1$. Observe that
\begin{align}
	  \mathbb{F}^{D_{{\mf}}}_{\mu,R}[g(Z)]&=\sup_{\nu\ll \mu:~D_{\mf}(\nu\|\mu)\le R}\mathbb{E}_{Z\sim\nu}[g(Z)]
	\\
	&=
	\sup_{\nu}\min_{\lambda\geq 0,a\in\mathbb{R}}\mathbb{E}_{Z\sim\nu}[g(Z)]+\lambda (R- D_{\mf}(\nu\|\mu))+a\left(1-\mathbb{E}_{\mu}\left[\frac{d\nu}{d\mu}(Z)\right]\right)\label{pf7a}
\end{align}
where in \eqref{pf7a} we take maximum over all non-negative measures $\nu$ that are not necessarily normalized. Observe that the definition for $D_{\mf}(\nu\|\mu)$ can be extended to unnormalized distributions $\nu$. Next, one can switch maximum and minimum in \eqref{pf7a} because the maximization is a concave optimization problem with a feasible choice of $\nu=\mu$ meeting the constraint $D_{\mf}(\mu\|\mu)\le R$. So, the Slater's condition, a sufficient condition for strong duality is satisfied. Therefore,
    \begin{align}
	  \mathbb{F}^{D_{{\mf}}}_{\mu,R}[g(Z)]
	&=
	\min_{\lambda\geq 0,a\in\mathbb{R}}\sup_{\nu}\left\{\mathbb{E}_{Z\sim\nu}[g(Z)]+\lambda R-\lambda D_{\mf}(\nu\|\mu)+a\left(1-\mathbb{E}_{\mu}\left[\frac{d\nu}{d\mu}(Z)\right]\right)\right\}
	\\
	&=	\min_{\lambda\geq 0,a\in\mathbb{R}}\sup_{\nu}\left\{\lambda R+a+\lambda\int d\mu(z)\left[\frac{d\nu}{d\mu}(z)\cdot\frac{g(z)-a}{\lambda}-\mf\left(\frac{d\nu}{d\mu}(z)\right)\right]\right\}\label{eq100145}\\
		&\leq\min_{\lambda\ge0,a\in\mathbb{R}}\left\{\lambda R+a+\lambda \int d\mu(z)\sup_{t(z)\geq 0}\left[t(z)\cdot\frac{g(z)-a}{\lambda}-\mf\left(t(z)\right)\right]\right\}\label{eq100146}
\\
	&=\min_{\lambda\ge0,a\in\mathbb{R}}\left\{\lambda(R+a)+\lambda \mathbb{E}_{Z\sim\mu}\mf^{*}\left(\frac{g(Z)}{\lambda}-a\right)\right\}
\end{align}
\emph{Proof of the part (b): }
First we show that $\mathbb{F}^{D_{\mf}}_{\nu,R}[|g(Z)|]:\mathcal{G}_{\mf^{*}}\to \mathbb{R}_{+}$ is a norm. Clearly, $\mathbb{F}^{D_{\mf}}_{\nu,R}[|g(Z)|]\ge
\mathbb{E}_{Z\sim\mu}[|g(Z)|]\geq
0$. Moreover, $\mathbb{F}^{D_{\mf}}_{\nu,R}[|g(Z)|]=0$ implies
$\mathbb{E}_{Z\sim\mu}[|g(Z)|]=0$ and hence
 $g(Z){=}0$ almost surely with respect to measure $\mu$. We also have $\mathbb{F}^{D_{\mf}}_{\nu,R}[|c\cdot g(Z)|]=|c|\cdot\mathbb{F}^{D_{\mf}}_{\nu,R}[|g(Z)|]$. The triangle inequality property is immediate from the sub-additivity of the function $|\cdot|$. 

It remains to show the following inequality for a convex function $\mf$ with $\mf(1)=0$:
\begin{align}
  \mathbb{F}^{D_{{\mf}}}_{\mu,R}[|g(Z)|]\le 2\max\{1,R\}\|g(Z)\|_{L_{\mf^{*}}(\mu)}.
\end{align}
We use the equivalence between the Ameniya norm and the Orlicz norm (see \cite{hudzik2000amemiya}) \begin{align}\|g(Z)\|_{L_{\mf^{*}}(\mu)}\le\|g(Z)\|^{A}_{L_{\mf^{*}}(\mu)}\le 2\|g(Z)\|_{L_{\mf^{*}}(\mu)}\label{eqnequao}\end{align} 
where the Ameniya norm defined as follows:
\begin{align}
    \|g(Z)\|^{A}_{L_{f^*}(\mu)}\triangleq \inf_{t>0}\left[\frac{1+\mathbb{E}_{Z\sim\mu}f^*(|tg(Z)|)}{t}\right].
\end{align}

From part (a) of Theorem \ref{th_lowerbound_Fr_Df_equiv_Orlicz} we have
\begin{align}
    \mathbb{F}^{D_{{\mf}}}_{\mu,R}[|g(Z)|]
    &=\inf_{\lambda\ge0,a\in\mathbb{R}}\left\{\lambda(R+a)+\lambda \mathbb{E}_{Z\sim\mu}\mf^{*}\left(\frac{|g(Z)|}{\lambda}-a\right)\right\}
    \\&\le R\cdot \inf_{\lambda\ge0}\left\{\lambda+\lambda \mathbb{E}_{Z\sim\mu}\left(\frac{\mf^{*}}{R}\left(\frac{|g(Z)|}{\lambda}\right)\right)\right\}
    \\&=R\cdot \|g(Z)\|^{A}_{L_{\frac{\mf^{*}}{R}}(\mu)}\label{neq1}
\end{align}
Consider two cases: if $R>1$, then from \eqref{neq1} and \eqref{eqnequao} we obtain
$\mathbb{F}^{D_{{\mf}}}_{\mu,R}[|g(Z)|]\leq 2R\|g(Z)\|_{L_{\frac{\mf^{*}}{R}}(\mu)}$. On the other hand, for $R>1$ we have $\|g(Z)\|_{L_{\frac{\mf^{*}}{R}}(\mu)}\leq \|g(Z)\|_{L_{\mf^{*}}(\mu)}$, so we get
$\mathbb{F}^{D_{{\mf}}}_{\mu,R}[|g(Z)|]\leq 2R\|g(Z)\|_{L_{\mf^{*}}(\mu)}$ as desired. Next, if $R\leq 1$, then from the definition of $\mathbb{F}^{D_{{\mf}}}_{\mu,R}$, \eqref{neq1} for the choice of $R=1$, and \eqref{eqnequao} we have
\[\mathbb{F}^{D_{{\mf}}}_{\mu,R}[|g(Z)|]\leq \mathbb{F}^{D_{{\mf}}}_{\mu,1}[|g(Z)|]\leq \|g(Z)\|^{A}_{L_{\mf^{*}}(\mu)}\leq 2\|g(Z)\|_{L_{\mf^{*}}(\mu)}.\]
This yields the desired inequality in both cases.
\end{document}